\documentclass[a4paper,twocolumn,11pt,unpublished]{quantumarticle}
\pdfoutput=1
\usepackage[T1]{fontenc}
\usepackage[utf8]{inputenc}
\usepackage{amsmath,amssymb,amsthm,mathtools}
\usepackage{mathrsfs}
\usepackage{graphicx,booktabs,microtype,xcolor}
\usepackage{needspace}
\makeatletter
\patchcmd{\section}{\sffamily}{\sffamily\raggedright\hyphenpenalty=10000\exhyphenpenalty=10000}{}{\PackageError{manuscript}{Section format changed}{Check the class heading definition.}}
\patchcmd{\subsection}{\sffamily}{\sffamily\raggedright\hyphenpenalty=10000\exhyphenpenalty=10000}{}{\PackageError{manuscript}{Subsection format changed}{Check the class heading definition.}}
\patchcmd{\subsubsection}{\sffamily}{\sffamily\raggedright\hyphenpenalty=10000\exhyphenpenalty=10000}{}{\PackageError{manuscript}{Subsubsection format changed}{Check the class heading definition.}}
\makeatother
\AddToHook{cmd/section/before}{\Needspace{8\baselineskip}}
\AddToHook{cmd/subsection/before}{\Needspace{5\baselineskip}}
\AddToHook{cmd/subsubsection/before}{\Needspace{4\baselineskip}}
\usepackage[numbers,sort&compress]{natbib}
\usepackage{hyperref}
\DeclareMathOperator{\Tr}{Tr}
\newtheorem{proposition}{Proposition}
\hypersetup{pdftitle={Krylov geometry and optimal recovery under uncertain coherent noise},pdfauthor={Abhishek Chowdhury}}

\begin{document}
\raggedbottom
\title{Krylov geometry and optimal recovery under uncertain coherent noise}
\author{Abhishek Chowdhury}
\email{achowdhury@iitbbs.ac.in}
\affiliation{School of Basic Sciences, Indian Institute of Technology Bhubaneswar, Argul, Odisha 752050, India}
\date{}
\maketitle

\begin{abstract}

Recovery from a coherent rotation with an unobserved angle depends on both the encoding and what calibration reveals about the noise. We determine their combined limit for a known finite-dimensional generator and one final recovery. Two polynomial projections organize the answer. A physical Krylov chain separates correctable error amplitudes from their first state-dependent distortion; projection onto lower angle powers identifies the part that a coherent decoder can cancel. Under explicit rank and moment conditions, their combination gives the exact leading optimal entanglement infidelity at weak noise and a recovery attaining it for any finite logical dimension. The optimal worst-input error lies between one and $d$ times that infidelity at leading order, for logical dimension $d$. When the angle distribution is unknown, we optimize one decoder for every distribution compatible with the calibration. When spectral and calibration frequencies are integer multiples of a common frequency, moment positivity turns this into an exact finite optimization. Small measured means can still allow rare finite rotations. If the sole calibration datum is $\beta_r=\mathbb E[\sin^{2r}(\theta/2)]$, with $\theta$ the wrapped rotation phase at that common frequency, the optimal worst-prior infidelity is $O(\beta_r)$ exactly when generator powers through degree $r-1$ are correctable. When the first uncorrectable power has degree $m<r$, the optimal loss scales as $\beta_r^{m/r}$, with a sharp coefficient fixed by a universal moment problem. In a loss-free four-photon memory, population measurements at two calibration durations retain the fully informed quartic error coefficient for a fixed symmetric narrow-noise family with positive variance and finite fourth moment; a single duration permits loss quadratic in the noise width.
\end{abstract}

\section{Introduction}
\label{sec:introduction}
A quantum memory must retain an unknown superposition, including its entanglement with other systems. A fluctuating field can threaten that information even when its effect in each run is a coherent rotation. If the angle were known, the rotation could be reversed. When it is unavailable, a final recovery must use the redundant physical encoding and the information obtained from calibration. Which features of the encoding decide the best recovery, and how much must calibration determine?

We study a fixed encoding $V$ that preserves inner products and embeds a finite logical register in a finite physical Hilbert space. One known Hermitian generator $H$ acts during storage through $e^{-i\sigma xH}$, where $x$ is drawn independently of the logical input and $\sigma$ controls the angle scale. The realized $x$ remains unannounced. One scalar multiplies the whole generator, which may act on many physical qubits or optical modes. The distribution of $x$, or prior, can be known or constrained by observations. A deterministic recovery acts after storage and can use joint operations and auxiliary systems. The same recovery acts on every unknown logical input and every unannounced angle realization. This setting includes a field that is constant during each run and varies between runs, as well as a time-dependent scalar field with the same integrated angle. Independently fluctuating couplings, noncommuting time-dependent generators and control during storage define additional problems.

The physical obstacle is logical distinguishability. Redundancy helps when different error images carry the same unknown logical coefficients, allowing their error dependence to move into a disposable auxiliary coordinate, called the syndrome. The Knill--Laflamme conditions for quantum error correction (QEC) express this requirement through logical-state-independent error overlaps~\cite{KnillLaflamme1997,KnillLaflammeViola2000}. Approximate and perturbative correction quantify departures from that factorization~\cite{LeungEtAl1997,Beny2011}; constructive recovery and environment-based bounds connect it to attainable fidelity~\cite{BarnumKnill2002,NgMandayam2010,BenyOreshkov2011,ZhengEtAl2024,LiEtAl2025}. Coherent errors also behave differently from stochastic Pauli errors under syndrome extraction and decoding~\cite{BealeEtAl2018,BravyiEtAl2018}. These results frame a quantitative question for coherent noise. Once the first imperfectly correctable amplitude is reached, which part can a final recovery remove, and what error must remain?

Expanding the rotation produces the amplitudes $V,HV,H^2V,\ldots$, so the physical error itself selects a Krylov sequence. Orthogonalizing these maps produces a block Krylov chain, following the entire encoded register through successive error subspaces. Its matrix-valued spectral measure~\cite{DamanikPushnitskiSimon2008} keeps the energy weights as operators on the logical register and retains the full finite-angle orbit. The initially correctable subspaces each carry the same logical amplitudes; at the first harmful boundary, those amplitudes undergo a state-dependent distortion. A separate orthogonalization of angle powers identifies which part can be cancelled by coherently combining earlier error amplitudes. These two constructions determine an attaining recovery and a matching bound on every final recovery. 

Lanczos chains also organize operator growth~\cite{ParkerEtAl2019} and quantum-state spread~\cite{BalasubramanianEtAl2022Spread,NandyEtAl2025}. Here each encoded state has its own state Krylov chain. Under the correctability hypotheses stated below, the initial recurrence coefficients are shared by every logical input. Their first input-dependent values detect the physical obstruction to further correction, whose recovery cost also depends on the angle statistics. Operator Krylov methods for a fixed decoder and a Markovian noise generator address a complementary observable-based problem~\cite{GordeychukKiselev2026}; the encoded amplitudes here retain the freedom to optimize recovery itself.

Our first principal result evaluates this optimal boundary cost. Suppose one recovery corrects generator powers through degree $L$; these powers form a correctable prefix. Put $m=L+1$. Correctability means that the compressed moments $V^\dagger H^jV$, acting on the logical register, are scalar multiples of its identity for $0\le j\le2L$. Assume also that each orthogonal error subspace through degree $L$ retains the full logical dimension. Hold the rescaled prior fixed with finite $2m$-th moment, and require $1,x,\ldots,x^L$ to be linearly independent in mean square under that prior. Under these hypotheses, Proposition~\ref{thm:optimal} gives
\begin{equation}
\begin{gathered}
1-F_{e,\mathrm{opt}}=C_{\mathrm{opt}}\sigma^{2m}+o(\sigma^{2m}),\\
C_{\mathrm{opt}}=\left(\frac{g_L}{m!}\right)^2h_m\mathcal V.
\end{gathered}
\label{eq:intro-leading-coefficient}
\end{equation}
Entanglement fidelity $F_e$ tests preservation of correlations with an untouched reference. The factor $g_L/m!$ is the propagation amplitude to the boundary beyond this correctable prefix. The quantity $\mathcal V$ measures its logical distortion after optimal physical alignment, and $h_m$ is the mean-square part of $x^m$ left after subtracting lower angle powers. One coherent adjustment of the decoder performs that subtraction. A matching bound covers all completely positive trace-preserving recoveries, including recovery of several logical qubits and singular leakage. The same boundary determines the precise leading gap of Petz recovery~\cite{Petz1986,BarnumKnill2002}, a widely used decoder constructed from the channel. 

For a $d$-dimensional logical register with $C_{\mathrm{opt}}>0$, the optimal worst-input error, half the diamond distance from the identity channel, has order $\sigma^{2m}$. Its ratio to the optimal entanglement infidelity has lower and upper limits bounded by $1$ and $d$. A known logical rotation appended to the attaining recovery gives the upper bound (Proposition~\ref{prop:worst-input}). Finite noise alphabets can cancel the leading statistical residual; their improved orders and sensitivity to contamination follow from the corresponding polynomial projection.

Calibration introduces a distinct limitation. Robust and universal recovery already address incomplete channel knowledge~\cite{KosutShabaniLidar2008,TaghaviKosutLidar2010,ElimelechAlbertBarg2026Universal}. Estimation within a specified channel family and partial noise characterization can improve decoding~\cite{Danageozian2023IncompleteRecovery,IyerEtAl2025}. Here we ask for the best guarantee valid for every angle distribution compatible with specified ensemble data. In the weak-noise limit, uncertain moments select one optimal residual polynomial and hence one common decoder. At finite angle, positivity and the measured Fourier moments define the compatible channels exactly when the occupied spectral gaps and calibration frequencies are integer multiples of a common frequency. This construction uses a moment bootstrap in which
 all moments must belong to the same positive probability measure. The shared positivity structure also underlies moment optimization~\cite{Lasserre2001} and quantum-mechanical bootstrap methods~\cite{HanHartnollKruthoff2020}; here the unknown object is an angle distribution, while the latter constrain quantum-state expectations. The resulting finite convex optimization includes a common recovery and a matching adverse prior.

A small calibration mean alone need not imply uniformly small rotations. Rare finite rotations can remain compatible with it; distributional tails already change coherent-error suppression in specified code models~\cite{CladerEtAl2021}. Our second principal result ties this ambiguity to the physical correction order. After removing a scalar energy shift and choosing a common frequency unit $\Lambda>0$, write the generator with integer occupied energies. The wrapped phase is $\theta\equiv\Lambda\sigma x\pmod{2\pi}$. For an integer $r\ge1$, measure $\beta_r=\mathbb E[\sin^{2r}(\theta/2)]$, the mean of a nonnegative observable that vanishes as $\theta^{2r}$ near zero.
For $r=1$, $\beta_1=(1-\mathbb E[\cos\theta])/2$ measures the reduction of the cosine interference contrast at the fundamental frequency. Suppose the sole known statistic is $\beta_r=\epsilon$, and allow every probability distribution of the phase modulo $2\pi$ with that mean. As $\epsilon\downarrow0$, one final recovery has worst-prior infidelity $O(\epsilon)$ precisely when the error span $I,H,\ldots,H^{r-1}$ is exactly correctable. Proposition~\ref{prop:calibration-prefix} also determines the sharp coefficient through a finite optimization. Thus the order justified by the data can be matched directly to a correctable Krylov prefix.

If the first defective generator power is $H^m$ and $r>m$, the optimal worst-prior loss scales as $\beta_r^{m/r}$ under the boundary hypotheses of Proposition~\ref{prop:calibration-mismatch}. Its sharp coefficient follows from a universal moment problem. For the first protected case $m=2$, an explicit three-point adverse prior attains that coefficient. For every $m$, a dual polynomial bound with one nonnegative multiplier determines the constant, which approaches the Chebyshev value for a known support interval monotonically as $r$ increases. This connects the local optimality result directly to a finite-calibration limit. Christoffel functions bound the largest probability that any compatible prior can place on one damaging phase~\cite{KordaPutinarMezic2020}. They 
  explain which finite-angle alternatives the data have excluded.

A four-photon memory turns these general limits into explicit calibration and recovery choices. The code of Ref.~\cite{ChuangLeungYamamoto1997} undergoes a calibrated number-conserving mode rotation in its loss-free sector. Krylov dynamics under time-dependent optical Hamiltonians has been studied in Ref.~\cite{ChowdhuryMahapatra2024}; here the Krylov chain of the stored rotation determines the recovery cost. Population measurements on separate calibration copies constrain the same quasi-static rate distribution as storage. We derive exact optimal recovery guarantees from one, two and three calibration durations, with all unmeasured sine harmonics allowed. Two durations retain the quartic coefficient of the fully informed decoder for a fixed symmetric narrow prior with positive variance and finite fourth moment; one duration leaves a quadratic worst-case loss. Proposition~\ref{prop:sector-data}, applied to the measured optical cross block and completed in Appendix~\ref{app:optical-bootstrap}, constructs one recovery throughout each data class, including mixtures of coherent alignments at singular boundaries. Finite-shot bounds, axis errors and a uniform contamination crossover specify what the ideal guarantee requires. A two-qubit example separately identifies joint coherence information that local channel measurements can miss.

The argument follows the information available to the decoder. We first retain the logical dependence of the energy weights and identify the first amplitude that an exact syndrome separation cannot absorb (Sections~\ref{sec:model} and~\ref{sec:geometry}). Its recovery cost separates into a physical boundary distortion and a statistical polynomial residual. The same construction controls the worst-input error once a known logical rotation is removed. Partial moment or support information changes the polynomial chosen by one common decoder, while finite alphabets can make its leading coefficient vanish (Section~\ref{sec:uncertainty}). A small periodic calibration mean poses a different problem because rare finite rotations can remain compatible with it. The exact moment construction in Section~\ref{sec:finite-angle} retains those alternatives and matches the protection justified by calibration to correctable Krylov depth. Conserved logical sectors then reduce the common recovery to coherence alignment (Section~\ref{sec:sector-applications}). The four-photon memory in Section~\ref{sec:optical-application} evaluates that route from physical shells to measured populations, an attaining decoder and its finite-precision requirements. Table~\ref{tab:model-map} keeps the hypotheses and information classes distinct.

Appendix~\ref{app:attainment-details} completes the attaining recovery; Appendix~\ref{app:petz-details} controls the singular Petz expansion.
 Appendices~\ref{app:alphabet-details} and~\ref{app:contamination-details} control finite-alphabet stability. Appendices~\ref{app:calibration-moments} and~\ref{app:calibration-witnesses} justify finite Fourier compatibility and polynomial fidelity bounds at singular data, while Appendices~\ref{app:calibration-rare-pair} and~\ref{app:calibration-interval} determine the adverse-phase bounds. The optical completion and harmonic frontiers are derived in Appendices~\ref{app:optical-implementation} and~\ref{app:optical-bootstrap}; Appendices~\ref{app:benchmark-details} and~\ref{app:joint-details} develop the finite-width comparisons and uniform crossover. The two-generator counterexample in Appendix~\ref{app:triangle-obstruction} identifies the matrix information needed for a multivariate extension. Appendix~\ref{app:computational-checks} summarizes the numerical checks of the analytical results.

\section{Encoded memory and its spectral geometry}
\label{sec:model}

The information to be stored is a density matrix $\rho$ on a $d$-dimensional logical Hilbert space. For a pure input, its amplitudes in $|\psi\rangle=\sum_a\psi_a|a\rangle$ are unknown to the memory. An encoding $V:\mathbb C^d\to\mathcal H$ embeds this register into a finite physical Hilbert space of dimension at least $d$. It is an isometry, meaning $V^\dagger V=I_d$, so it preserves inner products and hence every superposition. Its image is the code subspace. When the physical space is larger, its extra dimensions are the redundancy available for correction. The logical dimension is arbitrary; $d=2^k$ describes $k$ stored qubits.

The known Hermitian operator $H$ generates the unwanted rotation. For example, $H=\sum_j h_j^{\rm field}Z_j$ describes a fixed spatial profile of couplings to physical spins, where $Z_j$ is the Pauli $Z$ operator at site $j$ and the real $h_j^{\rm field}$ are calibrated. One random scalar multiplying this entire operator represents a fluctuating common field strength. More general finite Hermitian operators, including interactions, are allowed. We set $\hbar=1$, absorb the known coupling and time units into $H$, and write the dimensionless integrated angle as $\sigma x$. The scale $\sigma\ge0$ will control the weak-noise limit; the dimensionless random variable $x$ has probability distribution $\nu$. Averaging over runs defines the encoded noise channel
\begin{equation}
 \mathcal N_{\sigma,\nu}(\rho)
 =\int e^{-i\sigma xH}V\rho V^\dagger e^{i\sigma xH}\,d\nu(x).
 \label{eq:channel}
\end{equation}
 The channel maps a logical input to a physical output. In each run the stored state undergoes a unitary rotation; the average appears because the realized $x$ is unavailable. The distribution $\nu$ is independent of $\rho$. It can describe quasi-static noise, whose amplitude is constant during a run but varies between runs. A time-dependent scalar multiplying the same $H$ also reduces to Eq.~\eqref{eq:channel} through its integrated angle. One uncertain scalar can therefore drive an entire finite register. Independently fluctuating coefficients or noncommuting time-dependent generators generally need additional error operators. Finally, $\sigma$ sets only the scale of the angle; the root-mean-square angle is $\sigma\sqrt{\mathbb E_\nu[x^2]}$ when that moment exists and equals $\sigma$ only for unit second moment.

We use $\mathbb E_\nu[f(x)]=\int f(x)\,d\nu(x)$ for an expectation. The prior is the probability distribution of the per-run angle variable. Its support is the closed set of allowed angle values, defined by positive probability in every neighborhood of each such value. An atom is one value with positive probability, and its mass is that probability; more generally, the mass of a set is the probability assigned to that set. The integral notation in Eq.~\eqref{eq:channel} covers continuous densities, discrete choices and mixtures. Moments are expectations of powers of the angle variable, so the same conventions apply to every noise class below.

An unencoded qubit illustrates the loss caused by this average. Set $V=I_2$ and $H=Z=\operatorname{diag}(1,-1)$. The populations are unchanged, while the off-diagonal matrix element becomes
\begin{equation}
 \rho_{01}\longmapsto\gamma\rho_{01},\qquad
 \gamma=\mathbb E_\nu[e^{-2i\sigma x}].
 \label{eq:unencoded-dephasing}
\end{equation}
 In a run with known $x$, the factor $e^{-2i\sigma x}$ is a known phase and can be undone. Averaging different phases can give $|\gamma|<1$, reducing coherence even when their distribution is known exactly. For example, equally likely $x=\pm1$ gives $\gamma=\cos(2\sigma)$. This reduction of off-diagonal coherence is called dephasing. Redundant encoding can retain the relative phases in extra physical degrees of freedom. The question is when a final operation can recover them from those degrees of freedom.

The storage evolution in Eq.~\eqref{eq:channel} is prescribed and uncontrolled; recovery acts after it. Control during storage changes the problem.  For a quasi-static field, a known unitary $U_{\rm echo}$ satisfying $U_{\rm echo}^\dagger HU_{\rm echo}=-H$ gives the exact echo
\begin{equation}
  U_{\rm echo}^\dagger e^{-i\sigma xH/2} U_{\rm echo} e^{-i\sigma xH/2}=I
 \label{eq:storage-echo}
\end{equation}
  for every $x$. More generally, if the two halves accumulate angles $\tau_1,\tau_2$, this sequence leaves $e^{-i(\tau_1-\tau_2)H}$ and is exact when those angles agree. For $H=\sum_jh_j^{\rm field}Z_j$, a simultaneous Pauli $X$ pulse on every qubit reverses the field, implementing a spin echo~\cite{Hahn1950}. For a spin generator, a $\pi$ rotation about a perpendicular axis has the same effect. Our optimum concerns final recovery of the prescribed channel, without these intermediate pulses.

 A recovery, also called a decoder, is one deterministic operation $\mathcal R$ on the physical output that returns a $d$-dimensional logical state. It can act jointly on the entire memory, use auxiliary systems, and discard them. Mathematically these operations are completely positive trace-preserving (CPTP) maps~\cite{Watrous2018}. Trace preservation excludes postselection; complete positivity ensures a valid operation even on a memory entangled with another system. The map is selected using $V$, $H$, $\sigma$ and the declared noise information, and is then applied to every unknown logical input and every unannounced angle realization.\footnote{The recovery may measure the physical memory and condition later operations on its outcomes. Such measurement and feedback are included in one CPTP map. What is unavailable is an independent announcement of the angle, or postselection that discards unsuccessful runs.}

 A quantum memory must preserve the logical basis labels, their relative phases and any entanglement with other systems. We test this by entangling the input with an untouched reference and asking how much of that joint state survives. The resulting entanglement fidelity is
\begin{equation}
 F_e(\mathcal R\mathcal N)
 =\langle\Omega_d|[(\mathcal R\mathcal N)\otimes\mathrm{id}]
 (|\Omega_d\rangle\langle\Omega_d|)|\Omega_d\rangle,
 \label{eq:fidelity}
\end{equation}
 where $|\Omega_d\rangle=d^{-1/2}\sum_{a=1}^d|a\rangle_L|a\rangle_R$ is a maximally entangled logical/reference state and $\mathrm{id}$ acts on the reference. The logical marginal alone is maximally mixed, but its correlations with the reference contain the coherence that the memory must preserve. Perfect fidelity therefore means the recovered channel is the identity on the whole logical register. This untouched-reference definition is entanglement fidelity~\cite{Schumacher1996}. We write $F_{e,\mathrm{opt}}(\sigma,\nu)=\max_{\mathcal R}F_e(\mathcal R\mathcal N_{\sigma,\nu})$ and call $1-F_e$ the entanglement infidelity. Entanglement fidelity is the primary objective. Worst-input trace-distance error, allowing an entangled reference, is half the diamond-norm distance from the identity channel~\cite{Watrous2018}. It is never smaller than the entanglement infidelity and can be much larger for a coherent logical error. Proposition~\ref{prop:worst-input} bounds the optimal worst-input error below by the optimal entanglement infidelity and above by $d$ times it plus a higher-order remainder. The attaining recovery followed by a known logical rotation satisfies this upper bound at the first harmful order. The average fidelity over pure logical inputs equals $(dF_e+1)/(d+1)$~\cite{HorodeckiEtAl1999,Nielsen2002}, so each entanglement-fidelity statement below also determines it.

We use $T\succeq0$ for a positive-semidefinite matrix, $|T|=(T^\dagger T)^{1/2}$ for the positive operator absolute value, $\|T\|_1=\Tr|T|$ for the trace norm, and $\|T\|_\infty$ for the largest singular value. The unnormalized Hilbert--Schmidt norm is $\|T\|_{\rm HS}^2=\Tr(T^\dagger T)$. These conventions also apply to rectangular error maps.

\begin{table*}[t]
\centering\small
\begin{tabular}{@{}p{.16\textwidth}p{.27\textwidth}p{.26\textwidth}p{.235\textwidth}@{}}
\toprule
Problem & Fixed physical input & Available noise information & Recovery question\\\midrule
Weak coherent noise & Finite $V,H$ with a nondegenerate correctable prefix; arbitrary logical dimension & A fixed rescaled prior with the required finite moments & Leading optimal error, worst-input bracket and Petz comparison, Propositions~\ref{thm:optimal}--\ref{thm:petz}\\
Uncertain weak-noise statistics & The same physical prefix & Realizable ordinary moments with uniform tail control, trusted support, or a fixed alphabet and contamination & One common decoder; Proposition~\ref{thm:minimax}, alphabet stability in Propositions~\ref{prop:finite-alphabet} and~\ref{prop:contamination}\\
Finite periodic calibration & Commensurate occupied energy gaps and calibration frequencies & Measured Fourier moments compatible with a probability measure & Exact finite convex optimization, Proposition~\ref{prop:finite-calibration}\\
Calibration order & Integer-normalized occupied spectrum; rank degeneracy allowed in the prefix criterion & One positive periodic moment, with every compatible circle prior allowed & Sharp linear or fractional loss, Propositions~\ref{prop:calibration-prefix}--\ref{prop:calibration-mismatch}\\
Joint calibration & Conserved logical sectors; includes $V=I_4$, $H=Z_1+Z_2$ & Partial real or complex coherences; separate qubit data may omit a joint harmonic & Common recovery, Proposition~\ref{prop:sector-data}; joint-data example in Section~\ref{sec:joint-calibration-example}\\
Four-photon memory & Specified code and calibrated axis in Section~\ref{sec:optical-application} & One, two or three population means from ensembles sharing the rate distribution & Exact fidelity frontiers, required precision and axis/noise stability\\\bottomrule
\end{tabular}
\caption{The systems, information and objectives used in the paper. Each problem fixes the memory and selects one deterministic recovery using the stated information. The local coefficient needs the nondegenerate prefix of Proposition~\ref{thm:optimal}; the periodic prefix criterion also covers degenerate chains. Its fractional-loss extension returns to a nondegenerate first defective boundary. The optical application evaluates these general mechanisms for one specified code and generator. The proposition statements give the full hypotheses.}
\label{tab:model-map}
\end{table*}

The three subsections below develop the tools used throughout the paper. The first states the exact-correction condition in terms of logical overlaps between error images. The second applies that condition to the powers of $H$ generated by a weak rotation; their overlaps form a matrix-valued spectral measure, which also fixes the physical space the rotation can reach. The third introduces the complementary environment and the orthogonal polynomials of the angle distribution, so that the physical and statistical projections behind the leading recovery error appear side by side.

\subsection{What an error-correcting recovery must preserve}

Redundancy is useful when different physical error images carry the same unknown logical coefficients. Let $E_\alpha$ denote a member of an error family; the label $\alpha$ need not identify a measured event. Exact correction means that a single isometric realization $U_R$ of recovery, with a sufficiently large auxiliary space, can separate an error-dependent auxiliary vector from the logical input,
\begin{equation}
 U_R E_\alpha V|\psi\rangle
 =|\xi_\alpha\rangle_{\rm syn}\otimes|\psi\rangle_L.
 \label{eq:recovery-factorization}
\end{equation}
 The vectors $|\xi_\alpha\rangle$ can be unnormalized. If the $E_\alpha$ are the branch maps of a trace-preserving noise channel, their branch weights sum to one, and discarding the auxiliary system returns the normalized logical state. More generally, the $E_\alpha$ can be amplitude maps such as $H^n$; these are terms in an error expansion, not separate trace-preserving processes. This auxiliary information about the error is called the syndrome. Figure~\ref{fig:memory-task} illustrates this separation. Recovery can carry this syndrome coherently into the discarded auxiliary system. The same map must preserve the logical input for arbitrary superpositions of the error amplitudes.

Taking inner products of Eq.~\eqref{eq:recovery-factorization} gives $V^\dagger E_\alpha^\dagger E_\beta V=\langle\xi_\alpha|\xi_\beta\rangle I_d$. The converse is the Knill--Laflamme criterion~\cite{KnillLaflamme1997}. Thus correction is possible precisely when the overlaps between error images are independent of the logical state. This motivates keeping the entire logical matrix
\begin{equation}
 \mathscr G(E,F)=V^\dagger E^\dagger FV.
 \label{eq:kernel}
\end{equation}
 A nonscalar value obstructs the factorization. The syndrome vectors themselves need not be orthogonal. The map $\mathscr G$ is a positive matrix-valued kernel because $\sum_{ij}v_i^\dagger\mathscr G(E_i,E_j)v_j=\|\sum_jE_jVv_j\|^2\ge0$ for all $v_j\in\mathbb C^d$. 

\begin{figure*}[t]
\centering
\includegraphics[width=.98\textwidth]{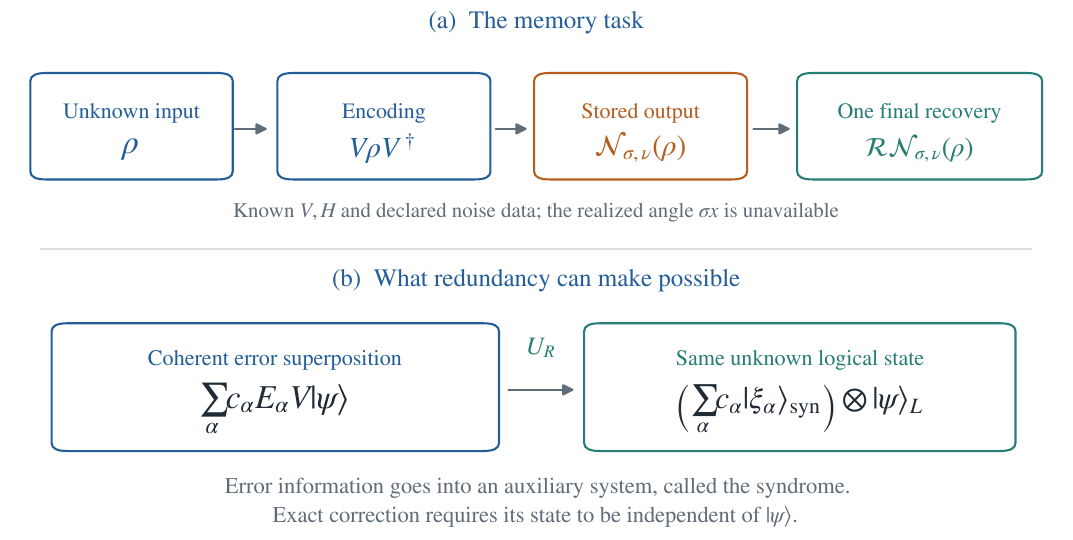}
\caption{The encoded-memory problem. (a) One final recovery acts on the averaged physical output. It knows the encoding, generator and declared noise information, but not the angle in the stored run. (b) Exact error correction sends every allowed error superposition to an auxiliary vector tensored with the same unknown logical state. The auxiliary vector can carry error amplitudes but must be independent of the logical input. Discarding it then preserves the stored quantum information.}
\label{fig:memory-task}
\end{figure*}

\subsection{Krylov spaces and their encoded spectral data}

The error family relevant to a weak rotation follows directly from the stored state,
\begin{equation}
\begin{split}
 e^{-i\sigma xH}V|\psi\rangle
 ={}&V|\psi\rangle-i\sigma xHV|\psi\rangle\\
 &-\frac{\sigma^2x^2}{2}H^2V|\psi\rangle+\cdots.
\end{split}
\label{eq:physical-error-hierarchy}
\end{equation}
 At order $n$, the rotation explores the image of $H^nV$. Choose a nonnegative integer $L$ such that one recovery corrects $I,H,\ldots,H^L$. Linearity puts the angle dependence of these terms into the auxiliary amplitude while preserving the logical state for the whole noise ensemble. Under the noise-moment assumptions stated below, the first potentially harmful amplitude has order $\sigma^{L+1}$ and its averaged squared error starts at order $\sigma^{2L+2}$. Our quantitative problem is to determine the best coefficient and the decoder that attains it, including cases where that coefficient vanishes.

 Different powers can reach the same physical directions. Subtracting the lower-order images from $H^nV$ separates genuinely new error directions from directions already reached. Gram--Schmidt on the sequence $V,HV,H^2V,\ldots$ is the block Krylov construction, or block Lanczos method~\cite{GolubUnderwood1977}; its moment representation is described by matrix-valued orthogonal polynomials~\cite{DamanikPushnitskiSimon2008}. Each block has $d$ logical input columns, so the orthogonalization follows an entire encoded register. The correctability and rank conditions below determine when the new images are full undistorted logical copies. The construction organizes the actual error amplitudes into orthogonal physical subspaces that a common recovery can map to syndrome labels. A finite list of moments determines a local expansion near zero angle. Finite-angle recovery depends on the complete exponential and can distinguish spectra sharing that list. We therefore first construct one spectral object retaining the whole orbit. Its degree filtration will give the weak-noise boundary in Section~\ref{sec:geometry}; its complete weights will determine the finite-angle problem in Section~\ref{sec:finite-angle}.

\paragraph{The spectral measure and exact cyclic reduction.}
\label{sec:spectral-data}

 Write $H=\sum_r E_rP_r$, where the distinct $E_r$ are eigenvalues and $P_r$ is the orthogonal projector onto the corresponding energy eigenspace, including any degeneracy. The logical matrix $S_r=V^\dagger P_rV$ determines its population in an encoded state, $\langle\psi|S_r|\psi\rangle$. The matrix weights retain how energy occupations distinguish logical states. The matrix-valued spectral measure collects them as
\begin{equation}
\begin{gathered}
 \Sigma(dE)=\sum_r S_r\delta_{E_r}(dE),\\
 S_r=V^\dagger P_rV\succeq0,\qquad \sum_rS_r=I_d.
\end{gathered}
 \label{eq:matrix-measure}
\end{equation}
Here $\delta_{E_r}$ is a unit point mass at energy $E_r$. Integrating a function against $\Sigma$ means summing that function times $S_r$. For scalar polynomials $p,q$, and real angle parameters $s,t$, the spectral calculus gives
\begin{equation}
\begin{aligned}
 V^\dagger\overline{p}(H)q(H)V
 &=\int\overline{p(E)}q(E)\,\Sigma(dE),\\
 V^\dagger e^{-i(t-s)H}V
 &=\int e^{-i(t-s)E}\,\Sigma(dE).
\end{aligned}
\label{eq:polynomial-kernel}
\end{equation}
The trace of the matrix measure averages over logical states and can erase their different energy distributions. To see why this matters, encode a qubit as $V|\psi\rangle=|0\rangle_s\otimes|\psi\rangle$ and compare $H_1=X_s\otimes I_2$ with $H_2=I_s\otimes Z_L$, using Pauli matrices $X,Z$. Under $H_1$ the state becomes $(\cos t|0\rangle-i\sin t|1\rangle)_s\otimes|\psi\rangle$ and discarding the first qubit recovers $|\psi\rangle$ exactly. Under $H_2$ that qubit stays fixed while the unknown rotation acts on the logical state. A centered Gaussian angle of variance $\sigma^2$ then has optimum $(1+e^{-2\sigma^2})/2$. Nevertheless, both scalar measures $d^{-1}\Tr\Sigma$ equal $(\delta_{-1}+\delta_1)/2$, so their scalar Lanczos chains, obtained by orthogonalizing generator powers with that scalar measure, coincide. Their matrix weights, respectively $I_2/2$ and $(I_2\pm Z_L)/2$, distinguish the two situations. A scalar Krylov description can track motion while losing precisely the logical dependence relevant to recovery.

The full matrix measure also determines an exact reduction of the reachable physical space. The code-cyclic space
\begin{equation}
 \mathcal H_{\rm cyc}(H,V)=\operatorname{span}\{H^nVv:n\ge0,\ v\in\mathbb C^d\}
 \label{eq:code-cyclic-space}
\end{equation}
is the smallest $H$-invariant subspace containing the code. The measure constructs it directly from vector polynomials $f(E)=\sum_nE^nv_n$ in $\mathscr P=\mathbb C[E]\otimes\mathbb C^d$, with
\begin{equation}
 \langle f,g\rangle_\Sigma
 =\sum_r f(E_r)^\dagger S_rg(E_r).
 \label{eq:vector-polynomial-inner-product}
\end{equation}
This form can have null vectors. Let $\mathscr N_\Sigma$ be its zero-norm subspace. Identifying polynomials whose difference lies in $\mathscr N_\Sigma$ gives the isometry
\begin{equation}
\begin{gathered}
 \mathscr P/\mathscr N_\Sigma\longrightarrow\mathcal H_{\rm cyc}(H,V),\\
 [f]\longmapsto\sum_nH^nVv_n.
\end{gathered}
\label{eq:cyclic-polynomial-quotient}
\end{equation}
 The inner product of its images is exactly Eq.~\eqref{eq:vector-polynomial-inner-product}, and its image is the cyclic space. Multiplication by $E$ becomes $H$, while constant polynomials embed the logical register as $V$. These operations are well defined because a null polynomial obeys $\sqrt{S_r}f(E_r)=0$ at each energy, also after multiplication by $E$. Thus the matrix measure fixes the entire reachable pair. Polynomial filtering and finite-support quotients also organize Krylov state changes~\cite{ChowdhuryMahapatra2026Christoffel}. Here the quotient removes absent directions, as in minimal dilation constructions~\cite{Stinespring1955}; it assumes no commutativity between the weights $S_r$.

\begin{proposition}[Cyclic equivalence]
\label{prop:equivalence}
Two finite code--Hamiltonian pairs with the same $\Sigma$, in a common logical basis, have equal optimal recovery fidelity for every distribution of rotation angles. More generally, they admit the same set of recovered logical channels, so every figure of merit of the recovered channel, including worst-input error, has the same optimum. Their minimal cyclic realizations are isometrically equivalent.
\end{proposition}
\begin{proof}
 Both pairs realize the quotient in Eq.~\eqref{eq:cyclic-polynomial-quotient}.
Their quotient isometries define a unitary $J:\mathcal H_{{\rm cyc},1}\to
\mathcal H_{{\rm cyc},2}$ with $JV_1=V_2$ and $JH_1=H_2J$ on the cyclic spaces.
Thus $J$ transfers every encoded rotation. For a recovery $\mathcal R_1$,
let $\Pi_{{\rm cyc},2}$ project onto the second cyclic space and fix a logical
density matrix $\rho_0$. The map
\begin{equation}
\begin{split}
\mathcal R_2(\rho)={}&
\mathcal R_1(J^\dagger\Pi_{{\rm cyc},2}\rho\Pi_{{\rm cyc},2}J)\\
&+\Tr[(I-\Pi_{{\rm cyc},2})\rho]\,\rho_0
\end{split}
\label{eq:cyclic-recovery-transfer}
\end{equation}
is CPTP and satisfies $\mathcal R_2\mathcal N_2=\mathcal R_1\mathcal N_1$
for every angle prior. Interchanging the two spaces gives the reverse transfer,
so both optimizations have the same value.
\end{proof}

A minimal model can be read directly from the matrix weights,
\begin{equation}
\begin{gathered}
 \mathcal H_\Sigma=\bigoplus_r\operatorname{Ran}S_r,\qquad
 V_\Sigma v=\bigoplus_r\sqrt{S_r}\,v,\\
 H_\Sigma=\bigoplus_r E_r I_{\operatorname{Ran}S_r}.
\end{gathered}
 \label{eq:minimal-model}
\end{equation}
The identification is $[f]\mapsto\bigoplus_r\sqrt{S_r}f(E_r)$. Polynomial interpolation at the distinct energies reaches every summand, giving dimension $\sum_r\operatorname{rank}S_r$. This reduction preserves unrestricted recovery, although the cyclic isometry can change laboratory implementation costs. A different operational reduction fixes a decoder and follows its decodable observables under a Lindbladian; code-visible operator Krylov spaces then describe the resulting logical channel~\cite{GordeychukKiselev2026}. Here the encoded amplitude measure is fixed before selecting a decoder, which permits comparison and optimization over all final recoveries.

For a prefix correcting powers through $L$, write $M_j=V^\dagger H^jV=\mu_jI_d$ for $j\le2L$. These scalar physical moments make the polynomial inner product factor into an error coordinate and a logical coordinate. For $p(E)=\sum_{i=0}^L\widehat p_iE^i$, $q(E)=\sum_{j=0}^L\widehat q_jE^j$ and $u,v\in\mathbb C^d$,
\begin{equation}
 \langle p\otimes u,q\otimes v\rangle_\Sigma
 =\left(\sum_{i,j=0}^L\overline{ \widehat p_i} \widehat q_j\mu_{i+j}\right)\langle u,v\rangle.
 \label{eq:prefix-polynomial-factorization}
\end{equation}
Scalar orthogonalization therefore produces whole logical shells. The boundary matrices test whether this factorization continues by one further power; beyond the prefix the matrix weights retain its logical dependence.

\subsection{The environment and two complementary projections}

 There is a complementary interpretation of the same overlaps. If an unobserved system can distinguish logical states, the memory alone cannot retain all their superpositions. An environment for the averaged channel makes this information loss explicit. It purifies the random-unitary average and represents information inaccessible to the decoder. Its construction requires no microscopic bath model. The externally sampled $x$ remains input independent; the purifying environment can nevertheless correlate with the logical state.

On $L^2(\nu)$ let $X_\nu$ multiply a function by $x$, and let $|1\rangle$ be the normalized constant function. The unitary $e^{-i\sigma H\otimes X_\nu}$ acts on $V|\psi\rangle\otimes|1\rangle$. Spectral calculus defines it even when $X_\nu$ is unbounded; moments are needed only for polynomial expansions. In the energy basis this same construction reads as follows.
Let $\varphi_\nu(u)=\int e^{-ixu}\,d\nu(x)$ and choose environment vectors $e_r^{\rm env}(x)=e^{-i\sigma xE_r}$ in $L^2(\nu)$. Their overlaps are $\langle e_r^{\rm env},e_s^{\rm env}\rangle=\varphi_\nu(\sigma(E_s-E_r))$. The isometry
\begin{equation}
 V|\psi\rangle\longmapsto
 \sum_rP_rV|\psi\rangle\otimes|e_r^{\rm env}\rangle
\end{equation}
 is a purification of the encoded channel, also called a Stinespring dilation~\cite{Stinespring1955}. Tracing its environment reproduces Eq.~\eqref{eq:channel}, since the overlaps of the $e_r^{\rm env}$ average the relative energy phases. Start instead with the logical/reference state $|\Omega_d\rangle$ and trace the physical output. The remaining reference--environment state is
\begin{equation}
 \rho_{RE}=\frac1d\sum_rS_r^{\mathsf T}\otimes|e_r^{\rm env}\rangle\langle e_r^{\rm env}|.
 \label{eq:environment}
\end{equation}
The transpose refers to the basis defining $|\Omega_d\rangle$ and follows from $(T\otimes I)|\Omega_d\rangle=(I\otimes T^{\mathsf T})|\Omega_d\rangle$. The weights $S_r$ describe logical dependence in energy; the noise determines how distinguishable the environmental vectors are. With root fidelity $f(\rho,\omega)=\|\sqrt\rho\sqrt\omega\|_1$, fixed-input recovery duality gives
\begin{equation}
 F_{e,\mathrm{opt}}
 =\max_{\omega_E}f\!\left(\rho_{RE},I_R/d\otimes\omega_E\right)^2.
 \label{eq:decoupling}
\end{equation}
The maximization is over environmental density matrices $\omega_E$. A product state $I_R/d\otimes\omega_E$ contains no information about the reference, so this formula measures how close the joint reference--environment state is to a product with the fixed reference marginal. The identity follows from recovery dilation and Uhlmann's purification theorem~\cite{BenyOreshkov2011,KonigRennerSchaffner2009}. It gives an interpretation of the recovery limit without granting the decoder access to the environment. Near zero noise, new environmental eigenvalues can emerge from zero; the shell proof in Section~\ref{app:optimal} handles that change of rank directly.

Orthogonalizing powers of $H$ on the code and powers of $X_\nu$ on $|1\rangle$ therefore performs two projections in the same dilation. The physical inner product is matrix valued because it must retain every logical direction. The statistical inner product is scalar because the sampled angle is independent of the input. If the monic noise polynomials, the residuals of $x^j$ orthogonal to all lower powers in $L^2(\nu)$ and thus with leading coefficient one, have squared norms $h_j>0$, their successive off-diagonal Jacobi coefficients, called hoppings, are $b_j^{\nu}=\sqrt{h_j/h_{j-1}}$, with $h_0=1$.  Hence
\begin{equation}
h_m=\prod_{j=1}^m( b_j^{\nu})^2.
\label{eq:noise-chain-product}
\end{equation}
A finite alphabet terminates this statistical chain. Gaussian uncertainty gives the Hermite chain. Neither construction grants recovery access to the auxiliary coordinate $x$. Their payoff is a factorization of the first unavoidable error into physical logical distortion and an uncancellable statistical amplitude.

Physical leakage and statistical cancellation use the same operation on a positive Gram matrix. Let $F$ and $X$ be linear maps with a common Hilbert-space codomain and finite-dimensional domains, and write
\begin{equation}
 \mathsf K=
 \begin{pmatrix}F^\dagger F&F^\dagger X\\X^\dagger F&X^\dagger X\end{pmatrix}
 =\begin{pmatrix}\mathsf G&\mathsf C\\\mathsf C^\dagger&\mathsf M\end{pmatrix}\succeq0.
 \label{eq:gram-residual-data}
\end{equation}
Then $\operatorname{Ran}\mathsf C\subseteq\operatorname{Ran}\mathsf G$. The Moore--Penrose inverse $\mathsf G^+$ acts as the inverse on that range and as zero on its kernel. Consequently
\begin{equation}
\begin{gathered}
 P_F=F\mathsf G^+F^\dagger,\\
 \mathsf M-\mathsf C^\dagger\mathsf G^+\mathsf C
 =X^\dagger(I-P_F)X\succeq0.
\end{gathered}
\label{eq:gram-residual-projection}
\end{equation}
The operator $P_F$ projects orthogonally onto $\operatorname{Ran}F$, so the Schur complement is the Gram matrix of the residual after removing the old images. This identity includes linearly dependent old images. In the physical problem, $F$ will collect the old error images and $X$ the next one; in the statistical problem they will collect lower angle powers and the next power. The same Schur-complement identity gives the physical boundary residual and the statistical residual $h_m$, respectively, 
 in Sections~\ref{sec:boundary-moments} and~\ref{sec:noise-moments}. Polar alignment and removal of the scalar logical part turn its matrix residual into a recovery cost.

The matrix measure $\Sigma$ fixes the physical recovery problem. Its moment prefix will identify the first logical distortion, while its occupied energy gaps will determine which Fourier moments of the angle distribution enter the full channel. This separates the fixed dynamics from the information available about the noise. We first evaluate the local optimum, then ask which part of that guarantee survives incomplete calibration.

\section{Correctable Krylov prefixes and optimal recovery}
\label{sec:geometry}

 We now quantify the first departure from exact error correction. There are two separate questions. At which order does the physical error image cease to carry every logical state in the same way? How much of that harmful amplitude can a decoder cancel using the lower-order amplitudes already present? Orthogonalizing the physical images answers the first question. Orthogonalizing their angle dependence answers the second. Figure~\ref{fig:two-orthogonalizations} summarizes how these two projections enter the recovery coefficient.

\subsection{Correctable shells and their boundary}

Let the nonnegative integer $L$ be a correction depth for this particular generator, meaning that the code exactly corrects $I,H,\ldots,H^L$. The correctable prefix is the initial error span $\operatorname{span}\{I,H,\ldots,H^L\}$ on the code, together with one recovery that corrects every member. This depth is specific to $H$; code distance describes correction against a different, prescribed family of physical errors. Every encoding allows $L=0$; larger depths require the scalar-moment conditions below. The overlap criterion of Eq.~\eqref{eq:kernel} turns this assumption into scalar compressed physical moments,
\begin{equation}
 M_j=V^\dagger H^jV=\mu_jI_d,\qquad 0\le j\le2L.
 \label{eq:scalar-moments}
\end{equation}
Indeed, $(H^i)^\dagger H^j=H^{i+j}$, so correction of the powers through $L$ tests moments through $2L$. The numbers $\mu_j$ describe the known code and generator; they are distinct from the moments of the random variable $x$.

Set $W_0=V$. The first subtraction makes the construction concrete. If $L\ge1$ and $\mu_2>\mu_1^2$, define
\begin{equation}
\begin{gathered}
 W_1=\frac{(H-\mu_1I)V}{\sqrt{\mu_2-\mu_1^2}},\\
 V^\dagger W_1=0,\qquad W_1^\dagger W_1=I_d.
\end{gathered}
 \label{eq:first-shell-explained}
\end{equation}
Thus $W_1|\psi\rangle$ is an undistorted copy of every logical state in a subspace orthogonal to the code. A recovery can send $V|\psi\rangle$ to $|0\rangle_{\rm syn}|\psi\rangle$ and $W_1|\psi\rangle$ to $|1\rangle_{\rm syn}|\psi\rangle$. Its syndrome coordinate changes, while all unknown coefficients of $|\psi\rangle$ survive.

Higher powers repeat this subtraction. Assume the scalar Gram matrix $[\mu_{i+j}]_{i,j=0}^L$ is positive definite. Such a matrix, whose entries depend on the sum of the indices, is called a Hankel moment matrix. Positivity means that no nonzero polynomial of degree at most $L$ annihilates the encoded states. Each new order therefore adds an independent $d$-dimensional image of the logical register.\footnote{The word ``copy'' denotes an alternative embedding $W_n|\psi\rangle$ into an orthogonal subspace. No operation creates several simultaneous copies of an unknown state. The nondegenerate prefix uses $(L+1)d$ physical dimensions, so $(L+1)d\le\dim\mathcal H$. Increasing the correction depth at fixed logical dimension requires enough physical space.} This is the nondegenerate-prefix assumption; it does not require a nondegenerate energy spectrum. Gram--Schmidt uses the same scalar coefficients for every logical vector because Eq.~\eqref{eq:scalar-moments} makes every overlap scalar. Its normalized maps obey $W_n^\dagger W_j=\delta_{nj}I_d$. We call their images the Krylov shells.

Their orthogonality permits one common decoding map on the old error span,
\begin{equation}
 U_{\rm pre}\sum_{n=0}^LW_nv_n
 =\sum_{n=0}^L|n\rangle_{\rm syn}\otimes v_n,
 \quad v_n\in\mathbb C^d.
 \label{eq:prefix-tensor-decoding}
\end{equation}
 It preserves inner products and can be extended, using a larger auxiliary space, to an isometry on the full physical space. Every $H^jV$ with $j\le L$ is a scalar combination of the $W_n$. For a fixed dimensionless angle $\tau$, equal to $\sigma x$ in one run of Eq.~\eqref{eq:channel}, the extended map therefore factors the evolving state through Taylor order $L$,
\begin{equation}
 \begin{split}
 U_{\rm pre}e^{-i\tau H}V|\psi\rangle
 ={}&\left[\sum_{n=0}^L f_n(\tau)|n\rangle_{\rm syn}\right]
       \otimes|\psi\rangle\\
 &{~}+O(\tau^{L+1}).
 \end{split}
 \label{eq:prefix-dynamic-factorization}
\end{equation}
The scalar polynomials $f_n$ do not depend on $|\psi\rangle$. Thus the rotation changes a syndrome amplitude while the unknown logical coefficients stay in one common factor. The remainder is where that separation can first fail.

Writing the subtracted degree-$n$ polynomial as $p_n$, with leading coefficient one, gives
\begin{equation}
 W_n=p_n(H)V/g_n,\quad
 g_n=\prod_{j=1}^n b_j,\quad g_0=1.
 \label{eq:shells}
\end{equation}
Here $g_n$ normalizes the residual map so that $W_n$ is an isometry. The positive numbers $b_j$ are the hoppings identified in Eq.~\eqref{eq:correctable-recurrence} below. Hermiticity of $H$ makes its matrix between these shells tridiagonal. Indeed, $W_j^\dagger HW_n=0$ for $j<n-1$, since $HW_j$ already lies in lower shells. For $n<L$ this gives
\begin{equation}
 HW_n=b_nW_{n-1}+a_nW_n+b_{n+1}W_{n+1},
 \label{eq:correctable-recurrence}
\end{equation}
 with scalar $a_n\in\mathbb R$ and $b_{n+1}>0$; the first term is absent at $n=0$. Comparing leading powers gives $b_{n+1}=g_{n+1}/g_n$, hence the product in Eq.~\eqref{eq:shells}. This tridiagonal representation is the Lanczos chain~\cite{Lanczos1950}, and its off-diagonal coefficients $b_n$ are its hoppings. Its sites label the orthogonal error subspaces reached by successive generator powers. For a single encoded state $V|\psi\rangle$, the same recurrence defines the Krylov basis used for spread complexity~\cite{BalasubramanianEtAl2022Spread}, with coefficients fixed by the moments $\langle\psi|M_j|\psi\rangle$. Scalar moments through $2L$ therefore give every encoded state the same $a_0,\ldots,a_{L-1}$ and $b_1,\ldots,b_L$. Conversely, if all encoded states share these coefficients, each $M_j$ with $j\le2L$ has the same expectation in every logical state and is scalar. A correctable prefix is thus a logical-state-independent initial segment of the encoded Lanczos chain. The first coefficient that can vary with the input is $a_L(\psi)=\langle\psi|W_L^\dagger HW_L|\psi\rangle$.

Until this point, $H$ moves the syndrome label with no logical dependence. Acting once more on the last correctable shell can break that property. Let $\Pi_L$ project onto all old shells and separate the part remaining in the last shell from the part entering new directions,
\begin{equation}
\begin{aligned}
 \Pi_L&=\sum_{n=0}^LW_nW_n^\dagger,\\
 A_\star&=W_L^\dagger HW_L, \\
 E_\star&=(I-\Pi_L)HW_L,\label{eq:defects}\\
 B_\star&=(E_\star^\dagger E_\star)^{1/2}.
\end{aligned}
\end{equation}
 The matrix $A_\star$ describes the action on the logical amplitudes within shell $L$. The map $E_\star$ describes displacement out of the old error span into new error subspaces; it is motion within the physical space, not loss from it. Its polar decomposition separates an orientation from the positive matrix $B_\star$. The orientation preserves inner products on the support of $B_\star$, including when the leakage has deficient rank. Recovery can undo the orientation. Unequal eigenvalues of $B_\star$, which are the singular amplitudes of $E_\star$, mean that different logical directions leak by different amounts. This dependence cannot be removed by an orientation change. The construction covers zero or deficient-rank leakage and arbitrary relative eigenbases of $A_\star$ and $B_\star$.

A small block Hamiltonian illustrates both effects. On $\mathbb C^3\otimes\mathbb C^d$, encode $V|\psi\rangle=|0\rangle\otimes|\psi\rangle$ and take
\begin{equation}
\begin{gathered}
 H=\begin{pmatrix}0&bI_d&0\\bI_d&A&B\\0&B&0\end{pmatrix},\\
 b>0,\qquad A=A^\dagger,\qquad B\succeq0.
\end{gathered}
 \label{eq:three-shell-teaching-model}
\end{equation}
 Here $M_1=0$, $M_2=b^2I_d$, $W_1|\psi\rangle=|1\rangle\otimes|\psi\rangle$, and the boundary matrices are $A_\star=A$, $B_\star=B$. For a small fixed angle $\tau$,
\begin{equation}
\begin{aligned}
 &e^{-i\tau H}V|\psi\rangle\\
 &\quad=|0\rangle\otimes(1-b^2\tau^2/2)|\psi\rangle\\
 &\qquad+|1\rangle\otimes(-ib\tau I_d-b\tau^2A/2)|\psi\rangle\\
 &\qquad-|2\rangle\otimes(b\tau^2B/2)|\psi\rangle+O(\tau^3).
\end{aligned}
\label{eq:three-shell-teaching-expansion}
\end{equation}
 The order-$\tau$ displacement carries an intact logical copy. At order $\tau^2$, $A$ acts on that copy while $B$ controls a new displacement. If $A$ and $B$ are both scalar, the whole Hamiltonian acts on the first tensor factor and the logical state survives exactly. The general construction isolates precisely the nonscalar parts of these two amplitudes.

A scalar action changes every logical direction equally and belongs to disposable syndrome motion. For a logical amplitude $T$, its overlap with the maximally entangled test state is $\Tr(T)/d$, since $\langle\Omega_d|(T\otimes I_d)|\Omega_d\rangle=\Tr(T)/d$. Subtracting this scalar component isolates the amplitude orthogonal to that test state. This explains why entanglement fidelity selects the normalized trace below. We write the normalized trace $\operatorname{tr}_d$, centered part and, for Hermitian $T$, variance as
\begin{equation}
\begin{gathered}
 \operatorname{tr}_d(T)=\Tr(T)/d,\qquad T_c=T-\operatorname{tr}_d(T)I_d,\\
 \operatorname{Var}_d(T)=\operatorname{tr}_d(T^2)-\operatorname{tr}_d(T)^2.
\end{gathered}
 \label{eq:variance}
\end{equation}
In particular, $\operatorname{Var}_d(T)=\|T_c\|_{\rm HS}^2/d$ for Hermitian $T$ and vanishes precisely when the action is logically independent. The combined boundary distortion is
\begin{equation}
 \mathcal V=\operatorname{Var}_d(A_\star)+\operatorname{Var}_d(B_\star).
 \label{eq:physical-defect}
\end{equation}

In the Lanczos language introduced after Eq.~\eqref{eq:correctable-recurrence}, every encoded state has $a_L(\psi)=\langle\psi|A_\star|\psi\rangle$, and its next hopping obeys $b_{L+1}(\psi)^2=\langle\psi|(A_\star-a_L(\psi))^2|\psi\rangle+\langle\psi|B_\star^2|\psi\rangle$. The two terms are squared norms of orthogonal components, $W_L(A_\star-a_L(\psi))|\psi\rangle$ in the last shell and $E_\star|\psi\rangle$ outside the old shells. They explain how logical dependence enters the next step even when the earlier chain is shared. Both boundary coefficients are input independent precisely when $\mathcal V=0$. Indeed, a constant expectation makes $A_\star$ scalar, and then a constant next hopping makes $B_\star^2$, and hence its positive square root, scalar. This includes a vanishing next hopping.

The defect $\mathcal V$ is also the value of a best-approximation problem. A perfectly correctable next boundary would have a scalar action $a$ in the old shell and one common amplitude $\beta$ into an undistorted new logical copy $J$. Compare such boundaries with the actual column $D_L=HW_L-b_LW_{L-1}=W_LA_\star+E_\star$, omitting the subtraction when $L=0$. In the normalized Hilbert--Schmidt norm $\|T\|_{2,d}^2=d^{-1}\Tr(T^\dagger T)$, their least squared distance is
\begin{equation}
 \min_{a\in\mathbb R,\,\beta\ge0,\,J}
 \|D_L-aW_L-\beta J\|_{2,d}^2=\mathcal V,
 \label{eq:boundary-distance}
\end{equation}
where $J$ is an isometry into a new shell orthogonal to the old ones. A larger unused codomain is allowed when the physical complement has dimension below $d$.

 The occupied-shell and leakage terms are orthogonal. Scalar projection fixes $a= \operatorname{tr}_d(A_\star)$. Polar decomposition gives $\max_J\operatorname{Re}\Tr(J^\dagger E_\star)=\Tr B_\star$, after which $\beta= \operatorname{tr}_d(B_\star)$. This proves Eq.~\eqref{eq:boundary-distance}. The old shells have fixed the logical frame, but the orientation of the new leakage can be undone. Its unequal singular values remain harmful. Thus the boundary cost depends on the spread of the singular amplitudes in $B_\star$.
Knowing the matrices does not by itself remove this distortion. In the block model, the occupied-shell amplitude contains both $-ib\tau I_d$ and $-b\tau^2A/2$. A fixed change of logical frame must preserve the first term for every input while correcting the second for every unknown $\tau$. These demands generally differ. The next step quantifies what can still be cancelled by mixing amplitudes from other shells with coefficients chosen from the noise information.

\subsection{The noise residual and the optimal coefficient}

Put $m=L+1$. The first potentially harmful physical amplitude carries the factor $\sigma^mx^m/m!$ from Eq.~\eqref{eq:physical-error-hierarchy}. The lower shells already carry amplitudes proportional to $1,x,\ldots,x^{m-1}$. A fixed small coherent rotation in the recovery can mix those amplitudes into the harmful one. Its coefficients can be chosen using the distribution, even though the realized $x$ is never measured. The best subtraction is consequently the least-squares approximation of $x^m$ by lower powers. The coefficient of $x^m$ is fixed by the physical Taylor amplitude; recovery chooses the lower-power coefficients. The residual is therefore monic, meaning that its highest coefficient is one. A fixed recovery here means one operation for the declared $\sigma$ and noise information, used in every run.

For the asymptotic result, keep the rescaled prior $\nu$ fixed as $\sigma\downarrow0$. Assume $\int|x|^{2m}d\nu<\infty$ and a positive-definite noise moment matrix through degree $m-1$. There is then a unique monic residual polynomial
\begin{equation}
\begin{gathered}
 q_m(x)=x^m-\sum_{n=0}^{m-1}d_nx^n,\\
 h_m=\int q_m(x)^2\,d\nu(x),
\end{gathered}
 \label{eq:noise-poly}
\end{equation}
 satisfying $\int q_m(x)x^n\,d\nu(x)=0$ for $n<m$. Its squared norm $h_m$ is the part of the harmful angle amplitude that lower powers cannot cancel in mean square. This orthogonality is with respect to the noise distribution, whereas the shell orthogonality used physical error overlaps. A positive-definite noise moment matrix through degree $m-1$ requires at least $m$ support points. The next norm $h_m$ can vanish if $x^m$ is already a lower-degree polynomial on that support. The polynomial relation holds across the declared noise ensemble and permits one common subtraction for its members.

\begin{proposition}[Optimal weak-noise recovery]
\label{thm:optimal}
 Fix the finite encoding $V$ and Hermitian generator $H$. Suppose
$M_j=\mu_jI_d$ for $0\le j\le2L$ and $[\mu_{i+j}]_{i,j=0}^L\succ0$.
Put $m=L+1$. Let the rescaled prior $\nu$ be fixed, with finite $2m$-th moment
and $[\mathbb E_\nu [x^{i+j}]]_{i,j=0}^{m-1}\succ0$.
With the boundary geometry of Eqs.~\eqref{eq:defects}--\eqref{eq:physical-defect}
and the noise residual of Eq.~\eqref{eq:noise-poly}, optimization over all CPTP recoveries gives
\begin{equation}
 1-F_{e,\mathrm{opt}}(\sigma,\nu)
 =\frac{g_L^2h_m}{(m!)^2}\mathcal V\,\sigma^{2m}
 +o(\sigma^{2m}).
 \label{eq:optimal-coefficient}
\end{equation}
The formula includes singular $B_\star$, arbitrary finite logical dimension and the terminal case $h_m=0$.
It also includes $L=0$, where the sole physical condition is $V^\dagger V=I_d$ and $h_1=\operatorname{Var}_\nu(x)$.
\end{proposition}
\begin{proof}
Section~\ref{app:optimal} proves the converse by fixing the limiting logical frame on the earlier shells and bounding the remaining polar-alignment cost. Its coherent polynomial subtraction attains that bound; Appendix~\ref{app:attainment-details} gives the generators and the finite-moment remainder estimate.
\end{proof}

As in Eq.~\eqref{eq:intro-leading-coefficient}, we write $C_{\mathrm{opt}}=g_L^2h_m\mathcal V/(m!)^2$ for the coefficient in Eq.~\eqref{eq:optimal-coefficient}. Equation~\eqref{eq:optimal-coefficient} is both a performance limit and a recovery construction. The factor $g_L/m!$ is the propagation amplitude from the code to the boundary. The factor $\mathcal V$ measures its logical distortion after the best physical alignment, and $h_m$ measures the angle dependence left after coherent subtraction. The converse rules out a better coefficient from any deterministic recovery, including ones outside the shell construction.

To see which statistical information fixes the coefficient, first take a centered unit-variance Gaussian. It has $q_m=\mathrm{He}_m$ and $h_m=m!$, where $\mathrm{He}_m$ is the probabilists' Hermite polynomial. The first protected case, $m=2$, then tests whether variance alone is sufficient. For a general centered unit-variance prior the residual is
\begin{equation}
 q_2=x^2-(\mathbb E[x^3])x-1,\quad
 h_2=\mathbb E[x^4]-1-(\mathbb E[x^3])^2.
 \label{eq:m2}
\end{equation}
For symmetric noise this reduces to $q_2=x^2-1$. In the block model of Eq.~\eqref{eq:three-shell-teaching-model}, the optimal quartic coefficient is $b^2h_2[\operatorname{Var}_d(A)+\operatorname{Var}_d(B)]/4$. Unit-variance Gaussian noise has $h_2=2$, whereas equally likely $x=\pm1$ has $h_2=0$. Both have the same angle variance, but every binary run has the same squared angle and permits the same leading cancellation. A zero coefficient postpones the loss to higher order; Section~\ref{sec:finite-support} distinguishes this from exact recovery.

The subtraction can be seen directly in the same block model. For symmetric noise with $y_2=\mathbb E[x^2]$, let
\begin{equation}
\begin{aligned}
 D&=|1\rangle\langle0|\otimes A_c+|2\rangle\langle0|\otimes B_c,\\
 G_\sigma&=\frac{b\sigma^2y_2}{2}(D-D^\dagger).
\end{aligned}
\label{eq:teaching-coherent-subtraction}
\end{equation}
Applying the unitary $e^{G_\sigma}$ before discarding the shell coordinate adds $(b\sigma^2y_2/2)DV|\psi\rangle$ to the initial code amplitude. It changes the leading traceless error $-(b\sigma^2x^2/2)DV|\psi\rangle$ into $-[b\sigma^2(x^2-y_2)/2]DV|\psi\rangle$. Thus one fixed operation subtracts the ensemble mean of the harmful squared-angle amplitude without measuring either the angle or the logical state. The terms omitted here start at higher amplitude order; the following proof handles the general polynomial and controls the averaged remainder.

 \begin{figure*}[t]
\centering
\includegraphics[width=\textwidth]{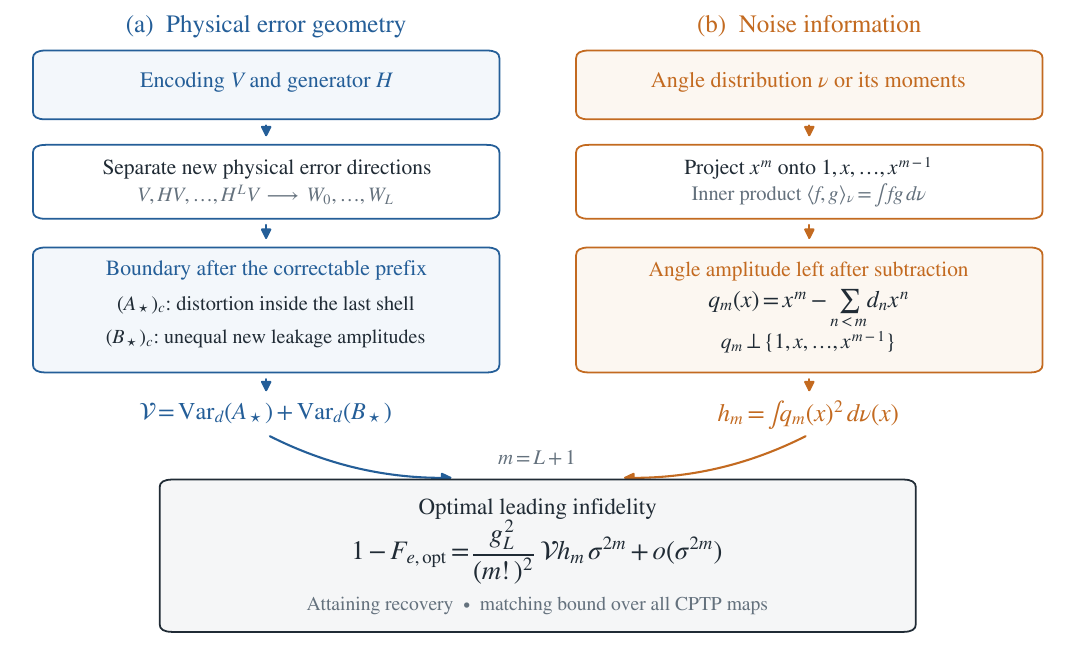}
\caption{The two projections determining the leading optimal recovery error. Physical orthogonalization separates undistorted images of the logical register. The first potentially harmful action, at degree $m=L+1$, has logical distortion $\mathcal V$ after its scalar part and reversible orientation are removed. Statistical orthogonalization subtracts lower angle powers from $x^m$, leaving a mean-square residual $h_m$. The factor $g_L$ is the product of hoppings through the correctable prefix; multiplication by $g_L^2/(m!)^2$ gives the leading optimal infidelity coefficient. The constructed recovery attains it, and the converse covers all CPTP recoveries under Proposition~\ref{thm:optimal}'s hypotheses.}
\label{fig:two-orthogonalizations}
\end{figure*}

\subsection{Proof and construction of the optimal recovery}
\label{app:optimal}\label{app:polynomial-setting}

This subsection proves Proposition~\ref{thm:optimal} and then uses the same construction to bound the worst-input error.
\begin{proof}[Proof of Proposition~\ref{thm:optimal}]
 The argument has two physical steps. The correctly recoverable shells fix a common logical frame. A coherent adjustment of their syndrome amplitudes then removes the predictable part of the first harmful angle power. We first express entanglement infidelity as a squared amplitude norm. This identity constrains every recovery and verifies the explicit construction.

The leading coefficient of $p_n$ gives
\begin{equation}
 \begin{gathered}
 H^nV=g_nW_n+\sum_{j<n}a_{jn}W_j,\\
 0\le n\le L,\qquad a_{jn}\in\mathbb R.
 \end{gathered}
 \label{eq:app-triangular-powers}
\end{equation}
Applying $H$ at $n=L$ shows that there are scalar coefficients $a_{jm}$
such that
\begin{equation}
 H^mV=\sum_{j=0}^{L}a_{jm}W_j
       +g_LW_L(A_\star)_c+g_LE_\star.
 \label{eq:app-first-uncontrolled-power}
\end{equation}
This is the sole physical expansion needed for the leading-order
optimality proof.

For the proof, let $q_n$ denote the monic noise polynomial at every $0\le n\le m$, set $h_n=\int q_n^2d\nu$, and use $\phi_n=q_n/\sqrt{h_n}$ whenever $h_n>0$. The prefix assumption makes $\phi_0,\ldots,\phi_{m-1}$ well defined.
These orthonormal functions separate components of the angle dependence under averaging with $\nu$.

The finite $2m$-th noise moment gives the
vector-valued Taylor estimate
\begin{equation}
\begin{gathered}
e^{-i\sigma xH}V=\sum_{j=0}^m\frac{(-i\sigma x)^j}{j!}H^jV+r_\sigma(x),\\
\|r_\sigma\|_{L^2(\nu;\|\cdot\|_{2,d})}=o(\sigma^m).
\end{gathered}
\label{eq:app-l2-taylor}
\end{equation}
For each fixed $x$, the remainder divided by $\sigma^m$ tends to zero.
Its norm is bounded by a constant times $|x|^m$, uniformly in small
$\sigma$. Dominated convergence applies to its square. This argument
requires neither bounded noise support nor completeness of the
polynomials in $L^2(\nu)$.

 Write $\mathcal H_{\rm log}=\mathbb C^d$. A recovery has a Kraus representation $\mathcal R(\rho)=\sum_jR_j\rho R_j^\dagger$, where each $R_j:\mathcal H\to\mathcal H_{\rm log}$ maps physical amplitudes to the logical register. Trace preservation requires $\sum_jR_j^\dagger R_j=I_{\mathcal H}$. Keeping the index coherently in an auxiliary space $\mathcal J$ gives the isometry $U=\sum_j|j\rangle\otimes R_j$; tracing that space implements the recovery. A channel from the physical space to the $d$-dimensional logical space admits a Kraus representation with at most $d\dim\mathcal H$ operators~\cite{Watrous2018}. Fix $\mathcal J$ large enough for this bound and the additional orthogonal syndrome directions used in the construction. Every recovery considered below then has an isometric realization in the same finite auxiliary space.
For a map $T:\mathcal H_{\rm log}\to\mathcal J\otimes\mathcal H_{\rm log}$, write
$T=\sum_j|j\rangle\otimes T_j$ and define
\begin{equation}
 \mathcal Q(T)=\sum_j|j\rangle\otimes(T_j)_c.
 \label{eq:app-traceless-amplitude}
\end{equation}
This is an orthogonal projection in the normalized Hilbert--Schmidt
inner product and is independent of the auxiliary basis.

To see why this projection measures error, set $T(x)=Ue^{-i\sigma xH}V$. Contracting each logical block against $|\Omega_d\rangle$ gives $ \operatorname{tr}_d(T_j(x))$, so
\begin{equation}
 F_e=\int\sum_j| \operatorname{tr}_d(T_j(x))|^2\,d\nu(x).
\end{equation}
Since $T(x)$ is an isometry, its blocks obey $\sum_j\|T_j(x)\|_{2,d}^2=1$. Separating their scalar and traceless parts gives the exact identity
\begin{align}
1-F_e(\mathcal R_U\mathcal N_{\sigma,\nu})
 &=\int\|\mathcal Q(Ue^{-i\sigma xH}V)\|_{2,d}^2d\nu(x),\notag \\
\mathcal R_U(\rho)&=\Tr_{\mathcal J}(U\rho U^\dagger). \label{eq:app-exact-error-functional}
\end{align}
Equivalently, $\mathcal Q$ removes the component proportional to
$|\Omega_d\rangle$ from the recovered logical/reference amplitude.
In particular,
$\mathcal Q(T)=0$ if and only if $T=|e\rangle\otimes I_d$ for some
auxiliary vector $e$.

\paragraph{The old shells fix the limiting logical frame.}
The optimizing recovery may depend on $\sigma$ without admitting a Taylor expansion. We therefore take a convergent subsequence of its isometric realizations. To achieve the order-$\sigma^{2m}$ target, its limit must correct every earlier shell. A surviving error at degree $k<m$ would occur at a larger order $\sigma^{2k}$. Orthogonal noise functions make that comparison precise and prevent cancellation between different powers from obscuring it.

A decoder that corrects $I,H,\ldots,H^L$ has error
$O(\sigma^{2m})$ by Eq.~\eqref{eq:app-l2-taylor}.
For any sequence $\sigma\downarrow0$, choose an optimal recovery
isometry $U_\sigma$ in the fixed auxiliary space above.
The finite-dimensional isometry space is compact, so a subsequence
converges to an isometry $U_*$. Its optimal error is bounded by
$C\sigma^{2m}$ for some fixed $C$.

For $0\le k<m$, projection onto $\phi_k$ in the noise variable gives
\begin{equation}
\begin{split}
&\int\phi_k(x)U_\sigma e^{-i\sigma xH}V\,d\nu(x)\\
&\qquad=\frac{(-i)^k\sqrt{h_k}}{k!}\sigma^kU_\sigma H^kV+o(\sigma^k).
\end{split}
\label{eq:app-noise-projection}
\end{equation}
The error functional bounds the norm of the traceless part of this
projection by $O(\sigma^m)$. Dividing by $\sigma^k$ and passing to the
subsequence limit therefore gives
$\mathcal Q(U_*H^kV)=0$.
Equation~\eqref{eq:app-triangular-powers} then implies
\begin{equation}
\begin{gathered}
U_*W_n=|e_n\rangle\otimes I_d,\\
\langle e_j|e_n\rangle=\delta_{jn},\qquad 0\le j,n\le L.
\end{gathered}
\label{eq:app-forced-prefix-decoding}
\end{equation}
The auxiliary vectors are orthonormal because $U_*$ is isometric and
the shells are orthogonal. Positivity of every $h_k$, $k<m$, is essential
to this inference.

\paragraph{The unavoidable cost of the leakage map.}
Recovery isometries cannot change the squared amplitude matrix $E_\star^\dagger E_\star=B_\star^2$. The remaining question is how much of an amplitude with this fixed matrix can become proportional to the logical identity. The following bound answers it without assuming invertibility.

Let $T:\mathcal H_{\rm log}\to\mathcal J'\otimes\mathcal H_{\rm log}$ obey
$T^\dagger T=B^2$, where $B\succeq0$, and write
$T=\sum_j|j\rangle\otimes T_j$.
The auxiliary vector of its scalar logical component is
$z_T=\sum_j|j\rangle \operatorname{tr}_d(T_j)$.
For any unit vector $z\in\mathcal J'$,
$S_z=(\langle z|\otimes I_d)T$ satisfies $S_z^\dagger S_z\preceq B^2$.
Consequently
\begin{equation}
 \|z_T\|
 =\max_{\|z\|=1}|\operatorname{tr}_d(S_z)|
 \le\max_{\|z\|=1}\operatorname{tr}_d(|S_z|)
 \le\operatorname{tr}_d(B),
 \label{eq:app-polar-trace-bound}
\end{equation}
where the last inequality uses monotonicity of the positive square
root. It follows that
\begin{equation}
 \|\mathcal Q(T)\|_{2,d}^2
 = \operatorname{tr}_d(B^2)-\|z_T\|^2
 \ge\operatorname{Var}_d(B).
 \label{eq:app-leakage-variance-bound}
\end{equation}
Polar alignment $T=|e\rangle\otimes B$ attains this inequality.
No inverse of $B$ appears, so the result includes deficient rank.

Apply this bound to $T=U_*E_\star$.
The physical orthogonality $W_n^\dagger E_\star=0$ and
Eq.~\eqref{eq:app-forced-prefix-decoding} imply that its auxiliary
support is orthogonal to every $e_n$, $n\le L$.
Its traceless part is therefore orthogonal to
$|e_L\rangle\otimes(A_\star)_c$.
Using Eq.~\eqref{eq:app-first-uncontrolled-power}, we obtain
\begin{equation}
 \big\|\mathcal Q(U_*H^mV)\big\|_{2,d}^2
 \ge g_L^2\left[
     \operatorname{Var}_d(A_\star)
    +\operatorname{Var}_d(B_\star)\right].
 \label{eq:app-boundary-converse}
\end{equation}

 If $h_m>0$, project Eq.~\eqref{eq:app-exact-error-functional} onto
$\phi_m$. Its orthogonality and normalization read
\begin{equation}
\int\phi_m x^j\,d\nu=
\begin{cases}
0,&j<m,\\
\sqrt{h_m},&j=m.
\end{cases}
\end{equation}
The Taylor expansion therefore gives
\begin{equation}
\begin{split}
&\liminf_{\sigma\downarrow0}\frac{1-F_{e,\mathrm{opt}}(\sigma,\nu)}{\sigma^{2m}}\\
&\qquad\ge\frac{h_m}{(m!)^2}\|\mathcal Q(U_*H^mV)\|_{2,d}^2\\
&\qquad\ge\frac{g_L^2h_m}{(m!)^2}
 [\operatorname{Var}_d(A_\star)+\operatorname{Var}_d(B_\star)].
\end{split}
\label{eq:app-optimal-converse}
\end{equation}
The first line is understood along a subsequence realizing the lower
limit. Every such subsequence admits the compactness argument above,
so the final bound is independent of the limiting recovery.
For $h_m=0$, nonnegativity of infidelity is the required converse.

\paragraph{An attaining recovery.}
Choose an isometry $U_0$ that decodes the old shells and polar-aligns the leakage,
\begin{equation}
 U_0W_n=|e_n\rangle\otimes I_d,\qquad
 U_0E_\star=|e_m\rangle\otimes B_\star.
\label{eq:main-polar-alignment}
\end{equation}
The new auxiliary state $e_m$ is orthogonal to the old ones. Polar decomposition constructs this map on $\operatorname{Ran}B_\star$ even at rank loss; an unused auxiliary complement completes the isometry. Define
\begin{equation}
\begin{gathered}
\mathcal D=|e_L\rangle\otimes(A_\star)_c+|e_m\rangle\otimes(B_\star)_c,\\
c=(-i)^mg_L/m!.
\end{gathered}
\label{eq:main-boundary-amplitude}
\end{equation}
Its first bad amplitude is $c\sigma^mx^m\mathcal D$. The degree-$n$ old-shell amplitude begins at $(-i)^ng_n\sigma^nx^n/n!$. A coherent output rotation of order $\sigma^{m-n}$ can therefore subtract $c d_n\sigma^mx^n\mathcal D$. This subtraction mixes syndrome amplitudes with fixed coefficients determined by the declared noise moments; it never measures or learns $x$. Pair each forward coupling with its negative adjoint to keep the rotation unitary. For $n=L$, the occupied-shell diagonal coefficient is purely imaginary, so it is already anti-Hermitian and must be included once. This phase condition is why the construction works for nonsymmetric priors too.

The resulting recovery isometry $U_\sigma=e^{G_\sigma}U_0$, in which the anti-Hermitian generator $G_\sigma$ sums these paired couplings as in Eq.~\eqref{eq:app-attaining-recovery}, has
\begin{equation}
\mathcal Q(U_\sigma e^{-i\sigma xH}V)
=c\sigma^m q_m(x)\mathcal D+o_{L^2}(\sigma^m).
\label{eq:main-attained-residual}
\end{equation}
The leading residual is the simple tensor $q_m\otimes\mathcal D$ in $L^2(\nu)\otimes\operatorname{Ran}\mathcal Q$. Here $q_m$ is the minimum-norm representative of $x^m$ modulo lower powers, while $\mathcal Q$ removes scalar logical blocks. The two orthogonal projections act on different factors and commute, so $\|q_m\otimes\mathcal D\|^2=h_m\mathcal V$. Polar alignment precedes these linear projections and minimizes the leakage distortion. The converse bounds the residual norm for every recovery, and the coherent rotation realizes the statistical projection physically. This proves Eq.~\eqref{eq:optimal-coefficient}. The reverse couplings and repeated generator actions enter at higher order. Appendix~\ref{app:attainment-details} gives every generator and proves that remainder under the stated finite $2m$-th moment. If $h_m=0$, the same construction cancels the entire order-$m$ bad amplitude. Thus the converse and attainment cover the terminal noise case without normalizing a zero polynomial.
\end{proof}

Entanglement fidelity probes the logical register through a maximally entangled reference. A coherent logical rotation by a small angle $\epsilon$ lowers it only at order $\epsilon^2$, whereas the worst-input error $D_\diamond(\Phi)=\tfrac12\|\Phi-\mathrm{id}\|_\diamond$ changes at order $\epsilon$. The attaining recovery can leave such a rotation at order $\sigma^{2m-1}$. It arises where a scalar old-shell amplitude meets a traceless amplitude of total degree $2m-1$, either directly through the moment $M_{2m-1}$ or through the coherent adjustment. Its average is a known logical Hamiltonian proportional to $(A_\star)_c$. A final logical rotation removes it. The remaining error can then be bounded at order $\sigma^{2m}$ by the same boundary cost that controls entanglement infidelity.

\begin{proposition}[Worst-input error at the first boundary]
\label{prop:worst-input}
Under the hypotheses of Proposition~\ref{thm:optimal}, there are Hermitian logical operators $K_\sigma=O(\sigma^{2m-1})$, determined by $V$, $H$ and $\nu$, such that appending $\mathcal W_\sigma(\rho)=e^{iK_\sigma}\rho e^{-iK_\sigma}$ to the attaining recovery $\mathcal R_\sigma$ constructed in the proof of Proposition~\ref{thm:optimal} preserves Eq.~\eqref{eq:optimal-coefficient} and gives
\begin{equation}
\begin{aligned}
1-F_{e,\mathrm{opt}}&\le\min_{\mathcal R}D_\diamond(\mathcal R\mathcal N_{\sigma,\nu})\\%
&\le D_\diamond(\mathcal W_\sigma\mathcal R_\sigma\mathcal N_{\sigma,\nu})\\%
&\le d\,(1-F_{e,\mathrm{opt}})+o(\sigma^{2m}).%
\end{aligned}
\label{eq:worst-input-bracket}
\end{equation}
For a reflection-symmetric prior one can take $K_\sigma=O(\sigma^{2m})$, and for $m=1$ the attaining recovery already contains the order-$\sigma$ rotation.
\end{proposition}
\begin{proof}
Write $\Phi_\sigma=\mathcal R_\sigma\mathcal N_{\sigma,\nu}$ and split the logical block of $U_\sigma e^{-i\sigma xH}V$ in syndrome direction $e_n$ into a scalar part $s_n(x)I_d$ and a traceless part $t_n(x)$. Then
\begin{gather}
\Phi_\sigma(\rho)-\rho=X_\sigma\rho+\rho X_\sigma^\dagger+\int\sum_nt_n\rho t_n^\dagger\,d\nu-\delta_\sigma\rho,\notag \\%
X_\sigma=\int\sum_n\overline{s_n}\,t_n\,d\nu, \label{eq:worst-input-expansion}\\%
\delta_\sigma=1-\int\sum_n|s_n|^2d\nu. \notag %
\end{gather}
Equation~\eqref{eq:main-attained-residual} bounds the traceless parts by $O_{L^2}(\sigma^m)$, so the quadratic term is $O(\sigma^{2m})$. The blocks form an isometry, $\sum_n(\overline{s_n}I_d+t_n^\dagger)(s_nI_d+t_n)=I_d$ for every $x$, and averaging makes $X_\sigma+X_\sigma^\dagger$ and $\delta_\sigma$ of order $\sigma^{2m}$.

The scalar amplitude in $e_n$ is $O(\sigma^n)$ for $n\le L$ and $O(\sigma^m)$ in the leakage direction. For $n\le L$, the traceless amplitude in $e_n$ is $O(\sigma^{2m-1-n})$. A nonscalar logical block first occurs in the moment of degree $2L+1$, and the reverse couplings of the coherent adjustment return the order-$\sigma^L$ shell amplitude with a factor $O(\sigma^{m-n})$. Each of these blocks is a multiple of $(A_\star)_c=g_L^{-2}(M_{2L+1})_c$. In direction $e_L$, the leading traceless amplitude $c\sigma^mq_m(x)\mathcal D$ is orthogonal in $L^2(\nu)$ to the leading scalar amplitude, which is proportional to $x^L$. Hence $X_\sigma=-i\sigma^{2m-1}K_1+O(\sigma^{2m})$, where the Hermitian $K_1$ is a real multiple of $(A_\star)_c$. The leading part of $X_\sigma$ is anti-Hermitian. Under reflection symmetry each term of $K_1$ contains an odd moment of $\nu$ or a coefficient $d_n$ with $m-n$ odd, so $K_1=0$. For $m=1$ the occupied-shell diagonal coupling of the coherent adjustment already cancels it.

Let $\widetilde\Phi_\sigma$ denote $\Phi_\sigma$ followed by the logical rotation $e^{i\sigma^{2m-1}K_1}$. Each entry of the finite-dimensional noise channel has a Taylor expansion through order $2m$ because $\mathbb E_\nu[|x|^{2m}]<\infty$. The constructed recovery is analytic in $\sigma$, since its generator is polynomial in $\sigma$. Their composition therefore has an expansion through the same order. The corrected channel satisfies $\widetilde\Phi_\sigma=\mathrm{id}+\sigma^{2m}\mathcal L+o(\sigma^{2m})$ after these cancellations. Appendix~\ref{app:attainment-details} controls the remainder. For each $t\ge0$, the powers $\widetilde\Phi_\sigma^{\,\lfloor t\sigma^{-2m}\rfloor}$ converge to $e^{t\mathcal L}$ as $\sigma\downarrow0$, so $\mathcal L$ generates completely positive trace-preserving maps. It therefore has the Gorini--Kossakowski--Sudarshan--Lindblad form~\cite{GoriniKossakowskiSudarshan1976,Lindblad1976}
\begin{equation}
\mathcal L(\rho)=-i[K_2,\rho]+\sum_k\bigl(D_k\rho D_k^\dagger-\tfrac12\{D_k^\dagger D_k,\rho\}\bigr),%
\label{eq:worst-input-generator}
\end{equation}
with traceless $D_k$. Put $K_\sigma=\sigma^{2m-1}K_1+\sigma^{2m}K_2$. These logical rotations change $F_e$ only at order $\sigma^{4m-2}$ or higher when $m\ge2$, and $K_1=0$ when $m=1$, so Eq.~\eqref{eq:optimal-coefficient} is retained.

The corrected channel is $\mathrm{id}+\sigma^{2m}\mathcal L_D+o(\sigma^{2m})$, where $\mathcal L_D$ is the dissipative part of Eq.~\eqref{eq:worst-input-generator}. Evaluating it on $|\Omega_d\rangle$ gives $1-F_e=\sigma^{2m}\sum_k\|D_k\|_{2,d}^2+o(\sigma^{2m})$. For every input $\omega$ on the memory and a reference, $\|(\mathcal L_D\otimes\mathrm{id})(\omega)\|_1\le2\|\sum_kD_k^\dagger D_k\|_\infty\le2\sum_k\Tr(D_k^\dagger D_k)=2d\sum_k\|D_k\|_{2,d}^2$, which proves the upper bound. For every channel $\Phi$, the trace distance between $(\Phi\otimes\mathrm{id})(|\Omega_d\rangle\langle\Omega_d|)$ and $|\Omega_d\rangle\langle\Omega_d|$ is at most $D_\diamond(\Phi)$ and at least $1-F_e(\Phi)$. Minimizing over recoveries gives the lower bounds.
\end{proof}

When $C_{\mathrm{opt}}>0$, Eq.~\eqref{eq:worst-input-bracket} fixes both optimized errors at order $\sigma^{2m}$; it bounds the lower and upper limits of their ratio by $1$ and $d$, without assuming that the ratio converges. When $C_{\mathrm{opt}}=0$, both errors are $o(\sigma^{2m})$, and their next nonzero orders require further analysis. For $m\ge2$, an asymmetric prior can leave a coherent term of order $\sigma^{2m-1}$ before the final rotation. For example, at $m=2$ with $\mathbb E[x]=0$ and $\mathbb E[x^2]=1$, the scalar--traceless pairings give $K_1=-g_1^2\mathbb E[x^3](A_\star)_c/6$. The third moment thus controls a known logical rotation. The fourth-order entanglement cost is proportional to $h_2=\mathbb E[x^4]-1-(\mathbb E[x^3])^2$ and can vanish.

 \subsection{Recovering the boundary from compressed moments}
\label{sec:boundary-moments}

The projection identity in Eq.~\eqref{eq:gram-residual-projection} evaluates the physical boundary without constructing a full physical basis. For the physical problem take $F=\mathsf F_L=[V\;HV\;\cdots\;H^LV]$ and $X=H^mV$. Compressed moments through degree $2m$ determine the required data,
\begin{equation}
\begin{gathered}
 \mathsf G_L=[M_{i+j}]_{i,j=0}^L,\\
 \mathsf C_L=\operatorname{col}(M_m,M_{m+1},\ldots,M_{m+L}).
\end{gathered}
 \label{eq:physical-gram}
\end{equation}
Here $\operatorname{col}$ stacks the matrix blocks vertically without transposing their entries. Then
\begin{equation}
 \mathscr S_L=M_{2m}-\mathsf C_L^\dagger\mathsf G_L^{-1}\mathsf C_L
             =g_L^2B_\star^2\succeq0.
 \label{eq:physical-schur}
\end{equation}
 The physical prefix makes $\mathsf G_L$ invertible, $P_F=\Pi_L$, and the residual map is $g_LE_\star$. For the statistical problem use the same construction in $L^2(\nu)$ with $Fu=\sum_{n<m}u_nx^n$ and $Xz=zx^m$; the scalar Schur complement is $h_m$. The matrix residual measures leakage before polar alignment and trace-centering. Recovery therefore uses the singular-amplitude spread $\operatorname{Var}_d(\sqrt{\mathscr S_L})$. The statistical residual has already completed its scalar projection and equals the squared noise norm. Since all lower moments are scalar, $(A_\star)_c=g_L^{-2}(M_{2L+1})_c$. Therefore
\begin{equation}
 \frac{g_L^2\mathcal V}{(m!)^2}
 =\frac1{(m!)^2}\left[
 \frac{\operatorname{Var}_d(M_{2L+1})}{g_L^2}
 +\operatorname{Var}_d(\sqrt{\mathscr S_L})\right].
 \label{eq:moment-coefficient}
\end{equation}
This expression uses physical moments through degree $2m$ and is invariant under logical conjugation and unitary physical realization. Its evaluation can exploit the scalar prefix before doing any logical matrix algebra. In fact, $\mathsf G_L=[\mu_{i+j}]\otimes I_d$, so a block matrix inverse is unnecessary. A still more direct route uses the already computed scalar polynomial $p_L$. Write $p_L(E)^2=\sum_{j=0}^{2L}\eta_jE^j$. The recurrence gives
\begin{equation}
\begin{aligned}
 A_\star&=g_L^{-2}\sum_{j=0}^{2L}\eta_jM_{j+1},\\
 B_\star^2&=g_L^{-2}\sum_{j=0}^{2L}\eta_jM_{j+2}
             -A_\star^2-b_L^2I_d.
\end{aligned}
\label{eq:boundary-scalar-computation}
\end{equation}
For $L=0$, take $p_0=1$, $g_0=1$ and $b_0=0$. The second identity follows by taking the squared norm of $HW_L=b_LW_{L-1}+W_LA_\star+E_\star$, whose three physical images are orthogonal. Thus scalar polynomial orthogonalization, the two potentially nonscalar moments $M_{2L+1},M_{2L+2}$, and one $d\times d$ positive square root determine the local coefficient. This evaluation uses moments only through degree $2m$ and bypasses complete spectral diagonalization and recovery optimization. For numerical work near a small residual, orthogonal factorization of the physical error images evaluates the same positive quantity while avoiding subtraction of nearly equal moment terms.

The geometry also gives a leakage-rank obstruction. With $\ell^2= \operatorname{tr}_d(B_\star^2)$, Cauchy--Schwarz implies
\begin{equation}
 \operatorname{Var}_d(B_\star)\ge(1- \operatorname{rank}E_\star/d)\ell^2.
 \label{eq:rank-bound}
\end{equation}
A nonzero leakage shell with fewer than $d$ independent directions forces positive leading loss whenever $h_m>0$. The result applies to a whole multi-qubit register without selecting a preferred logical basis.

Compressed moments therefore determine both the harmful boundary and its optimal cancellation. We next use this optimum to assess the Petz decoder while keeping the prior fixed. Section~\ref{sec:uncertainty} then changes the information available to recovery, from a known prior to partial moment or support data. Section~\ref{sec:finite-angle} treats finite-angle calibration directly, without a weak-noise expansion.

 \subsection{The leading cost of Petz recovery}
\label{sec:petz}

We continue in the weak-noise limit with the code and noise distribution fixed. The next question is whether a familiar constructive decoder already attains the optimal coefficient, or whether the coherent adjustment in Section~\ref{app:optimal} improves it. The transpose, or Petz, recovery maps output amplitudes back with the adjoints of the error maps after normalizing their overlapping images~\cite{BarnumKnill2002,NgMandayam2010}. It can be computed from the channel without solving the full recovery optimization. For entanglement fidelity at a fixed reference density operator, here $I_d/d$, its general guarantee bounds infidelity within a factor of two of the optimum; the boundary geometry will determine the precise leading gap here.

 \subsubsection{The Petz map and its leading gap}

 A Kraus representation writes a channel as $\mathcal N(\rho)=\sum_jN_j\rho N_j^\dagger$, where $N_j$ are logical-to-physical maps and $\sum_jN_j^\dagger N_j=I_d$ ensures trace preservation. The index $j$ labels a Kraus representation of the channel; it carries no assumed classical side information for recovery. Set $\Omega_{\rm out}=\sum_jN_jN_j^\dagger=\mathcal N(I_d)$. This is the output of the unnormalized logical identity. The Petz recovery uses maps $N_j^\dagger \Omega_{\rm out}^{-1/2}$. The inverse square root normalizes the occupied physical output, and $N_j^\dagger$ maps its amplitudes back to the logical space. The adjoint is not assumed to be an inverse. The inverse acts on the support of $\Omega_{\rm out}$; a trace-preserving completion handles the unused complement. We use the actual channel and the maximally mixed reference throughout this section. Designing a Petz map from an assumed distribution would instead introduce calibration mismatch.

The error overlaps form the block Gram matrix $\mathsf N^\dagger \mathsf N$ of the row of Kraus maps $\mathsf N=[N_1\ N_2\ \cdots]$. As before, scalar logical blocks carry no distortion. Let $\mathcal Q_c$ remove the normalized logical trace from every block. The exact Petz expression of Zheng et al.~\cite{ZhengEtAl2024} then turns its error into a norm of the nonscalar part of the positive square root,
\begin{equation}
 1-F_{e,\mathrm{Petz}}=\frac1d\|\mathcal Q_c|\mathsf N|\|_{\mathrm{HS}}^2,
 \quad |\mathsf N|=\sqrt{\mathsf N^\dagger \mathsf N}.
 \label{eq:petz-gram}
\end{equation}
 The identity follows because the effective Kraus maps after recovery are the blocks of $\mathsf N^\dagger \Omega_{\rm out}^{-1/2}\mathsf N=|\mathsf N|$, and $\Tr \mathsf N^\dagger \mathsf N=d$.

\begin{proposition}[Petz gap from the boundary geometry]
\label{thm:petz}
 For the channel and maximally mixed logical reference of
Proposition~\ref{thm:optimal}, define
$\kappa=g_L^2h_m/(m!)^2$,
$v_A=\operatorname{Var}_d(A_\star)$ and
$v_B=\operatorname{Var}_d(B_\star)$.
The Petz map built from that actual channel obeys
\begin{gather}
 1-F_{e,\mathrm{Petz}}
 =\kappa(2v_A+v_B)\sigma^{2m}+o(\sigma^{2m}),\label{eq:petz-coefficient}\\
 F_{e,\mathrm{opt}}-F_{e,\mathrm{Petz}}
 =\kappa v_A\sigma^{2m}+o(\sigma^{2m}) .\label{eq:petz-gap}
\end{gather}
 It pays the optimal leading leakage cost and twice the occupied-shell distortion. The statement includes singular leakage and $h_m=0$.
\end{proposition}

\begin{proof}
Section~\ref{sec:petz-proof} evaluates the two relevant blocks of the Kraus Gram square root. Appendix~\ref{app:petz-details} justifies the reduction at changing rank and controls its remainder, including singular leakage.
\end{proof}

The distinction is physical. Leakage reaches a previously unused subspace, whose orientation Petz can reverse at the optimal leading cost. A nonscalar action within an already occupied shell competes with the logical frame fixed by lower errors. Petz pays twice that contribution. The proof below identifies these costs in two blocks of the Gram square root; this also explains why expanding an inverse around a zero-noise rank-deficient matrix would miss part of the answer.

 \subsubsection{Why two singular scales matter}
\label{sec:petz-proof}

\begin{proof}[Proof of Proposition~\ref{thm:petz}]

We first consider $h_m>0$ and then handle the zero norm. Choose a real orthonormal basis of $L^2(\nu)$ containing $\phi_0,\ldots,\phi_m$. Expanding the encoded propagator defines the Kraus maps
\begin{equation}
 \begin{gathered}
 N_j=\int\phi_j(x)e^{-i\sigma xH}V\,d\nu(x),\\
 \mathsf N=[N_0\ N_1\ \cdots].
 \end{gathered}
 \label{eq:noise-basis-kraus}
\end{equation}
Parseval's identity gives $\mathcal N(\rho)=\sum_jN_j\rho N_j^\dagger$ and $\|\mathsf N\|_{\rm HS}^2=d$. Orthogonal noise functions are thus a convenient Kraus labeling, with no assumption that the label is measured. The physical range stays finite even when the noise-label space is infinite.

The proof uses the rank-safe estimate
\begin{equation}
\||X|-|Y|\|_{\rm HS}\le\sqrt2\,\|X-Y\|_{\rm HS}.
\label{eq:main-absolute-lipschitz}
\end{equation}
It follows by applying the scalar Lipschitz function $|x|$ to Hermitian dilations; Appendix~\ref{app:petz-details} gives the spectral-basis proof. Unlike an inverse-square-root expansion, this estimate remains controlled when a new Gram eigenvalue emerges from zero.

\paragraph{The two relevant scales.}

Apply the attaining isometry before discarding its syndrome. It preserves $\mathsf N^\dagger \mathsf N$, while Eq.~\eqref{eq:main-attained-residual} isolates the first nonscalar noise column. Scalar singular bases then reduce the relevant part of $\mathsf N$ to
\begin{gather}
T_\sigma=\begin{pmatrix}s_LI_d&a_\sigma A_{\rm sv}\\0&a_\sigma B_\star\end{pmatrix}, \label{eq:main-petz-block} \\
a_\sigma=\sqrt\kappa\,\sigma^m,\quad s_L\asymp\sigma^L,
\quad A_{\rm sv}=e^{i\varphi_{\rm sv}}(A_\star)_c, \notag
\end{gather}
 up to $o_{\rm HS}(a_\sigma)$ and scalar blocks. The real phase $\varphi_{\rm sv}$ is fixed by scalar singular-vector conventions and does not affect $\|A_{\rm sv}\|_{\rm HS}$. Here $s_L\asymp\sigma^L$ means their ratio is bounded above and below by positive constants. The old scalar row of degree $n$ has size $\sigma^n$; orthogonalizing it in the noise variable selects $\phi_n$. Thus the last occupied row and first residual differ by one power of $\sigma$.

The two block rows in Eq.~\eqref{eq:main-petz-block} refer to the old syndrome direction and the new leakage direction; its columns refer to two orthogonal noise amplitudes. Each entry acts on the same $d$-dimensional logical register. The upper-right block is harmful amplitude entering an already occupied direction, whereas the lower-right block enters a new direction.
 The trace part of the new leakage combines with its centered part into the positive matrix $B_\star$. Appendix~\ref{app:petz-details} proves this singular-basis reduction, including $B_\star=0$, without an assumption on later shells.

\paragraph{The two-column square root.}

 Put $\epsilon=a_\sigma/s_L$ and define the unitary matrix
\begin{equation}
 R_\epsilon=
 \exp\begin{pmatrix}0&-\epsilon A_{\rm sv}\\\epsilon A_{\rm sv}^\dagger&0\end{pmatrix}.
 \label{eq:app-triangular-rotation}
\end{equation}
 We use $R_\epsilon$ solely to estimate $|T_\sigma|$ algebraically; the Kraus representation used to define the physical channel remains fixed.
Direct multiplication gives
\begin{equation}
 (T_\sigma/s_L)R_\epsilon
 =\operatorname{diag}(I_d,\epsilon B_\star)
  +O_{\mathrm{HS}}(\epsilon^2).
 \label{eq:app-triangular-diagonalization}
\end{equation}
 Unitary conjugation gives
\begin{equation}
|T_\sigma R_\epsilon|=R_\epsilon^\dagger|T_\sigma|R_\epsilon.
\end{equation}
Together with $B_\star\succeq0$ and the Lipschitz estimate
in Eq.~\eqref{eq:main-absolute-lipschitz}, this implies
\begin{equation}
 |T_\sigma|
 =\begin{pmatrix}s_LI_d&a_\sigma A_{\rm sv}\\a_\sigma A_{\rm sv}^\dagger&a_\sigma B_\star\end{pmatrix}
  +O_{\mathrm{HS}}(a_\sigma^2/s_L).
 \label{eq:app-triangular-absolute-value}
\end{equation}
 The error is $o_{\mathrm{HS}}(a_\sigma)$. The same Lipschitz estimate controls the remainder in the reduction
to Eq.~\eqref{eq:main-petz-block}, including at rank loss.
The two off-diagonal blocks contribute
$2a_\sigma^2\operatorname{Var}_d(A_\star)$ to
Eq.~\eqref{eq:petz-gram}; the new diagonal block contributes
$a_\sigma^2\operatorname{Var}_d(B_\star)$.
This proves Eq.~\eqref{eq:petz-coefficient} for $h_m>0$.

For $h_m=0$, Eq.~\eqref{eq:optimal-coefficient} and the
Barnum--Knill inequality \cite{BarnumKnill2002},
\begin{equation}
 F_{e,\mathrm{Petz}}\ge F_{e,\mathrm{opt}}^2,
 \quad
 1-F_{e,\mathrm{Petz}}\le2(1-F_{e,\mathrm{opt}}),
 \label{eq:app-barnum-knill}
\end{equation}
give the same zero-coefficient statement.
\end{proof}
 Proposition~\ref{thm:petz} controls the leading asymptotic coefficient. At subsequent orders, square roots of emerging Gram eigenvalues can produce a correction proportional to $|\sigma|$ in the scaled infidelity, even for symmetric noise.

 Subtracting Eq.~\eqref{eq:petz-gap} from Eq.~\eqref{eq:petz-coefficient} gives $C_{\mathrm{opt}}=\kappa(v_A+v_B)$. Whenever $C_{\mathrm{opt}}>0$, the Petz-to-optimal infidelity ratio therefore tends to $1+v_A/(v_A+v_B)\in[1,2]$. Thus the geometry identifies which part realizes the factor-two bound of Ref.~\cite{BarnumKnill2002}. When $A_\star$ is scalar, Petz attains the optimal leading coefficient even if leakage leaves a nonzero error. The occupied-shell defect creates the leading gap. Leading equality is weaker than exact finite-width Petz optimality, for which Li et al. give separate conditions~\cite{LiEtAl2025}.

A physical reflection can decide this comparison before either coefficient is computed. Suppose a Hermitian involution $\Xi$ fixes the code and reverses the generator about a scalar energy $E_{\rm ref}$,
\begin{equation}
 \Xi V=V,\quad
 \Xi(H-E_{\rm ref}I)\Xi=-(H-E_{\rm ref}I).
 \label{eq:physical-reflection}
\end{equation}
Every odd compressed moment of the shifted generator vanishes. Its scalar orthogonal polynomials have definite parity, so $\Xi W_n=(-1)^nW_n$ and $W_L^\dagger(H-E_{\rm ref}I)W_L=0$. Thus $A_\star=E_{\rm ref}I_d$ and Petz attains the optimal leading coefficient for every prior satisfying Proposition~\ref{thm:optimal}. The symmetry also separates the physical Hankel matrix into even and odd blocks. It places no symmetry requirement on the angle distribution. Reflection symmetry of that distribution is a separate statistical simplification discussed below.

The singular limit matters in comparing perturbative formulas. For $H=X_s\otimes B$, $B\succeq0$, encoded at syndrome zero and with centered unit-variance noise, an emerging Gram block is $\sigma^2B^2+o(\sigma^2)$. Taking its square root first gives the coefficient $\operatorname{Var}_d(B)$. Linearizing around the shrinking scalar average instead gives $\operatorname{Var}_d(B^2)/(4 \operatorname{tr}_d(B^2))$. For $B=\operatorname{diag}(1,2)$ these are $1/4$ and $9/40$. Higher derivatives can contribute at the same order, so a fixed-positive-reference expansion cannot be used uniformly here. The exact Gram formula remains applicable.

Perturbative correctability and polar recovery have earlier foundations~\cite{Beny2011,LeungEtAl1997}. Our result evaluates a nonzero optimal coefficient and resolves its separate boundary contributions. For example, at $L=0$ with a centered unit-variance prior, put $\beta= \operatorname{tr}_d(B_\star)$ and $v=\operatorname{Var}_d(A_\star)+ \operatorname{tr}_d(B_\star^2)$. B\'eny's comparison quantity specializes to $\Delta=1-\sigma^2(v-\beta\sqrt v)+o(\sigma^2)$, with $\Delta^2\le F_{e,\mathrm{opt}}\le\Delta$ after conversion to squared fidelity. Equation~\eqref{eq:optimal-coefficient} fixes the coefficient to $v-\beta^2$, consistent with but not fixed by those bounds. Appendix~\ref{app:petz-numerics} compares the coefficient with exact Gram square roots in five finite models, including two logical qubits with noncommuting defects and singular leakage at depth two.

\section{Recovery with incomplete noise information}
\label{sec:uncertainty}

 We now keep the physical encoding and generator fixed and change the information used to choose recovery. The weak-noise optimum of Proposition~\ref{thm:optimal} used the known angle distribution. Calibration often determines a finite set of moments, a support interval, or confidence intervals on measured quantities, leaving several distributions possible. A separate decoder optimized for each compatible distribution would use information that calibration has not obtained. We must choose one operation from the common data and assess its performance on every distribution still possible.

 We return first to the weak-noise scaling $\tau=\sigma x$ with fixed rescaled priors. The recovery construction in Section~\ref{app:optimal} makes the unknown distribution enter the leading error through a finite polynomial, so this decision can be made without reconstructing a probability density. We first derive the cost of choosing an inaccurate polynomial, then minimize it over one common polynomial under uncertainty.

\subsection{Moment information and calibration error}
\label{sec:noise-moments}

Write $y_j=\mathbb E_\nu[x^j]$ for the noise moments. Their noise moment matrix, the Gram matrix $\mathsf H_y=[y_{i+j}]_{i,j=0}^{m-1}$, contains overlaps among $1,x,\ldots,x^{m-1}$, while $r_y=(y_m,\ldots,y_{2m-1})^{\mathsf T}$ contains their overlaps with $x^m$. These are statistical moments; the physical matrices $M_j$ and the resulting boundary geometry are held fixed. For $\mathsf H_y\succ0$, the Gram projection in Eq.~\eqref{eq:gram-residual-projection} gives
\begin{equation}
 u_\nu=\mathsf H_y^{-1}r_y,\qquad
 h_m=y_{2m}-r_y^{\mathsf T}\mathsf H_y^{-1}r_y.
 \label{eq:noise-schur}
\end{equation}
Here $u_\nu=(d_0,\ldots,d_{m-1})^{\mathsf T}$ is the polynomial coefficient vector.

For any fixed real coefficient vector $u\in\mathbb R^m$, construct the recovery with $q_u(x)=x^m-\sum_{n<m}u_nx^n$ in place of $q_m$. Set
\begin{equation}
 \chi=\frac{g_L^2\mathcal V}{(m!)^2},\quad
 h_m(\nu;u)=\int q_u(x)^2\,d\nu(x).
 \label{eq:chi-J}
\end{equation}
The argument after the semicolon is a chosen coefficient vector. Without it, $h_m(\nu)=\min_u h_m(\nu;u)$ denotes the optimized residual. For priors with moment vector $y$, this cost is $h_m(y;u)=h_m(\nu;u)$. Expanding the square gives $h_m(y;u)=y_{2m}-2u^{\mathsf T}r_y+u^{\mathsf T}\mathsf H_yu$, whose minimizing coefficients obey $\mathsf H_yu=r_y$.
The same physical construction gives
\begin{equation}
 1-F_e(\mathcal R_{\sigma,u}\mathcal N_{\sigma,\nu})
 =\chi h_m(\nu;u)\sigma^{2m}+o(\sigma^{2m}).
 \label{eq:fixed-polynomial-cost}
\end{equation}
At $u=u_\nu$ the coefficient equals $\chi h_m(\nu)=C_{\mathrm{opt}}$, the optimal value in Proposition~\ref{thm:optimal}. Completing the square makes the price of calibration error explicit,
\begin{equation}
\begin{gathered}
  h_m(\nu;u)=h_m(\nu)+(u-u_\nu)^{\mathsf T}\mathsf H_y(u-u_\nu),\\
 u_\nu=\mathsf H_y^{-1}r_y.
\end{gathered}
 \label{eq:mismatch}
\end{equation}
The excess leading error is quadratic in the polynomial mismatch. It is unnecessary to reconstruct a probability density if the relevant moments are known.

For a reflection-symmetric prior, uniqueness of the monic residual gives $q_m(-x)=(-1)^mq_m(x)$. The noise moment matrix $\mathsf H_y$ splits into even and odd powers; coefficients $d_n$ with $m-n$ odd vanish. In particular $d_{m-1}=0$, removing the order-$\sigma$ last-shell adjustment in Appendix~\ref{app:attainment-details}; the first possible coherent adjustment is $O(\sigma^2)$. For a reflection-invariant uncertainty class, averaging any trial residual with $(-1)^mq(-x)$ gives the same parity restriction without increasing its worst squared norm. These reductions shorten the moment optimization and the physical decoder. They leave the occupied-shell defect unchanged, so symmetric noise alone does not imply leading Petz optimality.

If every member of a nonempty uncertainty class $\mathcal U$ of priors shares moments through degree $2m$, one recovery has the same leading coefficient throughout the class. A sufficient condition for uniform control is that the $2m$-th moment remain insensitive to arbitrarily distant tails,
\begin{equation}
 \lim_{R\to\infty}\sup_{\nu\in\mathcal U}
 \int_{|x|>R}|x|^{2m}\,d\nu(x)=0.
 \label{eq:uniform-integrability}
\end{equation}
This is uniform integrability of $|x|^{2m}$. A convenient sufficient condition is
\begin{equation}
 \sup_{\nu\in\mathcal U}\int|x|^{2m+2}\,d\nu(x)<\infty.
 \label{eq:uniform-tail}
\end{equation}
The higher-moment bound gives a uniform $O(\sigma^{m+1})$ Taylor error in amplitude. The weaker condition in Eq.~\eqref{eq:uniform-integrability} gives the uniform $o(\sigma^m)$ needed for the same leading fidelity. Proposition~\ref{thm:optimal} for any fixed compatible prior gives the matching converse. The tail control is what permits the remainder bound to hold over the entire prior class.\footnote{An adverse sequence can otherwise put rare probability at rescaled angles of order $1/\sigma$, where the weak-angle expansion is inaccurate. Uniform integrability excludes a leading contribution of that kind while allowing a class whose moments above order $2m$ are infinite.}

\subsection{A least-favorable prior and a common decoder}

For one chosen recovery, the least-favorable distribution is the compatible distribution giving its smallest fidelity. A minimax recovery maximizes this smallest value. The recovery is fixed from the calibration data before the true distribution is known; its guarantee must respect that order of choices. We will show that a single least-favorable moment vector both fixes the common decoder and bounds every possible recovery.

The bound and its attainment are separate steps. For any fixed compatible prior $\nu_*$ and any recovery,
\begin{equation}
 \begin{split}
 \inf_{\nu\in\mathcal U}F_e(\mathcal R\mathcal N_{\sigma,\nu})
 &\le F_e(\mathcal R\mathcal N_{\sigma,\nu_*})\\
 &\le F_{e,\mathrm{opt}}(\sigma,\nu_*).
 \end{split}
 \label{eq:one-prior-converse}
\end{equation}
For positive $\chi$, the strongest leading bound comes from maximizing $h_m(\nu_*)$. What must still be proved is that the polynomial chosen for this hardest moment vector also controls every other compatible prior. That will turn a bound for an informed decoder into an attainable guarantee for one common decoder.

Allow the moments themselves to be uncertain. Let their vectors through degree $2m$ form a nonempty compact convex set $\mathcal Y$ of realizable moments. Assume every vector is realized by a distribution in $\mathcal U$ and Eq.~\eqref{eq:uniform-integrability} holds.
For a singular noise moment matrix the same Gram projection gives $h_m(y)=y_{2m}-r_y^{\mathsf T}\mathsf H_y^+r_y$, since realizability ensures $r_y\in\operatorname{Ran}\mathsf H_y$. The minimizing residual is unique as an element of $L^2(\nu)$ even when its coefficient vector is not.

\begin{proposition}[Optimal recovery under moment uncertainty]
\label{thm:minimax}
 Retain the nondegenerate physical prefix of Proposition~\ref{thm:optimal}. Let the noise uncertainty class $\mathcal U$ have a nonempty compact convex set $\mathcal Y$ of moment vectors through degree $2m$, with every listed vector realized in $\mathcal U$, and impose uniform integrability as in Eq.~\eqref{eq:uniform-integrability}. Then
\begin{equation}
 \mathfrak h_{\mathcal U}
 =\min_{u\in\mathbb R^m}\sup_{\nu\in\mathcal U} h_m(\nu;u)
 =\max_{\nu\in\mathcal U}h_m(\nu),
 \label{eq:moment-minimax}
\end{equation}
\begin{equation}
 \begin{split}
 &\max_{\mathcal R}\inf_{\nu\in\mathcal U}
 F_e(\mathcal R\mathcal N_{\sigma,\nu})\\
 &\qquad=1-\chi\mathfrak h_{\mathcal U}\sigma^{2m}
 +o(\sigma^{2m}).
 \end{split}
 \label{eq:quantum-minimax}
\end{equation}
No uniform positive lower bound on the noise moment matrices is required.
\end{proposition}

\begin{proof}
For realizable moments, least squares over the finite-dimensional span of lower powers is attained, even if the coefficient vector is not unique. Thus $h_m(y)=y_{2m}-r_y^{\mathsf T}\mathsf H_y^+r_y$.
For fixed $u$, $h_m(y;u)$ is affine in $y$; its minimum over $u$ is therefore concave and upper semicontinuous. It is nonnegative and has a maximizer $y_*$ on the compact convex set $\mathcal Y$.

Suppose first that $h_m(y_*)>0$. A singular $\mathsf H_{y_*}$ would imply a nonzero polynomial of degree at most $m-1$ vanishing on the support of every representing measure. That support would have at most $m-1$ points, making the degree-$m$ monic minimum zero. Hence $\mathsf H_{y_*}\succ0$. Its minimizing vector $u_*$ is unique and varies smoothly in a neighborhood of $y_*$. Differentiating the minimized quadratic along a feasible segment from $y_*$ to $y$ gives
\begin{equation}
  h_m(y;u_*)- h_m(y_*;u_*)\le0.
\end{equation}
By the envelope theorem, the derivative of $h_m$ at $y_*$ toward $y$ equals this affine change at the fixed minimizer $u_*$, and maximality of $y_*$ on the convex set $\mathcal Y$ makes it nonpositive. It follows that $\sup_yh_m(y;u_*)=h_m(y_*)$. The reverse minimax inequality follows by evaluating any polynomial against $y_*$. The fixed-polynomial recovery and uniform integrability give the achievable coefficient. Proposition~\ref{thm:optimal} for any representing distribution at $y_*$ gives the converse over all recovery channels.

If the maximum is zero,
choose a moment vector $y^\circ$ in the relative interior of $\mathcal Y$  and a minimizing monic polynomial $q^\circ$ there.
Relative interior means interior within the affine space spanned by $\mathcal Y$.
 The quantity $\int(q^\circ)^2d\nu$, evaluated with any distribution having moment vector $y$, is an affine functional of $y$. It is nonnegative on $\mathcal Y$ and zero at its relative-interior point. It must vanish everywhere on $\mathcal Y$. Every compatible distribution is therefore supported on the same zero set of $q^\circ$, with at most $m$ real points. Conversely, any common alphabet of at most $m$ points admits a monic degree-$m$ polynomial vanishing throughout it. Its fixed-polynomial recovery has error $o(\sigma^{2m})$ uniformly by Eq.~\eqref{eq:uniform-integrability}. This proves both minimax equalities and the common-alphabet characterization without a lower eigenvalue bound for all noise matrices.

For completeness, the degree-$m$ Taylor remainder satisfies
\begin{equation}
 \|r_\sigma(x)\|_{2,d}
 \le C_H\sigma^m|x|^m\min\{1,\sigma|x|\}.
 \label{eq:uniform-taylor-truncation}
\end{equation}
On $|x|\le R$, its squared norm divided by $\sigma^{2m}$ is bounded by $C_H^2\sigma^2R^2\mathbb E [|x|^{2m}]$. On the complement, Eq.~\eqref{eq:uniform-integrability} bounds it by a uniformly vanishing tail. First taking $\sigma\to0$, then $R\to\infty$, proves the uniform $o_{L^2}(\sigma^m)$. The recovery generators have fixed coefficients and their remaining terms involve powers of $x$ through degree $m$, so the completion in Appendix~\ref{app:attainment-details} has the same uniform remainder. The all-CPTP upper bound uses a single fixed least-favorable distribution, so no supremum is interchanged with a pointwise asymptotic limit.
\end{proof}

Convexity supports the general minimax equality between the common-polynomial optimum and the largest individual residual norm in Eq.~\eqref{eq:moment-minimax}; it cannot be dropped in general. Each known deterministic angle is individually invertible, but the convexified class of angles in a nondegenerate interval has a positive robust noise factor $\mathfrak h_{\mathcal U}$. It produces a positive leading error when $\chi>0$. For a fixed recovery, fidelity is affine in the channel, so convexifying this uncertainty leaves its worst value unchanged.

 The robust-QEC optimizations of Refs.~\cite{KosutShabaniLidar2008,TaghaviKosutLidar2010} select a common recovery for an uncertain channel. Equations~\eqref{eq:moment-minimax}--\eqref{eq:quantum-minimax} evaluate their weak-noise objective in terms of limited moments and give one physical recovery with a matching unrestricted bound. Partial characterization can also improve decoding by completing partially learned Pauli-error probabilities~\cite{IyerEtAl2025}; the present guarantee applies simultaneously to every compatible prior.

\paragraph{Implementing the moment optimization.}
\label{app:common-moment-minimax}\label{app:minimax}
The finite moment vector $y$ contains all the statistical information entering the leading error. Introducing a scalar $t$ writes its optimization as
\begin{equation}
 \mathfrak h_{\mathcal U}=\max_{y\in\mathcal Y,t}\ t,
 \qquad
 \begin{pmatrix}\mathsf H_y&r_y\\r_y^{\mathsf T}&y_{2m}-t\end{pmatrix}\succeq0.
 \label{eq:moment-sdp}
\end{equation}
Equation~\eqref{eq:gram-residual-projection} identifies this positive block constraint with $t\le h_m(y)$, including singular moment matrices. For an arbitrary compact convex realizable set $\mathcal Y$, Eq.~\eqref{eq:moment-sdp} is an exact convex variational expression. It becomes a semidefinite program (SDP), with a linear objective and affine positive-semidefinite constraints, when $\mathcal Y$ itself has such a finite representation. Support and measurement constraints give an important case. For support in $[-R,R]$, the exact univariate interval conditions are
\begin{equation}
\begin{gathered}
y_0=1,\quad [y_{i+j}]_{i,j=0}^m\succeq0,\\
[R^2y_{i+j}-y_{i+j+2}]_{i,j=0}^{m-1}\succeq0.
\end{gathered}
\label{eq:interval-moments}
\end{equation}
 The first inequality expresses $\mathbb E[a(x)^2]\ge0$ for every polynomial $a$ of degree at most $m$. The second expresses $\mathbb E[(R^2-x^2)b(x)^2]\ge0$ for $\deg b\le m-1$, using the support bound. In one variable these tests are also sufficient at this truncation order~\cite{KreinNudelman1977,CurtoFialkow2008}, because every nonnegative interval polynomial has the corresponding sum-of-squares representation. Thus every feasible moment vector belongs to an actual probability distribution and obeys its joint positivity constraints. Linear confidence intervals on the measured moments can be intersected with Eq.~\eqref{eq:interval-moments} and used in Eq.~\eqref{eq:moment-sdp}.

\subsection{Support information and Chebyshev recovery}

With a support bound but no probability weights, a robust decoder must control every allowed angle. A distribution can concentrate at the angle where its residual polynomial is largest. The statistical optimization therefore changes from least squares for a known measure to minimizing the maximum residual over the interval. For the recovery of Eq.~\eqref{eq:fixed-polynomial-cost}, the leading loss is $\chi\mathbb E_\nu[q(x)^2]\sigma^{2m}$. Optimizing over all distributions supported in $[-R,R]$ replaces that expectation by $\max_{|x|\le R}|q(x)|^2$, since a deterministic allowed angle attains the maximum. The common decoder therefore needs the monic degree-$m$ polynomial with smallest maximum magnitude. Equioscillation identifies it as the rescaled Chebyshev polynomial in Eq.~\eqref{eq:chebyshev}; its alternating extrema also determine a compatible prior for the converse. 

 Retain the general physical coefficient $\chi$ and suppose the statistical information is $x\in[-R,R]$, $R>0$. The least-favorable problem becomes uniform approximation by a monic polynomial. With $T_m(\cos\phi)=\cos(m\phi)$, choose
\begin{equation}
 q_*(x)=2^{1-m}R^mT_m(x/R).
 \label{eq:chebyshev}
\end{equation}
Its squared largest magnitude is $R^{2m}/4^{m-1}$, by the monic minimax property~\cite{NISTDLMFChebyshev}, and hence
\begin{equation}
 \max_{\mathcal R}\inf_{\operatorname{supp}\nu\subseteq[-R,R]}F_e
 =1-\chi\frac{R^{2m}}{4^{m-1}}\sigma^{2m}+o(\sigma^{2m}).
 \label{eq:support-optimum}
\end{equation}
A matching prior puts weights $1/(2m)$ at the endpoints and $1/m$ at the interior extrema $x_j=R\cos(j\pi/m)$. Discrete cosine orthogonality makes $q_*$ orthogonal to all lower powers, and its norm equals its squared maximum. This one prior bounds every recovery; the recovery designed with $q_*$ attains the bound uniformly over the interval.

 Alternation and extremal-design results identify the polynomial and a measure attaining its uniform error~\cite{DetteTitoff2009}. Chebyshev approximation also enters QEC sensing-code design~\cite{LaydenEtAl2019}. Here the encoding is fixed, the approximation variable is the uncertain angle, and the optimized observable is memory fidelity. The physical recovery and converse turn the approximation error into an operational quantum limit.

 For a general compact set $\mathcal K\subset\mathbb R$ with at least $m+1$ points, let $q_{\mathcal K}$ be the monic polynomial minimizing its uniform norm and let $e_{\mathcal K}=\|q_{\mathcal K}\|_{\infty,\mathcal K}$. Alternation gives ordered contact points $t_0<\cdots<t_m$ with $q_{\mathcal K}(t_j)=(-1)^{m-j}e_{\mathcal K}$. Put $P(t)=\prod_{j=0}^m(t-t_j)$ and
\begin{equation}
 w_j=\frac{|P'(t_j)|^{-1}}{\sum_{\ell=0}^m|P'(t_\ell)|^{-1}}.
\end{equation}
 All weights are positive. Lagrange interpolation gives $\sum_jt_j^n/P'(t_j)=0$ for $0\le n<m$. The signs of $P'(t_j)$ alternate with $q_{\mathcal K}(t_j)$, so this prior makes $q_{\mathcal K}$ orthogonal to lower powers and has $h_m=e_{\mathcal K}^2$. Its optimal coefficient bounds every recovery. The fixed recovery using $q_{\mathcal K}$ has error $\chi e_{\mathcal K}^2\sigma^{2m}+o(\sigma^{2m})$ uniformly over distributions supported on $\mathcal K$. This construction connects extremal approximation to a physically admissible least-favorable channel.
The gapped-support and ternary examples below specialize this compact-set result. Their scalar residuals determine the statistical part of the recovery cost, while the same physical coefficient $\chi$ remains fixed.

 More detailed support information has a direct value. For $m=2$ and $\mathcal K=[-b,-a]\cup[a,b]$, $0\le a<b$,
\begin{equation}
 q_{\mathcal K}(x)=x^2-\frac{a^2+b^2}{2},\quad
 \mathfrak h_{\mathcal K}=\frac{(b^2-a^2)^2}{4}.
 \label{eq:gapped-support}
\end{equation}
Symmetrization removes the odd coefficient, after which $x^2$ ranges over $[a^2,b^2]$. Equal weights at $-b,-a,a,b$ attain the residual norm. Relative to knowledge of the enclosing interval, the leading worst-case error is multiplied by $(1-a^2/b^2)^2$.

 For $m=2$, the same support problem gives the smallest protected case in which an informed, nominal and robust decoder can be compared exactly. Set $R=1$. Each run either has zero angle or one of two opposite angles with fixed magnitude; the probability $p$ of a nonzero rotation remains uncertain.
\begin{equation}
 \nu_p=(1-p)\delta_0+\tfrac p2(\delta_{-1}+\delta_1),
 \qquad p\in[0,1].
 \label{eq:three-angle-family}
\end{equation} Knowing $p$ gives $q_2=x^2-p$ and $h_2=p(1-p)$, using the least-squares definition also at $p=0$. A common decoder using $x^2-u_{\rm c}$ has coefficient vector $u=(u_{\rm c},0)^{\mathsf T}$ and pays
\begin{equation}
  h_2(p;u_{\rm c})=p(1-u_{\rm c})^2+(1-p)u_{\rm c}^2.
 \label{eq:ternary-polynomial-decision}
\end{equation}
 Its best worst-case choice is $u_{\rm c}=1/2$, with cost $1/4$ for every $p$ in this family. The prior $p=1/2$ proves that no smaller worst cost is possible. Since $|x^2-1/2|\le1/2$ throughout $[-1,1]$, the same decoder works for the entire support class. The least-favorable prior places weights $(1/4,1/2,1/4)$ at $(-1,0,1)$. Adding zero-angle runs makes the squared angle uncertain even as it reduces the angle variance.
A nominal decoder designed for a declared $p_0\le1/2$ uses $u_{\rm c}=p_0$. Its worst leading cost over $p\in[0,1]$ is $\chi(1-p_0)^2$, compared with $\chi/4$ for the common minimax decoder. At the nominal channel its cost is $\chi p_0(1-p_0)$. These coefficients follow by maximizing the affine function $h_2(p;p_0)$ at an endpoint; compact support makes the weak-noise remainder uniform in $p$.  For $0<p<1$ and $\chi>0$, adding zero-angle runs restores a quartic loss that vanishes for the pure binary prior. For general $R$, $q_*=x^2-R^2/2$; on a known interval $[a,b]$, first invert the midpoint rotation and take $R=(b-a)/2$. For general degree $m$, the same decision is made by the monic residual on the declared support, with the physical geometry held fixed.

 \subsection{Finite noise alphabets and their stability}
\label{sec:finite-support}

 The leading coefficient in Eq.~\eqref{eq:optimal-coefficient} can vanish for two different reasons, with different physical consequences. The generator may never distinguish logical states on the reachable space, or the allowed angles may make one harmful power cancellable even though later powers are not. A finite noise alphabet means that $x$ takes a finite list of values. On that list, a polynomial can vanish identically, removing an error amplitude that would remain for a continuous distribution. This section determines what follows beyond the zero leading coefficient and how sensitive that improvement is to rare other angles.

A vanishing physical defect decides whether the Krylov shell construction continues or terminates. Under the nondegenerate-prefix hypothesis, $\mathcal V=0$ is equivalent to exact correction through $H^{L+1}$. Indeed, the recurrence then has $A_\star=a_LI_d$ and $B_\star=b_{L+1}I_d$, so every overlap involving the next power is scalar. Conversely, those scalar overlaps force $A_\star$ and $B_\star^2$ to be scalar, hence $\mathcal V=0$. If $b_{L+1}>0$, add the normalized next shell and repeat. If $b_{L+1}=0$, the existing shells form an invariant space and discarding their syndrome corrects the entire orbit. This stopping rule separates a postponed boundary cost from exact recovery without computing further moments. Two terminal mechanisms deserve separate treatment.
Temporarily relaxing the nondegenerate-prefix assumption, suppose the moments through degree $2L$ remain scalar but the first vanishing physical Gram--Schmidt norm is $g_{n_0}=0$, with $n_0\le L$. Its null polynomial obeys $p_{n_0}(H)V=0$. All later powers reduce to the first $n_0$ powers. The whole cyclic pair is then equivalent to
\begin{equation}
 V=|0\rangle\otimes I_d,\qquad H=J_{ n_0}\otimes I_d,
 \label{eq:physical-termination}
\end{equation}
 where $J_{n_0}$ is the $n_0\times n_0$ scalar tridiagonal matrix acting on the syndrome shells. Discarding the syndrome then recovers the logical register exactly for every angle distribution. If the physical prefix remains nondegenerate but $h_m=0$, the noise has exactly $m$ atoms. There are at most $m$ because $q_m$ vanishes on the support, and at least $m$ because Proposition~\ref{thm:optimal} assumes a positive-definite noise moment matrix through degree $m-1$. This can improve the error order without making the Hamiltonian orbit exactly correctable.

Binary noise makes the distinction visible. For $x=\pm1$,
\begin{equation}
 e^{-i\sigma xH}=\cos(\sigma H)-ix\sin(\sigma H).
 \label{eq:binary-operator-closure}
\end{equation}
There are two independent functions of the angle, $1$ and $x$. Their physical coefficients still contain all even and odd powers of $H$. Correcting $I,H$ controls their leading terms without necessarily correcting $\cos(\sigma H)$ and $\sin(\sigma H)$ exactly. Thus $h_2=0$ does not imply perfect finite-width recovery.

\begin{proposition}[Recovery for a fixed finite angle alphabet]
\label{prop:finite-alphabet}
 Retain the nondegenerate physical prefix through $L$. For $r$ distinct fixed
rescaled values $x_j$ with positive weights, actual angles $\sigma x_j$, and
$2\le r\le L+1$, a polar/Petz recovery gives the bound
\begin{equation}
 1-F_{e,\mathrm{opt}}=O(\sigma^{4L-2r+6}).
 \label{eq:alphabet-order}
\end{equation}
 This bounds the error from above; its first nonzero power may be higher.
\end{proposition}
\begin{proof}
Appendix~\ref{app:alphabet-details} combines the finite-node Vandermonde scales with cancellation of the first odd Gram term, then bounds the traceless square root. Squaring this amplitude estimate gives Eq.~\eqref{eq:alphabet-order}.
\end{proof} The constant depends on the fixed nodes and weights. A family in which nodes approach one another as $\sigma$ changes requires a separate uniform estimate. If the alphabet is known but its weights are unknown, construct the decoder using synthetic uniform weights. Its nonnegative single-angle errors have a sum bounded by $r$ times the uniform-prior error, so the same order holds uniformly over all actual weights, including zero weights. A single known angle is inverted exactly.

\paragraph{Why fewer angles improve the error order.}
The $r$ angle errors have a scalar Gram spectrum of orders $1,\sigma^2,\ldots,\sigma^{2r-2}$. An invertible weighted Vandermonde matrix fixes these scales. The first possible nonscalar Gram term is of order $\sigma^{2L+1}$ and is antisymmetric in the angle labels. The smallest-eigenvalue vector approaches a real vector, so its diagonal matrix element in this odd term vanishes. That extra cancellation improves the square-root perturbation to $O(\sigma^{2L-r+3})$. Squaring its traceless norm in Eq.~\eqref{eq:petz-gram} proves Eq.~\eqref{eq:alphabet-order}. Appendix~\ref{app:alphabet-details} derives the scales and the Sylvester-equation estimate, including the limiting case $r=L+1$. Ordinary denominator counting without the antisymmetric cancellation would miss the improved order.

 For binary noise $x=\pm1$, the identity $x^2=1$ makes the degree-two noise residual vanish, $h_2=0$. When $I,H$ are correctable, the quartic boundary term in Eq.~\eqref{eq:optimal-coefficient} consequently vanishes. 
  The decomposition in Eq.~\eqref{eq:binary-operator-closure} retains all higher physical powers. A symmetry now determines finite-width recovery for a class of binary channels. Suppose a Hermitian involution $\Xi=\Xi^\dagger=\Xi^{-1}$ obeys $\Xi H\Xi=-H$ and $\Xi V=V$, and the angles are equally likely $\pm\sigma$. The Kraus maps $\cos(\sigma H)V$ and $\sin(\sigma H)V$ lie in orthogonal $\Xi$ sectors.
These are the even and odd error branches, which recovery can treat separately. Their parity does not identify the unobserved sign of the angle.

 Set $E_{\rm odd}(\sigma)=V^\dagger\sin^2(\sigma H)V$. Polar decoding in the two sectors gives
\begin{equation}
 F_{e,\mathrm{opt}}=\frac{(\Tr\sqrt{E_{\rm odd}(\sigma)})^2
 +(\Tr\sqrt{I_d-E_{\rm odd}(\sigma)})^2}{d^2}.
 \label{eq:binary-exact}
\end{equation}
The singular-value bound used in Proposition~\ref{thm:optimal} proves optimality separately in each sector.

  The three-shell model tests whether the sixth-order bound is attained with a nonzero coefficient. Set $A=0$ in Eq.~\eqref{eq:three-shell-teaching-model}, so the boundary distortion comes entirely from leakage. The involution $\Xi=\operatorname{diag}(I_d,-I_d,I_d)$
reverses $H$ and fixes the code, so Eq.~\eqref{eq:binary-exact} applies.
For this specialization its exact odd-sector weight is
\begin{equation}
E_{\rm odd}(\sigma)=b^2(b^2I_d+B^2)^{-1}
\sin^2\!\left(\sigma\sqrt{b^2I_d+B^2}\right).
\end{equation}
 Its expansion has scalar quadratic term $b^2\sigma^2I_d$ and nonscalar
quartic term $-b^2(B^2)_c\sigma^4/3$.
Consequently $(\sqrt{E_{\rm odd}(\sigma)})_c
=-b(B^2)_c\sigma^3/6+O(\sigma^5)$; the centered square root of
$I_d-E_{\rm odd}(\sigma)$ enters at higher order. The two squared centered norms give
\begin{equation}
1-F_{e,\mathrm{opt}}
=\frac{b^2}{36}\operatorname{Var}_d(B^2)\sigma^6+O(\sigma^8).
\label{eq:binary-example}
\end{equation}
For $\operatorname{Var}_d(B^2)>0$, this attains the exponent $4L-2r+6=6$ of Proposition~\ref{prop:finite-alphabet}, with $L=1$ and two angles. The same Hamiltonian has a positive quartic coefficient under centered Gaussian noise. The alphabet removes the leading statistical obstruction while the higher physical powers still cause loss.

For any positive-weight finite alphabet, exact correction at a width $\sigma$ is equivalent to $V^\dagger e^{-i\sigma(x_k-x_j)H}V\in\mathbb C I_d$ for every pair of nodes. With two distinct nodes, exactness on an interval of widths forces every compressed energy projector to be scalar, by analyticity and independence of the finite exponential frequencies. Otherwise the exact widths are isolated. A vanishing perturbative coefficient alone does not imply an interval of exact correction.

 \paragraph{Contamination changes the projection problem.}

 An exactly discrete angle distribution is an idealization. To test its usefulness, add a small probability of runs outside the original alphabet. Their amplitudes need not obey the polynomial identity responsible for the improved correction. The following projection determines the leading weak-noise coefficient with the contamination probability fixed as $\sigma\to0$. We then ask when a bound uniform in that probability permits a simultaneous limit.

\begin{proposition}[Leading cost of a fixed contaminant]
\label{prop:contamination}
 Let $\nu_\epsilon=(1-\epsilon)\nu_{\rm at}+\epsilon\nu_c$, with $0<\epsilon<1$, where $\nu_{\rm at}$ has $r\le m$ distinct fixed nodes $x_j$ with positive weights. Let $\pi_{\rm at}(x)=\prod_j(x-x_j)$, and assume the contaminant has finite $2m$-th moment and a positive-definite moment matrix through degree $m$. As $\epsilon\downarrow0$,
\begin{align}
 h_m(\nu_\epsilon)&=\epsilon C_*+O(\epsilon^2),\label{eq:contamination}\\
 C_*&=\min_{\substack{s\ \mathrm{monic}\\\deg s=m-r}}
       \int \pi_{\rm at}(x)^2s(x)^2\,d\nu_c(x)>0.\nonumber
\end{align}
\end{proposition}
 \begin{proof}
Appendix~\ref{app:contamination-details} splits the residual into a polynomial vanishing on the original nodes and its complementary component. The resulting constrained quadratic minimum gives Eq.~\eqref{eq:contamination}; the same argument with a compact-set norm proves the robust extension below.
\end{proof}

Avoiding an order-one penalty from the original atoms requires the leading residual to vanish at every $x_j$, so it has the form $q=\pi_{\rm at}s$. The added runs then contribute the norm of $s$ with weight $\pi_{\rm at}(x)^2$.  This polynomial reweighting is a Christoffel modification of the noise measure~\cite{Gautschi2004,ChowdhuryMahapatra2026Christoffel}. Here polynomial reweighting acts on the angle distribution in the statistical projection. It measures how strongly each added angle violates the cancellation identity of the original alphabet, and turns the contamination problem into another polynomial projection.

 If the contaminant is unknown, assume its support lies in a fixed compact set $\mathcal K$ with at least $m+1$ points. Let $\mathcal U_\epsilon$ contain the mixtures $(1-\epsilon)\nu_{\rm at}+\epsilon\nu_c$ over all such contaminants. For the robust noise factor $\mathfrak h_\epsilon=\mathfrak h_{\mathcal U_\epsilon}$, Proposition~\ref{thm:minimax} gives
\begin{equation}
 \mathfrak h_\epsilon
 =\epsilon\min_{\substack{s\ \mathrm{monic}\\\deg s=m-r}}
       \|\pi_{\rm at}s\|_{\infty,\mathcal K}^2+O(\epsilon^2).
 \label{eq:robust-contamination}
\end{equation}
The same node constraint organizes known-contaminant least squares and unknown-contaminant uniform approximation. Their recovery consequence takes $\sigma\downarrow0$ at each fixed $\epsilon>0$. A simultaneous limit requires a remainder uniform in the prior as this coefficient vanishes.

 Allowing small departures from the node constraint changes the minimum at order $\epsilon^2$. For a known contaminant this follows from the constrained quadratic minimum; for unknown contamination a supporting affine bound to the squared compact-set norm controls the departure. Appendix~\ref{app:contamination-details} gives the exact quadratic identity and proves both expansions even when the atomic moment matrix is singular.

For a memory correcting $I,H$, symmetric noise makes the degree-two residual norm $h_2=y_4-y_2^2=\operatorname{Var}(x^2)$. It measures fluctuations in the squared rotation magnitude. Binary angles $x=\pm1$ have $h_2=0$; Proposition~\ref{prop:finite-alphabet} then bounds the loss by $O(\sigma^6)$; the sixth-order coefficient can vanish. For a symmetric contaminant $\nu_c$ mixed with that binary prior at probability $\epsilon$, direct evaluation gives $h_2=\epsilon C_*-\epsilon^2(y_2^{(c)}-1)^2$, where $C_*=\int(x^2-1)^2d\nu_c$. Proposition~\ref{prop:contamination} gives the first-order conclusion for contaminants satisfying its hypotheses. Independent unit-Gaussian broadening of each binary peak by a rescaled width $\delta$ gives $h_2=4\delta^2+2\delta^4$.

If a particular code has binary infidelity $c_6\sigma^6+o(\sigma^6)$ with $c_6>0$ and $\chi C_*>0$, balancing it against the restored term $\chi C_*\epsilon\sigma^4$ predicts a crossover at $\epsilon$ of order $\sigma^2$. This balance is a fixed-prior heuristic for a general code. A simultaneous-limit statement requires a remainder uniform as the prior changes, and a code whose binary term vanishes can have a different scale.

\paragraph{The boundary of the one-angle argument.}
The improved finite-alphabet order uses more than the number of independent noise functions. With several fluctuating coefficients, a new polynomial degree can contain several independent directions, and the logical matrix overlaps within that space matter. A commuting two-generator example makes this limitation explicit. Take three equally weighted angle vectors at the vertices of a rotated equilateral triangle. Every quadratic function on those three points is already a linear function on the alphabet. For the six-dimensional memory constructed in Appendix~\ref{app:triangle-obstruction}, the errors $\{I,H_1,H_2\}$ are exactly correctable, yet
\begin{equation}
1-F_{e,\mathrm{opt}}=\frac{\sin^2(3\vartheta)}{32}\sigma^4+O(\sigma^6),
\label{eq:triangle-obstruction}
\end{equation}
where $\vartheta$ is the triangle's orientation relative to the two generators. At $\vartheta=0$ recovery is exact; at $\vartheta=\pi/6$ the quartic coefficient is positive. Both cases have the same polynomial dimensions. Thus a multivariate extension must retain the matrix geometry of each new degree. This example delineates the one-scalar results above while preserving their applicability to arbitrary logical dimension.

\section{Recovery from finite calibration data}
\label{sec:finite-angle}
\label{sec:compatibility-calibration}
\label{sec:calibration-bootstrap}

The weak-noise results assume a specified family of angle distributions as the width decreases. Calibration poses a further question. A small measured disturbance can arise from frequent small rotations or from rare rotations of finite size. A recovery guarantee based on that measurement must cover both possibilities. We now fix the encoder and generator and optimize one final recovery over every phase distribution compatible with the observed means. Calibration uses separate preparations drawn from the same distribution as storage; the individual storage angle remains unavailable.

The physical Krylov prefix and the calibration data constrain different parts of this problem. The prefix determines the order to which a recovery can cancel harmful amplitudes near zero angle. Positivity of the phase distribution constrains how much probability can remain away from zero. We first express the compatible channels through finitely many moments, then determine which positive statistic controls each correction order. Bounds on adverse phase probability and finite measurement precision complete the general calibration analysis. Section~\ref{sec:sector-applications} will then use conserved logical sectors to simplify recovery in the applications.

\subsection{A finite description of every compatible channel}
\label{sec:finite-calibration}

Assume that the occupied energy gaps and all measured calibration frequencies are commensurate. Choose their common phase unit and remove a scalar energy shift, so the occupied generator takes the form
\begin{equation}
 H=\sum_a n_aP_a,\quad n_a\in\mathbb Z,\quad
 U_\theta=e^{-i\theta H}.
 \label{eq:calibration-integer-spectrum}
\end{equation} Here \(P_a\) are the energy projectors of the occupied levels; terms with \(P_aV=0\) can be omitted. Writing the original generator as $H_{\rm original}=E_0I+\Lambda H$, with a real scalar shift $E_0$ and common frequency unit $\Lambda>0$, gives $\theta\equiv\Lambda\sigma x\pmod{2\pi}$. The scalar shift contributes a global phase, so it disappears from the channel. The integer normalization is used for all periodic statistics in this section; the weak-noise lift in Proposition~\ref{prop:calibration-mismatch} returns to an ordinary real angle.
 We use the code-generated cyclic space as the physical space throughout this section. Proposition~\ref{prop:equivalence} guarantees that this reduction preserves the optimum for every compatible prior. The wrapped phase \(\theta\) lies on the circle modulo \(2\pi\); the unwrapped angle $\Lambda\sigma x$ may be unbounded. Its representatives within a phase class have the same action on the encoded memory. Arbitrary calibration durations need not share this periodicity, even for a commensurate storage spectrum.

Let \(\nu\) be the unknown circle probability measure and define its complex Fourier moments by
\begin{equation}
 \widehat\nu_k=\int e^{-ik\theta}\,d\nu(\theta),\qquad
 \mathsf E_a=P_aV.
 \label{eq:calibration-fourier-moments}
\end{equation}
The encoded channel depends on these moments through
\begin{equation}
 \mathcal N_{\widehat\nu}(\rho)
 =\sum_{a,b}\widehat\nu_{n_a-n_b}\mathsf E_a\rho \mathsf E_b^\dagger.
 \label{eq:calibration-channel}
\end{equation}
Only the Fourier moments $\widehat\nu_k$ at the occupied gaps $k=n_a-n_b$ enter. Thus recovery needs finite information about the prior even when the prior itself remains undetermined.

Suppose the measured real trigonometric polynomials \(g_j^{\rm cal}(\theta)\) have means \(y_j^{\rm cal}\), for \(1\le j\le p\), and include \(g_0^{\rm cal}=1\), \(y_0^{\rm cal}=1\). Collect these means as \(\mathbf y^{\rm cal}\). Write \(D_{\rm gap}=\max_{a,b}|n_a-n_b|\) and choose \(K\) at least \(D_{\rm gap}\) and every \(\deg g_j^{\rm cal}\). For a real trigonometric polynomial \(g(\theta)=\sum_{k=-K}^K\widehat g_k e^{-ik\theta}\), its mean is the real linear functional \(L_{\widehat\nu}(g)=\sum_k\widehat g_k\widehat\nu_k\). The compatible moment vectors form
\begin{align}
 \mathfrak D_{\rm cal}=\{\,\widehat\nu:\;&\widehat\nu_0=1,\quad
 \widehat\nu_{-k}=\overline{\widehat\nu_k},\nonumber\\
 &\mathsf T_K(\widehat\nu)
   =[\widehat\nu_{i-j}]_{i,j=0}^K\succeq0,\nonumber\\
 &L_{\widehat\nu}(g_j^{\rm cal})=y_j^{\rm cal}\,\}.
 \label{eq:calibration-moment-set}
\end{align}
The Toeplitz matrix \(\mathsf T_K\), whose entries are constant along each diagonal, encodes positivity of the prior. Indeed, with \(v_K(\theta)=(1,e^{-i\theta},\ldots,e^{-iK\theta})^{\mathsf T}\), it equals \(\int v_Kv_K^\dagger\,d\nu\). In one phase variable this condition is also sufficient because every positive Toeplitz matrix in Eq.~\eqref{eq:calibration-moment-set} has a representing circle measure with at most \(K+1\) atoms. Appendix~\ref{app:calibration-moments} gives the finite Gram-space proof and its rank refinement. This is the trigonometric moment construction underlying the Carath\'eodory--Fej\'er decomposition~\cite{YangXieStoica2016}.

Equation~\eqref{eq:calibration-moment-set} is a finite moment bootstrap for the unknown angle distribution. The measured means fix some linear combinations of moments; positivity of every squared trigonometric polynomial constrains the remaining combinations. The one-circle realization just proved makes these constraints exact at bandwidth $K$, without a limiting hierarchy. Moment-optimization hierarchies~\cite{Lasserre2001} and quantum-mechanical bootstrap methods~\cite{HanHartnollKruthoff2020} use the same positivity principle on their respective unknown moments. Here it closes the compatible noise-channel set.

We assume \(\mathfrak D_{\rm cal}\ne\varnothing\). Its elements obey \(|\widehat\nu_k|\le1\), so it is compact and convex. Every feasible vector already represents an actual prior. Once the channel and measured frequencies are included, adding unmeasured higher moment orders cannot strengthen the guarantee. Additional observations or a support restriction can change it. Several independent phases require a different moment-realization analysis.

To represent the physical recovery, use logical-output \(\otimes\) physical-input vectorization,
\begin{equation}
\begin{gathered}
 |R_\alpha\rangle\!\rangle
 =\sum_{l,i}(R_\alpha)_{li}|l\rangle_L|i\rangle_P,\\
 J_{\mathcal R}=\sum_\alpha
 |R_\alpha\rangle\!\rangle\langle\!\langle R_\alpha|.
\end{gathered}
\label{eq:calibration-recovery-choi}
\end{equation}
The Choi matrix \(J_{\mathcal R}\) represents a deterministic recovery precisely when \(J_{\mathcal R}\succeq0\) and \(\operatorname{Tr}_LJ_{\mathcal R}=I_P\). Define vectors \(\mathsf b_a=\operatorname{vec}(V^\dagger P_a)\) and the positive payoff matrix
\begin{equation}
 \mathsf W(\widehat\nu)
 =\frac1{d^2}\sum_{a,b}
 \widehat\nu_{n_b-n_a}|\mathsf b_a\rangle\langle\mathsf b_b|.
 \label{eq:calibration-payoff}
\end{equation}
It acts on the same logical-output/physical-input space as \(J_{\mathcal R}\), and
\begin{equation}
 F_e(\mathcal R\mathcal N_{\widehat\nu})
 =\operatorname{Tr}[J_{\mathcal R}\mathsf W(\widehat\nu)].
 \label{eq:calibration-choi-fidelity}
\end{equation}
The reversed frequency difference in Eq.~\eqref{eq:calibration-payoff} follows from vectorizing the adjoint encoded isometry. It fixes the complex conjugations; Appendix~\ref{app:calibration-moments} checks them directly through the Kraus traces.

\begin{proposition}[Optimal recovery from compatible Fourier data]
\label{prop:finite-calibration}
Under Eqs.~\eqref{eq:calibration-integer-spectrum}--\eqref{eq:calibration-moment-set}, the largest common entanglement-fidelity guarantee is
\begin{align}
 \mathcal F_{\rm cal}(\mathbf y^{\rm cal})
 &=\max_{\substack{J_{\mathcal R}\succeq0\\
             \operatorname{Tr}_LJ_{\mathcal R}=I_P}}
       \min_{\widehat\nu\in\mathfrak D_{\rm cal}}
       \operatorname{Tr}[J_{\mathcal R}\mathsf W(\widehat\nu)]
       \label{eq:calibration-minimax-sdp} \\
 &=\min_{\substack{\widehat\nu\in\mathfrak D_{\rm cal}\\Y=Y^\dagger}}
       \{\operatorname{Tr}Y:
       I_d\otimes Y\succeq\mathsf W(\widehat\nu)\}. \notag
\end{align}
The Hermitian matrix \(Y\) acts on the physical-input cyclic space and is dual to the constraint \(\operatorname{Tr}_LJ_{\mathcal R}=I_P\). An optimal common recovery and a least-favorable compatible prior exist. The prior can be chosen with at most \(K+1\) phase atoms.
\end{proposition}
\begin{proof}
The recovery Choi set is convex and closed, and positivity together with the partial-trace constraint fixes \(\operatorname{Tr}J_{\mathcal R}=\dim\mathcal H_{\rm cyc}\), making the set bounded and hence compact. The payoff is continuous and bilinear on this set and \(\mathfrak D_{\rm cal}\). Convex minimax exchange therefore gives a saddle and permits minimizing over \(\widehat\nu\) after optimizing recovery~\cite{Sion1958}. For fixed \(\widehat\nu\), the recovery semidefinite program has strictly positive feasible point \(I_d\otimes I_P/d\). Its dual is the \(Y\)-minimization in Eq.~\eqref{eq:calibration-minimax-sdp}; strong duality and dual attainment follow from this strict feasibility. The optimized fidelity is continuous in \(\widehat\nu\), so it has a minimizer on \(\mathfrak D_{\rm cal}\). A representing atomic measure preserves its channel and every measured mean.
\end{proof}

The recovery semidefinite program and robust minimax principle have earlier QEC formulations~\cite{FletcherShorWin2007,KosutShabaniLidar2008,ElimelechAlbertBarg2026Universal}. Here the compatible prior set is fixed by the measured Fourier data. This gives an exact finite problem for arbitrary finite logical dimension, including \(d=2^k\). The matrix dimensions and the required harmonic bandwidth remain resources; finite exactness alone implies no polynomial-cost algorithm for an unrestricted growing register.

A useful way to check a proposed common decoder is a pointwise inequality. Let \(\mathcal N_\theta(\rho)=U_\theta V\rho V^\dagger U_\theta^\dagger\) be the encoded channel at a deterministic phase, and put \(f_{\mathcal R}(\theta)=F_e(\mathcal R\mathcal N_\theta)\). Real coefficients \(\lambda_j\) satisfying
\begin{equation}
 f_{\mathcal R}(\theta)-\sum_{j=0}^p\lambda_jg_j^{\rm cal}(\theta)
 =v_K(\theta)^\dagger Z v_K(\theta),\quad Z\succeq0,
 \label{eq:calibration-polynomial-witness}
\end{equation}
imply \(F_e\ge\sum_j\lambda_jy_j^{\rm cal}\) for every compatible prior. Matching Fourier coefficients turns this identity into affine constraints. Fej\'er--Riesz factorization makes such a Gram representation equivalent to nonnegativity on the circle~\cite{DritschelRovnyak2010}. An optimal decoder for one least-favorable channel can fail on another compatible channel; a saddle decoder or Eq.~\eqref{eq:calibration-polynomial-witness} checks the required common guarantee. This check is complete. Even at singular data, the best witness value equals the decoder's worst fidelity over $\mathfrak D_{\rm cal}$. When an optimal witness is attained, every adverse compatible prior is supported on the zeros of its nonnegative residual, at most $K$ phases unless that residual vanishes identically. At singular data, the physical optimum remains attained while finite optimal multipliers in this unreduced witness can fail to exist. Appendix~\ref{app:calibration-witnesses} gives an explicit example and explains the limiting witness.

The same framework decides when an extra measurement matters. Let \(\mathfrak M_{\rm cal}\) be the set of moment vectors minimizing the optimized known-channel fidelity over \(\mathfrak D_{\rm cal}\). For a further real trigonometric observable, first enlarge $K$ to contain its harmonics, if needed, and define the adverse set in that enlarged moment space. Call it $g_{p+1}^{\rm cal}$ and write $t$ for its observed value. The guarantee stays unchanged exactly when
\begin{equation}
 t\in \{L_{\widehat\nu}(g_{p+1}^{\rm cal}):\widehat\nu\in\mathfrak M_{\rm cal}\}.
 \label{eq:calibration-measurement-criterion}
\end{equation}
The minimizing set is compact and convex, so these outcomes form a closed interval. Equality means that at least one old adverse channel survives the observation; otherwise compactness makes the improvement strict. Thus an extra exact observation can exclude the channels limiting recovery, while its worst possible outcome in advance retains the old guarantee.

\paragraph{What the data determine exactly.}
\label{sec:calibration-identifiability}
Exact identification asks whether the data fix a chosen target mean for every compatible prior.
For bounded real functions \(g_1^{\rm cal},\ldots,g_p^{\rm cal}\) and a target \(g_{\rm tgt}\) on an allowed noise-parameter set,
\(\mathbb E[g_{\rm tgt}]\) is determined by \((\mathbb E[g_j^{\rm cal}])_j\) for every probability measure if and only if \(g_{\rm tgt}\in\operatorname{span}\{1,g_1^{\rm cal},\ldots,g_p^{\rm cal}\}\).
If the inclusion fails, finite-dimensional linear algebra finds a finitely
supported signed measure that annihilates \(1,g_1^{\rm cal},\ldots,g_p^{\rm cal}\) but not \(g_{\rm tgt}\).
Its positive and negative parts have equal mass. After normalization they are
two probability measures with identical observations and different target
means. This excludes any nonlinear estimator as well as a linear one.
An observation can improve the guarantee without fully identifying the channel. Equation~\eqref{eq:calibration-measurement-criterion} tests removal of all currently adverse channels; the span criterion tests exact determination of one mean.

\subsection{Calibration order and Krylov protection}
\label{sec:calibration-order}
\label{sec:calibration-prefix}
The finite moment problem treats all angles at once. We now ask which calibration statistic justifies the protection suggested by a correctable power prefix, retaining the integer-normalized generator and cyclic physical space of Section~\ref{sec:finite-calibration}. This connects the physical geometry of Section~\ref{sec:geometry} to a bound valid for every prior compatible with the data. Restricting the available information to one positive statistic makes the correction order justified by that measurement explicit. For an integer \(r\ge1\), define
\begin{align}
 s(\theta)&=\sin^2(\theta/2),\qquad
 \beta_r=\mathbb E_\nu [s(\theta)^r], \label{eq:calibration-statistic}\\
 \beta_r
 &=\frac{1}{4^{r}}\binom{2r}{r}
 +\frac{2}{4^{r}}\sum_{k=1}^r(-1)^k
       \binom{2r}{r-k}\operatorname{Re}\widehat\nu_k. \notag
\end{align}
The integrand \(s(\theta)^r\) vanishes as \(\theta^{2r}\) near zero and stays positive everywhere else on the phase circle. Its mean therefore constrains both small rotations and the amount of finite-angle probability.

For a fixed recovery, set \(\ell_{\mathcal R}(\theta)=1-f_{\mathcal R}(\theta)\). Let \(\mathfrak R_r\) be the recoveries exactly correcting the error span of \(I,H,\ldots,H^{r-1}\) on \(V\). The generators in this span have the integer normalization of Eq.~\eqref{eq:calibration-integer-spectrum}; shifting and rescaling the original Hamiltonian preserves their span.

\begin{proposition}[Periodic factorization of a correctable prefix]
\label{prop:periodic-prefix-factorization}
For \(\mathcal R\in\mathfrak R_r\), the deterministic-phase loss factors as
\begin{equation}
 \ell_{\mathcal R}(\theta)=s(\theta)^r q_{\mathcal R}(\theta),
 \qquad q_{\mathcal R}(\theta)\ge0,
 \label{eq:periodic-prefix-factorization}
\end{equation}
where \(q_{\mathcal R}\) is a trigonometric polynomial of degree at most \(D_{\rm gap}-r\). If \(r>D_{\rm gap}\), the loss vanishes identically. For every circle prior,
\begin{equation}
 1-F_e(\mathcal R\mathcal N_\nu)
 \le \|q_{\mathcal R}\|_\infty\,\beta_r.
 \label{eq:periodic-prefix-bound}
\end{equation}
\end{proposition}
\begin{proof}
Write \(B_\alpha(\theta)=R_\alpha U_\theta V\) and remove its scalar logical part,
\begin{equation}
 B_{\alpha,c}=B_\alpha-\frac{\operatorname{Tr}B_\alpha}{d}I_d.
 \label{eq:calibration-traceless-amplitude}
\end{equation}
Trace preservation gives the positive loss identity
\begin{equation}
 \ell_{\mathcal R}(\theta)
 =\frac1d\sum_\alpha\|B_{\alpha,c}(\theta)\|_{\rm HS}^2.
 \label{eq:calibration-positive-loss}
\end{equation}
Exact prefix correction is equivalent to each \(R_\alpha H^jV\) being scalar for \(0\le j<r\). Thus every traceless amplitude vanishes to order at least \(r\), and the loss vanishes to order at least \(2r\). The nonnegative loss has trigonometric degree at most \(D_{\rm gap}\), so Fej\'er--Riesz factorization writes it as \(|p(e^{i\theta})|^2\), with \(\deg p\le D_{\rm gap}\). Its zero at \(1\) requires \(p(z)=(z-1)^r a(z)\). Since \(|e^{i\theta}-1|^2=4s(\theta)\), Eq.~\eqref{eq:periodic-prefix-factorization} follows with \(q_{\mathcal R}=4^r|a(e^{i\theta})|^2\). A polynomial of degree less than \(r\) with that zero is identically zero. Averaging completes the proof.
\end{proof}

The converse is equally useful. A finite bound on \(\ell_{\mathcal R}/s^r\) implies \(\ell_{\mathcal R}=O(\theta^{2r})\). The positive sum in Eq.~\eqref{eq:calibration-positive-loss} then forces every \(B_{\alpha,c}\) to vanish to order \(r\), hence \(\mathcal R\in\mathfrak R_r\). This argument permits a singular or prematurely terminating physical shell sequence.

Now suppose the only observed statistic is \(\beta_r=\epsilon\), with no additional restriction on the prior. Allow the recovery to depend on \(\epsilon\), and define
\begin{equation}
 \mathcal E_r(\epsilon)
 =\min_{\mathcal R}\ 
   \max_{\nu:\,\mathbb E_\nu [s^r]=\epsilon}
   \mathbb E_\nu [\ell_{\mathcal R}].
 \label{eq:calibration-robust-loss}
\end{equation}
The extrema range over all final CPTP recoveries and all compatible circle probability measures.

For a specified decoder this adverse-prior problem has a simple convex geometry. Define $a_{\mathcal R,r}(z)=\max_{\theta:s(\theta)^r=z}\ell_{\mathcal R}(\theta)$ for $0\le z\le1$. Then
\begin{equation}
 \max_{\nu:\,\mathbb E_\nu [s^r]=\epsilon}\mathbb E_\nu [\ell_{\mathcal R}]
 =\operatorname{cav}a_{\mathcal R,r}(\epsilon),
 \label{eq:calibration-concave-envelope}
\end{equation}
where $\operatorname{cav}$ is the least concave function lying above the curve. Averaging phase losses takes convex combinations of the points $(z,a_{\mathcal R,r}(z))$; their upper boundary is this concave envelope. A boundary point lies on a supporting chord, so at most two phase atoms attain the maximum for this fixed decoder~\cite{Winkler1988}. A saddle prior that also bounds every other recovery can need more atoms.\footnote{For $V=I_3$, $H=\operatorname{diag}(0,1,2)$ and $\beta_1=1/2$, uniform phases $0,2\pi/3,4\pi/3$ give $\widehat\nu_1=\widehat\nu_2=0$ and optimal fidelity $1/3$. In the sector-recovery formulation of Eq.~\eqref{eq:sector-gram-sdp}, the decoder is represented by a positive Gram matrix $\mathsf X$ with unit diagonal. Any nonzero off-diagonal channel coherence permits a small feasible perturbation of $\mathsf X=I_3$ that increases the fidelity above $1/3$. Thus every least-favorable channel has $\mathsf T_2=I_3$, requiring at least three phase atoms. The two-atom statement concerns optimization against a fixed recovery.}
Taking the infimum over recoveries preserves concavity, so $\mathcal E_r(\epsilon)$ is concave. Both endpoints vanish because $\epsilon=0$ fixes phase zero and $\epsilon=1$ fixes phase $\pi$, each a known invertible rotation. The optimal loss for exact-mean data therefore need not grow monotonically with the observed mean.

\begin{proposition}[Calibration order and correctable Krylov depth]
\label{prop:calibration-prefix}
As \(\epsilon\downarrow0\),
\begin{equation}
 \mathcal E_r(\epsilon)=O(\epsilon)
 \quad\Longleftrightarrow\quad
 \mathfrak R_r\ne\varnothing.
 \label{eq:calibration-prefix-equivalence}
\end{equation}
When this prefix is correctable, the finite attained coefficient
\begin{equation}
 \Upsilon_r
 =\min_{\mathcal R\in\mathfrak R_r}
   \sup_{\theta:\,s(\theta)>0}
        \frac{\ell_{\mathcal R}(\theta)}{s(\theta)^r}
 \label{eq:calibration-sharp-coefficient}
\end{equation}
satisfies
\begin{equation}
 \mathcal E_r(\epsilon)\le\Upsilon_r\epsilon,\qquad
 \mathcal E_r(\epsilon)
 =\Upsilon_r\epsilon+o(\epsilon).
 \label{eq:calibration-sharp-asymptotic}
\end{equation}
Moreover, \(\Upsilon_r=0\) precisely when one recovery exactly corrects the whole phase orbit.
\end{proposition}
\begin{proof}
The scalar Kraus conditions defining \(\mathfrak R_r\) are equivalent to vanishing derivatives of the nonnegative loss through order \(2r-1\). These are affine closed constraints on the compact recovery Choi set. Thus \(\mathfrak R_r\) is compact. For each member, Proposition~\ref{prop:periodic-prefix-factorization} makes the supremum in Eq.~\eqref{eq:calibration-sharp-coefficient} finite. As a supremum of continuous functions of the recovery it is lower semicontinuous, so its minimum is attained. Averaging the attaining pointwise bound proves the global upper bound in Eq.~\eqref{eq:calibration-sharp-asymptotic}.

For the matching lower bound, take optimal recoveries \(\mathcal R_\epsilon\) in Eq.~\eqref{eq:calibration-robust-loss} along a sequence realizing its lower limiting ratio. Compactness gives a subsequence converging to a recovery \(\mathcal R_0\). For every fixed \(\theta_*\ne0\), the two-atom prior
\begin{equation}
 \nu_{\epsilon,\theta_*}
 =\left(1-\frac{\epsilon}{s(\theta_*)^r}\right)\delta_0
  +\frac{\epsilon}{s(\theta_*)^r}\delta_{\theta_*}
 \label{eq:calibration-rare-prior}
\end{equation}
is feasible once \(\epsilon\le s(\theta_*)^r\). Nonnegative loss gives
\begin{equation}
 \frac{\mathcal E_r(\epsilon)}{\epsilon}
 \ge
 \frac{\ell_{\mathcal R_\epsilon}(\theta_*)}{s(\theta_*)^r}.
 \label{eq:calibration-rare-lower}
\end{equation}
A finite limiting ratio therefore bounds
\(\ell_{\mathcal R_0}(\theta_*)/s(\theta_*)^r\) at every nonzero phase. The converse following Proposition~\ref{prop:periodic-prefix-factorization} puts \(\mathcal R_0\) in \(\mathfrak R_r\). Taking the supremum over \(\theta_*\) proves that this ratio is at least \(\Upsilon_r\), completing the asymptotic equality. The same argument shows that any bounded subsequence of \(\mathcal E_r(\epsilon)/\epsilon\) requires \(\mathfrak R_r\ne\varnothing\); if the prefix is absent, the ratio diverges. Finally, attainment and nonnegative loss make \(\Upsilon_r=0\) equivalent to loss zero throughout the circle.
\end{proof}

The statement compares a measured order with an exactly correctable error span. Its necessity uses the full compatible prior class in Eq.~\eqref{eq:calibration-robust-loss}; known support or further data define a more informative problem. The rare prior in Eq.~\eqref{eq:calibration-rare-prior} also changes shape with \(\epsilon\), whereas the local coefficient fixes the rescaled distribution before taking the weak-noise limit. This explains how both results coexist.

For \(r=1\), every isometric encoding has a zero-angle corrector. Unless the whole orbit is exactly correctable, calibration of the fundamental cosine mean, or first fringe, consequently gives a strictly positive linear coefficient \(\Upsilon_1\). Increasing the local correction order alone cannot turn that information into a higher-order guarantee. A single rare phase $\theta_*$ of probability $p$ lowers the optimal fidelity by $p\,\kappa_{\rm pair}(\theta_*)+o(p)$, where $\kappa_{\rm pair}$ depends only on the compression $V^\dagger U_{\theta_*}V$ and vanishes exactly when that compression is scalar; consequently $\Upsilon_1\ge\sup_{\theta:\,s(\theta)>0}\kappa_{\rm pair}(\theta)/s(\theta)$. Appendix~\ref{app:calibration-rare-pair} gives $\kappa_{\rm pair}$ in closed form, Eq.~\eqref{eq:calibration-pair-coefficient}, and identifies its zero with the Knill--Laflamme condition for the pair $\{V,U_{\theta_*}V\}$. Clader et al. analyzed how rotation-angle tails affect coherent-error suppression in specific code models~\cite{CladerEtAl2021}. Here the information constraint specifies the entire prior class against which recovery is optimized.

This information constraint connects with distributionally robust optimization. The normalized chord distance \(d_{\rm ch}(\theta,\phi)=|e^{i\theta}-e^{i\phi}|/2\) satisfies \(d_{\rm ch}(\theta,0)^2=\sin^2(\theta/2)=s(\theta)\), so \(\beta_r\) is the mean \(2r\)-th power of the distance from zero rotation. The \(2r\)-Wasserstein distance minimizes the mean \(2r\)-th power transport cost and takes its \(2r\)-th root. Transport from the point measure \(\delta_0\) has a unique coupling, hence \(\mathrm W_{2r}(\nu,\delta_0)^{2r}=\beta_r\). Thus \(\beta_r\le\epsilon\) is a Wasserstein ball. Its recovery value can differ from that for \(\beta_r=\epsilon\) at finite \(\epsilon\), while the lower bound in Eq.~\eqref{eq:calibration-sharp-asymptotic} and the prefix upper bound give the same leading coefficient for both. Transport duality, sparse adverse priors and sensitivity to rare finite excursions have earlier optimization treatments~\cite{GaoKleywegt2023,BartlEtAl2021}. The QEC content of Proposition~\ref{prop:calibration-prefix} is the exact identification of admissible vanishing orders with correctable Krylov prefixes and its consequence for the recovery justified by calibration.

The sharp coefficient has a smaller finite optimization when the forced zero is removed first. For $r\le D_{\rm gap}$ put $k_{\rm red}=D_{\rm gap}-r$. Minimize $\gamma$ over recovery Choi matrices, the coefficients of a real trigonometric polynomial $q_{\mathcal R}$ of degree at most $k_{\rm red}$, and a positive matrix $Z$, imposing
\begin{equation}
\begin{aligned}
 1-f_{\mathcal R}(\theta)&=s(\theta)^r q_{\mathcal R}(\theta),\\
 \gamma-q_{\mathcal R}(\theta)&=v_{k_{\rm red}}(\theta)^\dagger Z
                                      v_{k_{\rm red}}(\theta),\quad Z\succeq0.
\end{aligned}
\label{eq:calibration-coefficient-sdp}
\end{equation}
Matching Fourier coefficients gives affine constraints. The first identity enforces the correctable-prefix condition; positivity of physical loss makes $q_{\mathcal R}\ge0$, including at zero by continuity. Fej\'er--Riesz factorization makes the second identity equivalent to $q_{\mathcal R}\le\gamma$. Thus the attained optimum is $\Upsilon_r$ when the prefix is correctable; otherwise the constraints are infeasible. The positivity matrix has size $D_{\rm gap}-r+1$, in place of $D_{\rm gap}+1$ for a direct representation of $\gamma s^r-\ell_{\mathcal R}$. Removing the common factor also removes its forced root singularity. For $r=D_{\rm gap}$ the residual constraint is scalar. For $r>D_{\rm gap}$, prefix feasibility already implies exact correction of the whole orbit and $\Upsilon_r=0$; otherwise it is infeasible. This reduction includes degenerate physical prefixes. The distinct multiplier issue in Appendix~\ref{app:calibration-witnesses} concerns an unreduced dual for incomplete moment data.

The same factorization explains the bandwidth of the measured statistic. A nonzero nonnegative trigonometric polynomial vanishing to order $2r$ at zero is $|p(e^{i\theta})|^2$ with a root of multiplicity at least $r$ in $p$. Its harmonic degree is therefore at least $r$. The choice $s^r$ saturates this lower bound. Harmonic degree specifies the highest frequency required; experimental settings and independent phase draws are separate resources.

\subsection{Sharp fractional loss when calibration exceeds correction order}
The divergence statement above can be closed quantitatively by returning to the statistical projection. Suppose the first imperfectly correctable power is $H^m$, with the nondegenerate physical prefix of Proposition~\ref{thm:optimal} and $\mathcal V>0$. The termination argument in Section~\ref{sec:finite-support} explains the rank condition here. If the scalar physical chain terminated before degree $m$, Eq.~\eqref{eq:physical-termination} would make the whole orbit exactly correctable. A genuine first defective power therefore excludes earlier termination. A statistic of order $r>m$ vanishes faster near zero than the residual physical loss. The worst compatible small angles then have scale $\epsilon^{1/(2r)}$. Since $\beta_r\approx\mathbb E[(\theta/2)^{2r}]$, an uncertain ensemble on this scale with positive rescaled statistical residual $h_m$ has optimal local loss of order $\epsilon^{m/r}$. This exceeds the order-$\epsilon$ cost of rare finite rotations because $m<r$. The leading worst-prior cost is determined by the same moment-minimax polynomial as before, with a bound on the $2r$-th moment replacing the detailed prior.

\begin{proposition}[Sharp recovery under a mismatched calibration order]
\label{prop:calibration-mismatch}
Under these physical hypotheses and the integer spectrum convention of Eq.~\eqref{eq:calibration-integer-spectrum}, let $r>m$ be an integer and set
\begin{equation}
 \mathfrak h_{m,r}
 =\max_{\nu:\,\mathbb E_\nu [|x|^{2r}]\le1}h_m(\nu).
 \label{eq:calibration-universal-moment}
\end{equation}
This maximum over probability measures is attained and lies in $(0,1]$. With $\chi=g_{m-1}^2\mathcal V/(m!)^2$, the coefficient of Eq.~\eqref{eq:chi-J} with $g_L=g_{m-1}$, evaluated for this normalized generator,
\begin{equation}
 \mathcal E_r(\epsilon)
 =4^m\chi\mathfrak h_{m,r}\epsilon^{m/r}
   +o(\epsilon^{m/r}).
 \label{eq:calibration-mismatch-coefficient}
\end{equation}
One common polynomial recovery attains this asymptotic guarantee, with a matching converse over all CPTP recoveries. Together with Proposition~\ref{prop:calibration-prefix}, the loss therefore has order $\epsilon^{\min\{1,m/r\}}$ for every integer $r\ge1$.
\end{proposition}
\begin{proof}
The measures with $\mathbb E [|x|^{2r}]\le1$ are tight, and their moments through degree $2m$ are uniformly integrable because $r>m$. Their lower-moment image is compact, convex and realizable. Proposition~\ref{thm:minimax} therefore gives an attaining common monic polynomial $q_*$ and a maximizing prior $\nu_*$. A compactly supported prior on at least $m+1$ points has $h_m>0$, while the trial polynomial $x^m$ bounds the maximum by one. Positivity of the maximum forces the noise moment matrix through degree $m-1$ to be positive definite. Scaling $x$ shows that $\nu_*$ has $\mathbb E [|x|^{2r}]=1$; otherwise its positive $h_m$ could be increased.

For the upper bound set $\sigma_\epsilon=2\epsilon^{1/(2r)}$ and represent every wrapped phase in $[-\pi,\pi]$. Lift a compatible phase prior to $x=\theta/\sigma_\epsilon$. Since $|\theta|\le\pi|\sin(\theta/2)|$, these lifted priors satisfy
\begin{equation}
 \mathbb E [|x|^{2r}]\le(\pi/2)^{2r}.
 \label{eq:calibration-lifted-moment-bound}
\end{equation}
Every weak subsequential limit has $\mathbb E [|x|^{2r}]\le1$. On each bounded $x$ interval, the calibration identity and $\sin(\sigma_\epsilon x/2)/(\sigma_\epsilon/2)\to x$ bound the truncated limiting moment by one, and the interval can then increase to the real line. The bound in Eq.~\eqref{eq:calibration-lifted-moment-bound} also gives convergence of all moments through $2m$. Consequently the cost of the fixed polynomial $q_*$ is at most $\mathfrak h_{m,r}+o(1)$ uniformly over the lifted priors; a violating sequence would have a limiting moment vector in the unit-moment class. Their common $2r$-th moment bound gives the uniform remainder in Eq.~\eqref{eq:fixed-polynomial-cost}. Multiplying by $\chi\sigma_\epsilon^{2m}$ proves the upper coefficient.

For the converse keep $\nu_*$ fixed and take angles $\theta=t x$, wrapped onto the circle. Integer normalization preserves the channel and statistic under this wrapping. Dominated convergence gives $\mathbb E [\sin^{2r}(tx/2)]=(t/2)^{2r}+o(t^{2r})$. Continuity permits a scale $t=t_\epsilon\sim2\epsilon^{1/(2r)}$ with this mean exactly $\epsilon$. Proposition~\ref{thm:optimal} for this one compatible prior gives the matching lower coefficient. For $r\le m$ the prefix is correctable and the nonzero defect excludes exact orbit correction, so the earlier coefficient $\Upsilon_r$ is strictly positive. This proves the final order statement.
\end{proof}

The strict inequality $r>m$ is essential. At $r=m$, finite-angle probability can enter at the same leading order; the coefficient is the full-orbit quantity $\Upsilon_m$. For $r>m$, the compactness argument suppresses that contribution at leading order and returns the problem to the local Krylov boundary. The constant in Eq.~\eqref{eq:calibration-universal-moment} is optimized over realizable probability moments, as in Proposition~\ref{thm:minimax}; positivity of a truncated Hankel moment matrix for a measure on the real line alone must not replace realizability at a singular boundary.

The first two boundaries distinguish a generator whose first power already distorts the logical state ($m=1$) from a memory that corrects the first power ($m=2$). Their explicit coefficients show how the same moment constraint acts on unprotected and protected memories. 
 For $m=1$, $\mathfrak h_{1,r}=1$, attained by equal weights at $x=\pm1$. The unencoded qubit with $H=\operatorname{diag}(0,1)$ has $\chi=1/4$ and even satisfies
\begin{equation}
 \mathcal E_r(\epsilon)=\epsilon^{1/r},\qquad
 0\le\epsilon\le2^{-r}.
 \label{eq:calibration-qubit-exact}
\end{equation}
Identity recovery has loss $s(\theta)$, so Jensen's inequality gives the upper bound. A symmetric pair with $s(\theta)=\epsilon^{1/r}$ has nonnegative real dephasing factor $1-2\epsilon^{1/r}$, for which that recovery is optimal, proving equality.

For the quartic boundary $m=2$ and every integer $r>2$, define the probability $w_r=\frac{r-2}{2(r-1)}$. Then
\begin{gather}
 \mathfrak h_{2,r}=(1-w_r)w_r^{\,1-2/r},\notag \\
 \nu_*=(1-w_r)\delta_0+\frac{w_r}{2}
       \bigl(\delta_{-w_r^{-1/(2r)}}+\delta_{w_r^{-1/(2r)}}\bigr), \notag \\
 q_*(x)=x^2-w_r^{\,1-1/r}. \label{eq:calibration-quartic-factor}
\end{gather}
To verify global optimality, put $z=x^2\ge0$ and $\lambda_r=w_r^{1-2/r}/(r-1)$. The convex function $\lambda_r z^{r-1}-z+2w_r^{1-1/r}$ has minimum zero at $z=w_r^{-1/r}$. Multiplication by $z$ gives the pointwise majorant
\begin{equation}
 (z-w_r^{1-1/r})^2\le w_r^{2-2/r}+\lambda_r z^r.
 \label{eq:calibration-quartic-majorant}
\end{equation}
Its expectation is at most $\mathfrak h_{2,r}$, and the stated three-atom prior attains equality and makes $q_*$ orthogonal to lower powers. This proves the scalar saddle directly.  For general $m$, the single moment constraint has a Lagrangian dual with one multiplier,
\begin{equation}
\mathfrak h_{m,r}={}\min_{\substack{q\ \mathrm{monic}\\ \deg q=m}}\ \min_{\lambda\ge0}\Bigl[\lambda +\sup_{x\in\mathbb R}\bigl(q(x)^2-\lambda x^{2r}\bigr)\Bigr].%
\label{eq:calibration-mismatch-dual}
\end{equation}
For fixed $q$, let $\mathfrak v_q(t)$ be the largest $\mathbb E[q(x)^2]$ over priors with $\mathbb E[|x|^{2r}]\le t$. For $t>0$, tightness and uniform integrability give a finite attained maximum. Mixing priors makes $\mathfrak v_q$ concave, and increasing the moment budget makes it nondecreasing. A supporting line at the interior budget $t=1$ thus has slope $\lambda\ge0$ and obeys $\mathfrak v_q(t)\le\mathfrak v_q(1)+\lambda(t-1)$. Applying this bound to a point mass at $x$ gives $q(x)^2\le\mathfrak v_q(1)-\lambda+\lambda x^{2r}$; continuity at zero includes $x=0$. Conversely, integrating any such pointwise majorant bounds every prior in the unit-budget class. The two inequalities prove the inner minimum in Eq.~\eqref{eq:calibration-mismatch-dual}, including attainment of its multiplier. Since a monic $q$ is unbounded, every finite majorant has $\lambda>0$.

Proposition~\ref{thm:minimax} exchanges the optimizations over $q$ and $\nu$. For fixed $q$ and $\lambda$, the supremum in Eq.~\eqref{eq:calibration-mismatch-dual} is the smallest $c$ for which $q^2\le c+\lambda x^{2r}$ holds pointwise. The bracket $c+\lambda$ bounds $\mathbb E[q(x)^2]$ under the unit moment budget. Equation~\eqref{eq:calibration-quartic-majorant} is the $m=2$ majorant. Averaging a prior with its reflection preserves the constraint and cannot decrease the concave function $h_m$, so a symmetric adverse prior exists and $q$ can be taken of parity $m$. Every adverse prior is supported where the optimal majorant touches $q^2$. These contact points identify its possible angles; positive weights must also enforce normalization, the moment constraint and orthogonality of $q$ to lower powers.

The constants approach the support problem of Section~\ref{sec:uncertainty}. Lyapunov's inequality nests the classes $\mathbb E[|x|^{2r}]\le1$ as $r$ increases, so $\mathfrak h_{m,r}$ is nonincreasing. Maximizing priors are tight, and along subsequences their moments through degree $2m$ converge to those of a limit supported in $[-1,1]$. Upper semicontinuity of $h_m$ bounds the limiting value by $4^{1-m}$. Every prior on $[-1,1]$ is feasible for every $r$, giving the opposite bound. Thus $\mathfrak h_{m,r}\to4^{1-m}$, the Chebyshev value of Eq.~\eqref{eq:support-optimum} at $R=1$. This is a limit of the normalized local moment problem, after the weak-noise limit at each fixed $r$.

The next protected boundary, $m=3$, gives an exact four-angle example of the contact construction. For $r=4$, choose
\begin{gather}
q(x)=x^3-3^{-1/4}x,\qquad
\lambda_{3,4}=\frac{3^{1/4}(2-\sqrt2)}4, \notag\\
z_\pm=3^{-1/4}\left(1+\frac1{\sqrt2}\pm\sqrt{\frac12+\sqrt2}\right). \label{eq:calibration-cubic-contact-data}
\end{gather}
The squared contact radii $z_-<z_+$ are positive. Direct factorization gives
\begin{equation}
q(x)^2-\lambda_{3,4}x^8
=\frac{\lambda_{3,4}}3 -\lambda_{3,4}(x^2-z_-)^2(x^2-z_+)^2.
\label{eq:calibration-cubic-majorant}
\end{equation}
The dual is bounded by $4\lambda_{3,4}/3$, with four contacts. Choose total masses
\begin{equation}
w_+=\frac{1-z_-^4}{z_+^4-z_-^4},\qquad w_-=1-w_+,
\label{eq:calibration-cubic-contact-weights}
\end{equation}
with mass $w_\pm/2$ at each of $\pm\sqrt{z_\pm}$. Both weights are positive and their sum is one. They give $\mathbb E[x^8]=1$ and $\mathbb E[xq(x)]=0$; symmetry gives $\mathbb E[q(x)]=\mathbb E[x^2q(x)]=0$. These are all the lower-power orthogonality conditions, so $q$ is the monic residual for this prior. Its support lies at the contacts, making the integrated majorant an equality. Therefore
\begin{equation}
\mathfrak h_{3,4}=\frac{4\lambda_{3,4}}3
=\frac{2-\sqrt2}{3^{3/4}}=0.2569794\ldots.
\label{eq:calibration-cubic-exact-value}
\end{equation}
Together with the $m=1,2$ cases, this example shows how a polynomial majorant and its contact weights determine an exact recovery coefficient without optimizing over an unrestricted distribution numerically.

\subsection{How the data constrain damaging phase probability}
\label{sec:calibration-christoffel}

A code suppresses sufficiently small angles through its correctable prefix. Finite-angle alternatives can still determine a common decoder, as Eq.~\eqref{eq:calibration-rare-prior} demonstrates. The Christoffel function quantifies the largest mass compatible with the measured moments at any chosen phase. It is an extremal inverse Gram form. The Christoffel modification in Section~\ref{sec:finite-support} reweighted a measure by a polynomial square; both constructions use polynomial Gram projections, with distinct roles in recovery.

Suppose the full Toeplitz matrix \(\mathsf T_K\) is known and positive definite. Subtracting an atom of weight \(w\) at \(\theta_0\) must leave
\(\mathsf T_K-wv_K(\theta_0)v_K(\theta_0)^\dagger\succeq0\).
Conjugation by \(\mathsf T_K^{-1/2}\) gives the sharp bound
\begin{equation}
 w\le
 \left[v_K(\theta_0)^\dagger
       \mathsf T_K^{-1}v_K(\theta_0)\right]^{-1}.
 \label{eq:calibration-christoffel-atom}
\end{equation}
At equality the residual is positive Toeplitz, so its representing measure attains the bound. The denominator is the reproducing kernel evaluated at the phase, and its inverse is the Christoffel function. Christoffel functions characterize this extremal probability~\cite{KordaPutinarMezic2020,BeckermannEtAl2021}. Here the characterization bounds the probability that a prior compatible with calibration assigns to one damaging phase.

If \(\mathsf T_K\) is singular, the maximum mass 
 is zero when \(v_K(\theta_0)\) lies outside its range; inside the range, its Moore--Penrose inverse replaces the inverse in Eq.~\eqref{eq:calibration-christoffel-atom}. For partially measured moments one maximizes \(w\) jointly over all compatible completions. Choosing an arbitrary completion changes the uncertainty question.

When calibration measures individual cosine harmonics, the first two means separate second-order from fourth-order bounds on finite rotations. They are phase-referenced components of interference fringes at the fundamental frequency and its second harmonic. Their joint positivity gives the following sharp bound on the probability of damaging phases. Put
\begin{equation}
\begin{gathered}
 u=\cos\theta,\quad a_c=\mathbb E [u]=\operatorname{Re}\widehat\nu_1,\\
 S_c=\mathbb E [u^2]=\frac{1+\operatorname{Re}\widehat\nu_2}{2},\quad
 v_c=S_c-a_c^2.
\end{gathered}
\label{eq:calibration-cosine-moments}
\end{equation}

These are feasible exactly when \(a_c^2\le S_c\le1\). For a fixed \(-1\le t<a_c\), the event \(u\le t\) consists of phases separated from zero. The sharp bound is
\begin{equation}
\begin{gathered}
 \sup\Pr(u\le t)=\frac{v_c}{v_c+(a_c-t)^2},\\
 a_c+\frac{v_c}{a_c-t}\le1.
\end{gathered}
\label{eq:calibration-two-fringe-tail}
\end{equation}
It is attained by a two-point cosine distribution; it is the one-sided Chebyshev, or Cantelli, bound, an extremal moment inequality~\cite{KreinNudelman1977}. Appendix~\ref{app:calibration-interval} gives its proof, the second branch when the interval endpoint is active, and the sharp mass at an individual cosine value.

For a fixed narrow family \(\theta=\sigma x\), with \(\mathbb E [x^2]>0\) and \(\mathbb E [x^4]<\infty\),
\begin{equation}
 v_c=\frac{\mathbb E [x^4]-(\mathbb E [x^2])^2}{4}
        \sigma^4+o(\sigma^4).
 \label{eq:calibration-fringe-variance}
\end{equation}
At each fixed \(-1\le t<1\), Eq.~\eqref{eq:calibration-two-fringe-tail} applies for sufficiently small \(\sigma\). Two fringes therefore limit compatible finite-angle probability to fourth order. With the first fringe alone, its sharp upper bound is \((1-a_c)/(1-t)\), of second order for this family. For an integer-spectrum memory correcting $I,H$, the same two cosine means determine $\beta_2=(3-4a_c+\operatorname{Re}\widehat\nu_2)/8=O(\sigma^4)$. Equation~\eqref{eq:periodic-prefix-bound} therefore gives a quartic guarantee for one prefix decoder. Whether this equals the informed coefficient is code dependent.

\subsection{Finite precision and the meaning of a guarantee}
\label{sec:calibration-confidence}

Measured means carry uncertainty. Write $y_j^{\rm obs}$ for the empirical estimate of the exact calibration mean $y_j^{\rm cal}$. If simultaneous intervals
\(y_j^{\rm cal}\in[y_j^{\rm obs}-\delta_j,y_j^{\rm obs}+\delta_j]\)
have coverage at least \(1-\alpha_{\rm conf}\), replace the equalities in Eq.~\eqref{eq:calibration-moment-set} by those intervals. On their coverage event the true prior belongs to the resulting set, so the common-recovery guarantee has the same coverage. This remains true when the selected recovery depends on the data. An empty intersection signals an inconsistency between the reported intervals and the assumed model.

For the explicit witness in Eq.~\eqref{eq:calibration-polynomial-witness}, the corresponding lower bound is
\begin{equation}
 F_e\ge\lambda_0+\sum_{j=1}^p\lambda_j y_j^{\rm obs}
                  -\sum_{j=1}^p|\lambda_j|\delta_j.
 \label{eq:calibration-confidence-witness}
\end{equation}
A projection of infeasible raw means to a single feasible point does not retain this worst-case statement. The compatible-class guarantee requires an information set with justified coverage. Bias, drift between calibration and storage, or generator uncertainty must enter the information set if their effect is to be covered.

The measured observables determine the compatible class and its best recovery guarantee. More trials estimate those same means more accurately, allowing the guarantee to approach its exact-mean value.

To allocate trials, fix the witness in Eq.~\eqref{eq:calibration-confidence-witness}. Suppose observable $j$ is estimated from $N^{\rm sh}_j$ independent outcomes in an interval of length $\Delta g_j^{\rm cal}$. Choose failure allowances $\alpha_j>0$ with $\sum_j\alpha_j\le\alpha_{\rm conf}$. Hoeffding's inequality and a union bound give simultaneous half-widths $\delta_j=A_j/\sqrt{N^{\rm sh}_j}$, where $A_j=\Delta g_j^{\rm cal}\sqrt{\log(2/\alpha_j)/2}$~\cite{Hoeffding1963}. When $S_{\rm shot}=\sum_j(|\lambda_j|A_j)^{2/3}>0$, treat the counts with $|\lambda_j|A_j>0$ as positive continuous variables. Minimizing $\sum_j|\lambda_j|A_j/\sqrt{N^{\rm sh}_j}$ subject to $\sum_jN^{\rm sh}_j=N^{\rm sh}_{\rm tot}$ gives $N^{\rm sh}_j=N^{\rm sh}_{\rm tot}(|\lambda_j|A_j)^{2/3}/S_{\rm shot}$. A Lagrange multiplier gives this allocation directly; substitution gives minimum penalty $S_{\rm shot}^{3/2}/\sqrt{N^{\rm sh}_{\rm tot}}$. Thus a target penalty $\epsilon_{\rm fid}$ has sufficient count $N^{\rm sh}_{\rm tot}\ge S_{\rm shot}^3/\epsilon_{\rm fid}^2$ in this allocation model. Terms with $|\lambda_j|A_j=0$ need no additional trials for this penalty, and positive counts can be rounded up to integers. If $S_{\rm shot}=0$, the penalty already vanishes, so this fixed witness requires no additional sampling. The result optimizes sampling for this witness and the fixed failure allowances; choosing the observables or optimizing the witness changes the design problem.

The positive statistic \(\beta_r\) also has a direct probe interpretation. Prepare \(r\) auxiliary qubits in \( |+\rangle_x=(|0\rangle+|1\rangle)/\sqrt2\), apply the same phase to all of them through \(\operatorname{diag}(1,e^{-i\theta})\), and measure each in the \(\{ |+\rangle_x, |-\rangle_x\}\) basis, with \( |-\rangle_x=(|0\rangle-|1\rangle)/\sqrt2\). Conditional on \(\theta\), the all-minus probability is \(s(\theta)^r\); its mean over independent trials is \(\beta_r\). This requires a shared phase within each trial, independent readouts conditional on that phase, and the storage prior across trials. Separately drawn probe phases instead measure \(\beta_1^r\). The matched probe coupling and readout are additional physical resources whose implementation must be specified for the memory.

For \(N_{\rm trial}\) independent trials, zero all-minus events give the exact one-sided binomial upper bound
\begin{equation}
 \beta_r\le 1-\alpha_{\rm conf}^{1/N_{\rm trial}}
       \le\frac{\log(1/\alpha_{\rm conf})}{N_{\rm trial}}
 \label{eq:calibration-positive-probe}
\end{equation}
at confidence \(1-\alpha_{\rm conf}\), since the probability of a zero count is \((1-\beta_r)^{N_{\rm trial}}\). For a specified prefix decoder, multiplication by \(\|q_{\mathcal R}\|_\infty\) converts this interval into an infidelity bound; a minimizer in Eq.~\eqref{eq:calibration-sharp-coefficient} uses \(\Upsilon_r\). The resource count is \(rN_{\rm trial}\) probes and \(N_{\rm trial}\) independent phase draws. The direct positive observable simplifies the confidence interpretation. A sampling comparison with harmonic estimation must also account for its variances, correlations and experimental costs.

\section{Recovery in conserved logical sectors}
\label{sec:sector-applications}
The calibration-order bounds apply to general encodings. We now specialize the recovery optimization to memories in which each logical basis label evolves in an orthogonal conserved physical sector. This structure reduces recovery to a common alignment of the sector coherences. We first derive that reduction, then determine which partial coherence data suffice and how local calibration can miss joint-register information. The resulting decoder preserves superpositions across the logical sectors while transferring syndrome information to a disposable auxiliary system.

\subsection{One physical recovery must align all conserved sectors}
\label{sec:register-calibration}

When logical basis labels remain in orthogonal sectors, their relative phases are carried by cross-sector coherences. A positive Gram matrix expresses the requirement that one operation recover all those coherences in arbitrary logical dimension. The sector assumption therefore simplifies finite-angle recovery; the weak-noise result of Proposition~\ref{thm:optimal} applies to general encodings.

Formally, let $\mathcal S_a$ be mutually orthogonal physical sectors, one for each logical basis state $a=0,\ldots,d-1$, and put $s_a=\dim\mathcal S_a$. The channel can have any
positive block matrix $\mathsf C=[\mathsf C_{ab}]$ with
$\Tr\mathsf C_{aa}=1$, acting as
$\mathcal N(|a\rangle\langle b|)=\mathsf C_{ab}:\mathcal S_b\to\mathcal S_a$.
The uncertain-Hamiltonian model has this form when its code-cyclic space splits into $H$-invariant $\mathcal S_a$ containing $V|a\rangle$. Proposition~\ref{prop:equivalence} permits restriction to that cyclic space before the sector reduction. In that realization, define
$|\Psi_a(\tau)\rangle=e^{-i\tau H}V|a\rangle\in\mathcal S_a$.

Here $\tau$ is the full rotation angle and $\nu_\tau$ its distribution. In the weak-noise families of Section~\ref{sec:geometry}, $\nu_\tau$ is obtained from the rescaled prior by $x\mapsto\sigma x$. The sector reduction also accepts arbitrary finite-angle distributions, without a commensurability assumption. The blocks are
\begin{equation}
\begin{gathered}
 \mathsf C_{ab}=\int|\Psi_a(\tau)\rangle
                      \langle\Psi_b(\tau)|\,d\nu_\tau(\tau),\\
 \mathsf C=[\mathsf C_{ab}]\succeq0,
 \qquad \Tr\mathsf C_{aa}=1.
\end{gathered}
\label{eq:sector-coherence-matrix}
\end{equation}
The block $\mathsf C_{ab}:\mathcal S_b\to\mathcal S_a$ is the noisy
image of $|a\rangle\langle b|$. Thus the matrix retains every logical
coherence and its syndrome dependence.

For a superposition $\sum_a\psi_a|a\rangle$, the noisy physical state is $\sum_{ab}\psi_a\overline{\psi_b}\mathsf C_{ab}$. The sectors retain the basis labels, while their cross blocks carry the relative phases. Consider a decoder sending sector $a$ to $|a\rangle$ and applying an isometry $\mathsf U_a$ to its remaining degrees of freedom. Discarding the common auxiliary space gives the logical action
\begin{equation}
 \begin{gathered}
 |a\rangle\langle b|\longmapsto\zeta_{ab}|a\rangle\langle b|,\\
 \zeta_{ab}=\Tr(\mathsf U_a\mathsf C_{ab}\mathsf U_b^\dagger),\qquad
 \zeta_{aa}=1.
 \end{gathered}
 \label{eq:register-recovered-coherence}
\end{equation}
 These decoders preserve populations and aim to retain relative phases. Their entanglement fidelity is $d^{-2}\sum_{ab}\zeta_{ab}$, with conjugate terms combining to a real number. Because all $\mathsf U_a$ act into the same auxiliary space, the overlaps $\mathsf X_{ab}=\mathsf U_a^\dagger\mathsf U_b$ form one positive Gram matrix with identity diagonal blocks. Conversely, every such matrix factors into a family of these isometries. The following reduction also bounds recoveries that change logical populations.
\begin{proposition}[Recovery with conserved logical labels]
\label{prop:sector-recovery}
For a channel with the orthogonal-sector block form above, its optimal
entanglement fidelity is
\begin{equation}
 F_{e,\mathrm{opt}}(\mathsf C)
 =\frac1{d^2}
 \max_{\substack{\mathsf X\succeq0\\
                  \mathsf X_{aa}=I_{\mathcal S_a}}}
       \Tr(\mathsf X\mathsf C).
\label{eq:sector-gram-sdp}
\end{equation}
A Gram factorization of an optimizer constructs an attaining decoder.
\end{proposition}
\begin{proof} Let $R_j$ be any recovery Kraus family and retain the rows
$r_{ja}=\langle a|R_j|_{\mathcal S_a}$. These are the rows entering
\begin{equation}
 F_e=\frac1{d^2}\sum_{j,a,b}r_{ja}\mathsf C_{ab}r_{jb}^\dagger.
\label{eq:sector-recovery-rows}
\end{equation}
Their Gram matrix $\mathsf X_{ab}=\sum_jr_{ja}^\dagger r_{jb}$ is
positive with $\mathsf X_{aa}\preceq I$. Adding its positive diagonal
deficits makes the diagonal blocks identity and increases the objective,
since each $\mathsf C_{aa}$ is positive. This proves the upper bound in
Eq.~\eqref{eq:sector-gram-sdp} for arbitrary recoveries, including those
that mix logical populations.

Conversely, a feasible Gram matrix factors as
$\mathsf X_{ab}=\mathsf U_a^\dagger\mathsf U_b$, where the maps
$\mathsf U_a:\mathcal S_a\to\mathcal H_{\rm aux}$ are isometries into a common
finite auxiliary space. The physical isometry
\begin{equation}
 \mathsf U_{\rm sec}\left(\bigoplus_a x_a\right)
    =\sum_a|a\rangle\otimes\mathsf U_a x_a
\label{eq:sector-gram-recovery}
\end{equation}
followed by discarding $\mathcal H_{\rm aux}$ attains the objective.
The logical label controls the syndrome transformation coherently; it is
not measured. The auxiliary dimension can be chosen at most
$\sum_as_a$. Unused physical output directions admit an arbitrary CPTP
completion.
\end{proof}
 Here positivity of $\mathsf X$ enforces a common set of decoder isometries, while its diagonal blocks are fixed to identity. This specializes the recovery semidefinite program of
Ref.~\cite{FletcherShorWin2007}. Orthogonal logical sectors and generalized
Procrustes alignment also occur in Ref.~\cite{WuEtAl2026}.
 The optimization matrix has size $\sum_as_a$, the dimension of the reachable physical space. For the Hamiltonian model, Eq.~\eqref{eq:sector-coherence-matrix}
constructs the required block matrix directly from the angle prior.

Every feasible $\mathsf X$ has $\|\mathsf X_{ab}\|_\infty\le1$.
Independent optimization of each pair therefore gives the upper bound
\begin{equation}
 F_{\rm pair}
 =\frac{d+2\sum_{a<b}\|\mathsf C_{ab}\|_1}{d^2}
 \ge F_{e,\mathrm{opt}}.
\label{eq:sector-pair-bound}
\end{equation}
For $d=2$, there is one independent cross block. Its supported polar alignment is feasible in the two-sector Gram matrix and gives $F_{e,\mathrm{opt}}=(1+\|\mathsf C_{01}\|_1)/2$~\cite{WuEtAl2026}. The two-sector optimum is therefore fixed once its single cross block is known. For a larger register, the preferred pairwise alignments must fit into one Gram matrix. Here a pair is the two-dimensional span of two logical basis states. Optimizing all such subcodes separately need not define one operation on their superpositions.
 The benchmark $F_{\rm pair}$ is attained precisely when the pairwise choices fit one feasible decoder Gram matrix. The full positive-matrix constraint in Eq.~\eqref{eq:sector-gram-sdp} retains this compatibility for any number of logical sectors.

A useful graph test can bypass that optimization. Join distinct labels $a,b$ when $\mathsf C_{ab}\ne0$. If this graph has no cycles, every pairwise optimum can be realized simultaneously. To construct the decoder, choose any isometry at a root in each component and attach one new vertex $b$ to its parent $a$. Choose a polar contraction $\mathsf X_{ab}:\mathcal S_b\to\mathcal S_a$ attaining the pairwise trace norm, and set
\begin{equation}
 \mathsf U_b=\mathsf U_a\mathsf X_{ab}
 +\mathsf Q_b\sqrt{I_{\mathcal S_b}-\mathsf X_{ab}^\dagger\mathsf X_{ab}},
 \label{eq:sector-forest-completion}
\end{equation}
where $\mathsf Q_b$ embeds $\mathcal S_b$ isometrically into fresh auxiliary directions orthogonal to the existing images. Then $\mathsf U_b^\dagger\mathsf U_b=I$ and $\mathsf U_a^\dagger\mathsf U_b=\mathsf X_{ab}$. Unspecified nonedge overlaps have zero weight in the objective. Induction gives $F_{e,\mathrm{opt}}=F_{\rm pair}$, including rectangular or rank-deficient cross blocks. This is the acyclic case of positive-matrix completion~\cite{GroneEtAl1984}. It covers every two-sector memory and the disjoint coherence pairs of Section~\ref{sec:joint-calibration-example}. Cycles are where compatibility can impose a genuinely joint cost.

\subsection{Reflection-conditioned recovery in arbitrary dimension}

 The preceding sector reduction used the full coherence matrix. We now retain its common-decoder constraint while allowing the angle distribution to remain partly unknown. If measurements fix the real cross-sector coherences, a real Gram matrix can have the same payoff for every compatible channel. The diagonal syndrome states need not be reconstructed, because every feasible decoder gives their fixed total contribution.

Assume that reversing every angle sends $\mathsf C$ to its entrywise conjugate in fixed branch bases. Write $\check\nu_\tau$ for the reflected prior, the distribution of $-\tau$, and assume that $(\nu_\tau+\check\nu_\tau)/2$ is compatible whenever $\nu_\tau$ is compatible. The compatible class is nonempty. Also assume the observations determine every $\operatorname{Re}\mathsf C_{ab}$ for $a\ne b$. The following result applies to this additional symmetry and measurement class. Let $\mathsf C_{\rm re}$ have those real off-diagonal blocks
and zero diagonal blocks.
The diagonal blocks contribute the constant $d$ to every feasible recovery
objective, and positivity is imposed on the decoder Gram
matrix $\mathsf X$.
\begin{proposition}[Recovery from real sector coherences]
\label{prop:sector-data}
Under these reflection and observation hypotheses, the unrestricted minimax fidelity is
\begin{equation}
\begin{split}
 &\max_{\mathcal R}\inf_{\nu_\tau\ {\rm compatible}}
                 F_e(\mathcal R\mathcal N_{\nu_\tau})\\
 &\quad=\frac{d+
      \max_{\mathsf X\succeq0,\,\mathsf X_{aa}=I}
                   \Tr(\mathsf X\mathsf C_{\rm re})}{d^2}.
\end{split}
\label{eq:sector-data-minimax}
\end{equation}
\end{proposition}
\begin{proof}
 The full matrix $\mathsf C_{\rm re}$ is real symmetric because its $ab$ and $ba$ blocks are transposes; an individual cross block need not be symmetric. Averaging an optimizer with its conjugate therefore makes $\mathsf X$ real without changing the objective. Its payoff on every
compatible channel is $d+\Tr(\mathsf X\mathsf C_{\rm re})$; the
diagonal contribution is fixed and the imaginary part does not contribute.
Equation~\eqref{eq:sector-gram-recovery} realizes this common recovery.
A compatible symmetrized prior attains exactly the same unrestricted
optimum, proving the converse.
\end{proof}
Using this reduction for a particular register requires a measurement scheme
that identifies its real cross blocks. This identification is a substantive requirement, since it fixes the part of the compatible channel needed for this common alignment.
 Even observations that characterize each qubit separately may miss joint coherences, because those coherences can depend on different angle harmonics. The two-qubit register in Section~\ref{sec:joint-calibration-example} 
 isolates that missing information.


\subsection{Identical local calibration can hide different register fidelities}
\label{sec:joint-calibration-example}

Local calibration can omit coherences that determine recovery of a register. When the same uncertain angle rotates several qubits, a joint coherence can depend on an angle harmonic that no individual qubit probes at that duration. Complete local channel data may then leave the joint channel undetermined. To isolate this issue from redundant encoding and conflicting syndrome alignments, take $V=I_d$ and a diagonal generator $H=\sum_aE_a|a\rangle\langle a|$. Averaging its phases gives
\begin{gather}
 \mathcal N_\Gamma(|a\rangle\langle b|)
       =\Gamma_{ab}|a\rangle\langle b|,\label{eq:register-schur}\\
 \Gamma_{ab}=\varphi_{\nu_\tau}(E_a-E_b),
 \quad \Gamma\succeq0,\quad\Gamma_{aa}=1. \notag
\end{gather}
Here $\varphi_{\nu_\tau}(\omega)=\int e^{-i\omega\tau}\,d\nu_\tau(\tau)$ is the characteristic function of Section~\ref{sec:model}, written for the full angle. This entrywise action is called a Schur channel. Its recovery problem is the
$s_a=1$ case of Eq.~\eqref{eq:sector-gram-sdp},
\begin{equation}
 F_{e,\mathrm{opt}}(\mathcal N_\Gamma)
 =\frac1{d^2}\max_{\mathsf X\succeq0,\,\mathsf X_{aa}=1}
                         \Tr(\Gamma\mathsf X).
\label{eq:register-correlation-sdp}
\end{equation}
If every $\Gamma_{ab}$ is real and nonnegative,
$|\mathsf X_{ab}|\le1$ makes identity recovery optimal, with value
$d^{-2}\sum_{ab}\Gamma_{ab}$.

Two qubits form the smallest register in which identical local channels can hide different joint coherence. For $H=Z_1+Z_2$, single-qubit coherences probe energy gap $2$, while $|00\rangle\langle11|$ probes gap $4$. Comparing those harmonics tests which recovery information local calibration misses. 
 The energies are $(2,0,0,-2)$ in the computational basis; angles are full rotation angles $\tau$ in $e^{-i\tau H}$, with no factor $1/2$. Compare the priors
\begin{equation}
\begin{split}
 \nu_A&=\tfrac12(\delta_0+\delta_{\pi/2}),\\
 \nu_B&=\tfrac14(\delta_0+\delta_{\pi/4}
                         +\delta_{\pi/2}+\delta_{3\pi/4}).
\end{split}
\label{eq:register-two-priors}
\end{equation}
These two full-angle priors both have $\varphi_{\nu_\tau}(2)=0$, so every one-qubit input experiences complete
$Z$ dephasing at this storage time. Their local process tomographies coincide.
The joint harmonic differs, with $\varphi_{\nu_A}(4)=1$ and
$\varphi_{\nu_B}(4)=0$. Writing $\gamma_{\rm J}=\varphi_{\nu_\tau}(4)$ gives
\begin{equation}
 \Gamma( \gamma_{\rm J})=
 \begin{pmatrix}
 1&0&0& \gamma_{\rm J}\\
 0&1&1&0\\
 0&1&1&0\\
  \overline\gamma_{\rm J}&0&0&1
 \end{pmatrix}.
\label{eq:register-harmonic-matrix}
\end{equation}

Both priors preserve the coherence between $|01\rangle$ and $|10\rangle$, which have the same energy. Prior $A$ also preserves $|00\rangle\langle11|$. Its channel is $\rho\mapsto(\rho+Z_1Z_2\rho Z_1Z_2)/2$, preserving two parity subspaces of dimensions $(2,2)$. Prior $B$ instead separates the energy subspaces of dimensions $(1,2,1)$. Dephasing into subspaces of dimensions $d_j$ has entanglement fidelity $\sum_jd_j^2/d^2$, as follows from the projector Kraus maps. The nonnegative entries of both matrices make identity recovery optimal by the preceding bound. Consequently
\begin{equation}
 F_{e,\mathrm{opt}}(\nu_A)=\frac12,
 \qquad F_{e,\mathrm{opt}}(\nu_B)=\frac38.
\label{eq:register-known-optima}
\end{equation}
 Correlated dephasing can preserve joint coherences absent from local channels~\cite{PalmaSuominenEkert1996}. Here that distinction changes the exact optimal register fidelity while leaving local calibration identical.

If the decoder knows $H$ and $\varphi_{\nu_\tau}(2)=0$, pinching, meaning removal of coherences between the
three energy sectors, with projectors $|00\rangle\langle00|$,
$|01\rangle\langle01|+|10\rangle\langle10|$ and
$|11\rangle\langle11|$, attains fidelity $3/8$ for every compatible
complex $\gamma_{\rm J}$. The prior $\nu_B$ gives the matching all-CPTP upper bound.
Hence
\begin{equation}
 \max_{\mathcal R}\inf_{\nu_\tau:\varphi_{\nu_\tau}(2)=0}
 F_e(\mathcal R\mathcal N_{\nu_\tau})=\frac38.
\label{eq:register-local-minimax}
\end{equation}
The missing information is the coherence between $|00\rangle$ and
$|11\rangle$. A joint coherence measurement determines its harmonic;
local tomography at twice the storage time can also determine it if the
same quasi-static distribution persists. The obstruction concerns local
data at the specified time. It does not require an additional noise
generator or a new encoding.

The same example has an exact answer for an entire complex data class. Continue to take \(V=I_4\), \(H=Z_1+Z_2\), and \(\varphi_{\nu_\tau}(2)=0\), with \(\gamma_{\rm J}=\varphi_{\nu_\tau}(4)\). The moment matrix on these three frequencies is
\begin{equation}
 \begin{pmatrix}1&0& \overline\gamma_{\rm J}\\0&1&0\\ \gamma_{\rm J}&0&1\end{pmatrix}.
 \label{eq:joint-calibration-disk-moment}
\end{equation}
Its eigenvalues are \(1,1\pm|\gamma_{\rm J}|\), so every \(|\gamma_{\rm J}|\le1\) is compatible with a circle prior. The data therefore leave a disk of possible joint coherences despite fixing the individual dephasing channels.

Applying Proposition~\ref{prop:sector-recovery} to the two disjoint coherence pairs and using Eq.~\eqref{eq:sector-forest-completion} gives an explicit test of how extra joint data improve the guarantee.
Suppose further observations restrict \(\gamma_{\rm J}\) to a nonempty closed convex set \(\mathcal C_{\rm J}\) in that disk. Then the exact common-recovery value is
\begin{equation}
 F_{e,\rm joint}^{\rm rob}
 =\frac{3+\operatorname{dist}(0, \mathcal C_{\rm J})}{8}.
 \label{eq:joint-calibration-disk-optimum}
\end{equation}
This formula permits complex observations and requires no reflection symmetry of the prior.

The two disjoint coherence pairs of $\Gamma(\gamma_{\rm J})$ form a forest. The all-recovery bound and simultaneous polar alignment in Eq.~\eqref{eq:sector-forest-completion} therefore give the known-channel optimum $(4+2+2|\gamma_{\rm J}|)/16=(3+|\gamma_{\rm J}|)/8$. The terms count the four diagonal contributions and the two aligned pairs. For these one-dimensional sectors a diagonal unitary attains the alignment.

Let \(\gamma_*\) be the point of \(\mathcal C_{\rm J}\) nearest the origin. If \(\gamma_*\ne0\), the first-order condition for projection onto a convex set gives
\begin{equation}
 \frac{\operatorname{Re}(\overline\gamma_*\gamma_{\rm J})}{|\gamma_*|}
 \ge |\gamma_*|,
 \qquad \gamma_{\rm J}\in\mathcal C_{\rm J}.
\end{equation}
Choose the fixed diagonal recovery aligning the outer pair with \(\gamma_*\) and leaving the central pair aligned. Its score is
\([3+\operatorname{Re}(\overline\gamma_*\gamma_{\rm J})/|\gamma_*|]/8\),
which is at least the right side of Eq.~\eqref{eq:joint-calibration-disk-optimum} throughout the class. The prior representing \(\gamma_*\) gives the matching all-recovery upper bound. If \(0\in\mathcal C_{\rm J}\), pinching between the energy sectors keeps the central pair and attains \(3/8\); \(\gamma_{\rm J}=0\) gives the converse.

For instance, the extra observation \(\operatorname{Re}\gamma_{\rm J}\ge\alpha\), with \(0\le\alpha\le1\), gives \(F_{e,\rm joint}^{\rm rob}=(3+\alpha)/8\). The added information is useful precisely when it excludes the origin. The example isolates a joint-calibration limitation in two unencoded logical qubits; the general framework applies separately to encoded registers of arbitrary finite dimension.
Encoded memories can add nontrivial syndrome spaces to these sectors. Under the same sector assumptions, the common-alignment construction retains their full logical register.

\section{Optical recovery from incomplete population data}
\label{sec:optical-application}

 The two-sector recovery formula of Proposition~\ref{prop:sector-recovery} and the real-data reduction of Proposition~\ref{prop:sector-data} turn population calibration into a common-decoder problem for a concrete memory. We now keep the optical encoding fixed and ask how specified population measurements improve its recovery guarantee at a chosen storage time. One logical qubit is stored in four photons shared by two optical modes. A fluctuating rotation mixes the modes and changes their relative phase while conserving the total photon number. Separate calibration copies experience the same distribution of rotation rates as the memory. Their probability of leaving the code subspace constrains the angle distribution and hence the best recovery guarantee. We first determine the exact optimum from three population means, then solve the two- and one-duration information problems. This order identifies what the extra measurements contribute while keeping the same physical memory and available recovery operations. We then derive the population readout and its precision, quantify sensitivity to the calibrated axis, and evaluate the declared-family and finite-alphabet mechanisms of Section~\ref{sec:uncertainty}. The finite-width benchmark and uniform optical crossover in Appendices~\ref{app:benchmark-details} and~\ref{app:joint-details} complete those comparisons.

\subsection{Physical model and the recovery observable}

We use the four-photon code of Chuang, Leung and Yamamoto~\cite{ChuangLeungYamamoto1997},
\begin{equation}
 |0\rangle_L=\frac{|40\rangle+|04\rangle}{\sqrt2},\qquad
 |1\rangle_L=|22\rangle.
 \label{eq:opt-code}
\end{equation}
The entries in a Fock ket count photons in modes $a,b$. Thus $|40\rangle$ means four photons in the first mode and none in the second. The unknown stored state is $\psi_0 |0\rangle_L+\psi_1 |1\rangle_L$, with $|\psi_0|^2+|\psi_1|^2=1$. Its code occupies two dimensions of the five-dimensional sector with total photon number four. The remaining dimensions are available for error motion without immediately identifying the logical state. Define
\begin{gather}
 J_x=(a^\dagger b+b^\dagger a)/2,\quad
 J_y=(a^\dagger b-b^\dagger a)/(2i), \notag \\
 J_z=(a^\dagger a-b^\dagger b)/2.
\end{gather}
 Here $a,b$ also denote the mode annihilation operators. The first two generators mix the modes coherently; the third changes their relative phase. They obey the angular-momentum algebra and act as spin two in this fixed-photon-number sector. We first choose a calibrated generator for which the first error shell is correctable. Section~\ref{sec:optical-axis-calibration} determines which other axes retain the correctable first shell.

\paragraph{The selected axis and its closed cyclic space.}

We now fix the dimensionless generator
\begin{equation}
 X_{\rm opt}=\frac{J_x+J_y+J_z}{\sqrt3}.
 \label{eq:opt-axis}
\end{equation}
 In each experimental run the Hamiltonian is \(H_\omega=\omega X_{\rm opt}\). Here
\(\omega\in\mathbb R\) is the signed rate of the noise-induced rotation. It is constant during that run and independent of the encoded input.
Its distribution \(\nu_\omega\) is common to the memory experiment and independently
prepared calibration ensembles. At storage time \(T>0\), the uncertain angle is
\(\tau=\omega T\). The exact finite-angle formulas below require no density, symmetry, moment or weak-noise assumption on
\(\nu_\omega\).
The encoded channel at this storage time is
\(\mathcal N_{\nu_\omega}(\rho)=\int e^{-i\omega TX_{\rm opt}}V\rho V^\dagger e^{i\omega TX_{\rm opt}}\,d\nu_\omega(\omega)\).

 For the encoding isometry
\(V\) in Eq.~\eqref{eq:opt-code}, the complete code-generated space consists
of the two orthogonal blocks \(W_0=V\) and \(W_1=X_{\rm opt}V/\sqrt2\). In this space,
\begin{equation}
 (W_0,W_1)^\dagger X_{\rm opt}(W_0,W_1)
 =\begin{pmatrix}0&\sqrt2 I_2\\ \sqrt2 I_2&\sigma_y\end{pmatrix}.
 \label{eq:opt-blocks}
\end{equation}
Here $\sigma_y$ acts on the logical input of Eq.~\eqref{eq:opt-code}. Block zero is the code, and block one is its orthogonal image $X_{\rm opt}V/\sqrt2$. Each block therefore carries a complete logical copy. The hopping $\sqrt2 I_2$ moves both logical amplitudes in the same way; their first departure from the code can be placed entirely in a syndrome label. The $\sigma_y$ action in block one distinguishes logical directions and is the first obstruction to continuing this separation.

 The identity $X_{\rm opt}^2V=2V+X_{\rm opt}V\sigma_y$, checked directly with the Schwinger generators on the two codewords, proves closure of these four dimensions. Together with $V^\dagger X_{\rm opt}V=0$ and $V^\dagger X_{\rm opt}^2V=2I_2$, it makes the first shell correctable. In the shell and boundary notation of Eqs.~\eqref{eq:shells}--\eqref{eq:physical-defect}, the terminal block has $L=1$, $m=2$, $g_1=\sqrt2$, $A_\star=\sigma_y$, $B_\star=0$ and $\mathcal V=1$, so $\chi=1/2$. Here $\sigma_y$ is the logical Pauli matrix and $\sigma$ remains the angle scale. The first harmful amplitude is thus an occupied-shell distortion. Proposition~\ref{thm:petz} gives the factor-two leading Petz gap because the defect is entirely in $A_\star$. Proposition~\ref{thm:optimal} gives optimal quartic coefficient $h_2/2$, equal to one for unit-variance Gaussian noise. Closure also makes the following finite-angle computation exact.

The eigenstates of $\sigma_y$, \(|\pm\rangle=(|0\rangle\pm i|1\rangle)/\sqrt2\),
therefore generate the conditional syndrome Hamiltonians
\begin{equation}
 \mathsf H_\pm=\begin{pmatrix}0&\sqrt2\\\sqrt2&\pm1\end{pmatrix}.
 \label{eq:opt-conditional}
\end{equation}
Only one of these blocks needs exponentiation. With $Z_{\rm syn}=\operatorname{diag}(1,-1)$,
\begin{equation}
 \mathsf H_-=-Z_{\rm syn}\mathsf H_+Z_{\rm syn},\quad
 u_-(\tau)=Z_{\rm syn}\overline{u_+(\tau)},
 \label{eq:opt-branch-symmetry}
\end{equation}
where $u_\pm(\tau)=e^{-i\tau\mathsf H_\pm}|0\rangle_{\rm syn}$. The second identity follows because $\mathsf H_+$ is real and $Z_{\rm syn}|0\rangle=|0\rangle$. This symmetry exchanges the two logical branches; it differs from the code-fixing reflection in Eq.~\eqref{eq:physical-reflection}.

 The identification $W_n|\psi\rangle\leftrightarrow|n\rangle_{\rm syn}\otimes|\psi\rangle$ is a change of coordinates on the four reachable physical dimensions. The syndrome label specifies the shell within the four-dimensional cyclic space. In these coordinates an arbitrary superposition of the conserved logical branches, with logical amplitudes $\psi_-$ and $\psi_+$, evolves as
\begin{align}
 &\psi_-|0\rangle_{\rm syn}|-\rangle
 +\psi_+|0\rangle_{\rm syn}|+\rangle \label{eq:optical-branch-superposition}\\
 &\quad\longmapsto
 \psi_-|u_-(\tau)\rangle_{\rm syn}|-\rangle
 +\psi_+|u_+(\tau)\rangle_{\rm syn}|+\rangle. \notag
\end{align}
The two-component vectors $u_\pm$ are syndrome amplitudes in the basis of code and terminal shell. The logical populations stay fixed, but the two logical branches can become entangled with different syndrome states. Discarding the syndrome then reduces logical coherence according to the inner product of these vectors. Evaluating the propagator for $\mathsf H_+$ from $|0\rangle$ and using Eq.~\eqref{eq:opt-branch-symmetry} gives
\begin{equation}
 u_\pm(\tau)=e^{\mp i\tau/2}
 \begin{pmatrix}
 \cos(3\tau/2)\pm \frac{i}{3}\sin(3\tau/2)\\
 -\frac{2i\sqrt2}{3}\sin(3\tau/2)
 \end{pmatrix}.
 \label{eq:opt-amplitudes}
\end{equation}
The probability of occupying shell one is $|(u_\pm)_1|^2$, the same in both logical branches. It therefore contains information about the angle without revealing $\psi_-$ or $\psi_+$. We estimate this population on separate calibration copies; the attaining recovery keeps the stored shell amplitudes coherent through its alignment. Explicitly,
\begin{equation}
 P_1(\tau)=\frac89\sin^2\frac{3\tau}{2}
 =\frac49(1-\cos3\tau).
 \label{eq:opt-population}
\end{equation}
The prefix decoder $\mathcal R_{\rm pre}$ discards the syndrome in these coordinates. We call this the fixed prefix decoder. Its logical coherence is described by
\begin{equation}
 \zeta(\tau)=\langle u_+(\tau)|u_-(\tau)\rangle
 =\frac{8e^{i\tau}+2e^{-2i\tau}-e^{4i\tau}}9.
 \label{eq:opt-fixed-coherence}
\end{equation}
 Its entanglement fidelity is \((1+\operatorname{Re}\mathbb E[\zeta])/2\).
Half the diamond distance~\cite{Watrous2018} of this decoded channel from the logical identity is \(|1-\mathbb E[\zeta]|/2\); this distance maximizes output distinguishability including a reference, and thus coincides with entanglement infidelity when the average coherence is real.

The deterministic-angle channel $\mathcal N_\tau$ gives a simple exact loss for this decoder. Substituting Eq.~\eqref{eq:opt-fixed-coherence},
\begin{align}
1-F_e(\mathcal R_{\rm pre}\mathcal N_\tau)
&=\frac12-\frac{\cos2\tau}{9}
+\frac{\cos4\tau}{18}-\frac{4\cos\tau}{9}\notag \\
&\le\min\{1,\tau^4/2\}. \label{eq:optical-tail-attainment}
\end{align}
To check the bound directly, put $c=\cos\tau$. The loss equals $(1-c)^2(4c^2+8c+6)/9$. The second factor is at most $18$ on $[-1,1]$, and $1-\cos\tau\le\tau^2/2$. To compare with the local result, take $\tau=\sigma x$ with $x$ a unit Gaussian. The prefix decoder has infidelity $3\sigma^4/2+O(\sigma^6)$, whereas the optimized recovery has $\sigma^4+o(\sigma^4)$. In Eq.~\eqref{eq:optimal-coefficient}, coherent subtraction replaces the raw fourth moment $3$ by the residual norm $h_2=2$. The exact finite-angle dynamics and the local polynomial projection thus describe two aspects of the same recoverable error.

 A coherent recovery can improve the overlap before discarding the syndrome. A unitary controlled by the logical branch rotates its syndrome vector without measuring that branch, so it preserves the superposition in Eq.~\eqref{eq:optical-branch-superposition}. For a known angle the two normalized vectors could be aligned perfectly. With the angle unavailable, the same relative unitary must act on their entire ensemble. Their averaged outer product retains the correlations entering the recovered overlap. The required matrix is
\begin{equation}
 C_{\nu_\omega}=\mathbb E_{\nu_\omega}\bigl[|u_-(\omega T)\rangle\langle u_+(\omega T)|\bigr].
 \label{eq:opt-cross-matrix}
\end{equation}
The branch symmetry gives $C_{\nu_\omega}^{\mathsf T}=Z_{\rm syn}C_{\nu_\omega}Z_{\rm syn}$ for every prior. Hence the cross matrix has opposite off-diagonal entries before any reflection average. This reduces its description to three complex coordinates and anticipates the two real polar branches used below.

  A relative syndrome unitary $O$ in the negative branch multiplies $|-\rangle\langle+|$ by $\Tr(OC_{\nu_\omega})$ while preserving populations. Its entanglement fidelity is $[1+\operatorname{Re}\Tr(OC_{\nu_\omega})]/2$. This is the two-sector alignment problem already solved in Proposition~\ref{prop:sector-recovery}~\cite{WuEtAl2026}. Polar alignment maximizes the ensemble overlap, and the converse in Proposition~\ref{prop:sector-recovery} bounds every CPTP recovery. Substituting the optical cross block gives
\begin{equation}
 F_{e,\mathrm{opt}}(\nu_\omega)=\frac{1+\|C_{\nu_\omega}\|_1}{2}.
 \label{eq:opt-polar-value}
\end{equation}

\subsection{Exact inference and a data-conditioned minimax decoder}
The exact finite-angle problem uses the same coherent syndrome separation as Eq.~\eqref{eq:prefix-tensor-decoding}. With two conserved logical branches, Proposition~\ref{prop:sector-recovery} reduces optimal recovery to the trace norm of one cross block. The measured real part then invokes Proposition~\ref{prop:sector-data}, so the observations determine one common alignment even while sine harmonics remain unknown.

 At calibration duration $t_{\rm cal}$, the population mean determines $\mathbb E[\cos(3\omega t_{\rm cal})]$, whereas the fixed fidelity depends on the cosine means at $T,2T,4T$ through Eq.~\eqref{eq:opt-fixed-coherence}. Independently prepared calibration ensembles at $t_{\rm cal}=T/3,2T/3,4T/3$ therefore determine the required means,
\begin{equation}
\begin{aligned}
 Q_1&=\mathbb E[P_1(\omega T/3)],\\
 Q_2&=\mathbb E[P_1(2\omega T/3)],\\
 Q_4&=\mathbb E[P_1(4\omega T/3)].
\end{aligned}
 \label{eq:opt-three-data}
\end{equation}
They determine \(c_k=\mathbb E[\cos(k\omega T)]=1-9Q_k/4\) for \(k=1,2,4\). In the notation of Section~\ref{sec:finite-calibration}, each $Q_k$ is a measured calibration mean $y_j^{\rm cal}$ for the population observable at duration $kT/3$, and $c_k=\operatorname{Re}\widehat\nu_k$. The guarantees $G$, $G_1$ and $G_2$ below are the corresponding optical specializations of $\mathcal F_{\rm cal}$ with different measured subsets. Equation~\eqref{eq:opt-fixed-coherence} then gives the exact identity
\begin{equation}
 1-F_{e,\mathrm{fixed}}=Q_1+\frac14Q_2-\frac18Q_4.
 \label{eq:opt-fixed-inference}
\end{equation}
 The negative coefficient is consistent with positivity because the three
expectations belong to one probability distribution. Feasible triples must arise from one rate distribution; independent choices of numbers in \([0,1]\) can violate this constraint.

The same observations determine more than a fixed-decoder score. They determine the entrywise real part of $C_{\nu_\omega}$ in the syndrome basis used in Eq.~\eqref{eq:opt-amplitudes}. This real matrix need not be symmetric; taking the Hermitian part would discard different information. Collect the measured means as $Q=(Q_1,Q_2,Q_4)$. Explicitly,
\begin{equation}
\begin{gathered}
 C_s(Q)=\operatorname{Re}C_{\nu_\omega}
 =\begin{pmatrix}A_Q&B_Q\\-B_Q&D_Q\end{pmatrix},\\
 \begin{aligned}
 A_Q&=1-Q_1-Q_2-Q_4/4,\\
 D_Q&=-Q_1+(Q_2+Q_4)/2,\\
 B_Q&=\frac{\sqrt2}{4}(2Q_2-Q_1-Q_4).
 \end{aligned}
\end{gathered}
 \label{eq:opt-real-matrix}
\end{equation}

 The unmeasured sine moments can change the channel. Recovery must therefore be selected from the measured data and assessed on every distribution still compatible with them. Let $\mathfrak M(Q)$ be the nonempty set of input-independent quasi-static frequency distributions with these population means. Proposition~\ref{prop:sector-data} now determines the optimal guarantee for the observed real optical cross block. 
For the code in Eq.~\eqref{eq:opt-code}, the calibrated generator in Eq.~\eqref{eq:opt-axis}
and storage time $T$, assume that storage and calibration ensembles share
one input-independent quasi-static rate distribution. For feasible population
data $Q$, the unrestricted recovery guarantee is
\begin{equation}
\begin{aligned}
 &\sup_{\mathcal R}\inf_{\nu_\omega\in\mathfrak M(Q)}
 F_e(\mathcal R\circ\mathcal N_{\nu_\omega})\\
 &\quad=G(Q)=\frac{1+\|C_s(Q)\|_1}{2}.
\end{aligned}
 \label{eq:opt-data-minimax}
\end{equation}
In explicit scalar form,
\begin{equation}
\begin{split}
 G(Q)=\frac12+\frac12\Big[
  &A_Q^2+D_Q^2+2B_Q^2\\
  &+2|A_QD_Q+B_Q^2|\Big]^{1/2}.
\end{split}
 \label{eq:opt-data-value}
\end{equation}
 One real polar alignment of $C_s(Q)$ attains $G(Q)$ for every distribution
in $\mathfrak M(Q)$. The absolute value in Eq.~\eqref{eq:opt-data-value}
is required even for symmetric distributions.

 Equation~\eqref{eq:opt-population} makes the observations invariant under rate reflection. Equation~\eqref{eq:opt-real-matrix} shows that they determine the real cross block, and reflection conjugates $C_{\nu_\omega}$. Thus the hypotheses of Proposition~\ref{prop:sector-data} hold. Its two-sector value is $G(Q)$. We construct the corresponding optical recovery explicitly.
Take a real singular-value decomposition \(C_s=O_LD_{\rm sv}O_R^{\mathsf T}\) with $O_L,O_R$ orthogonal and $D_{\rm sv}$ diagonal and nonnegative, and choose
\(O_Q=O_RO_L^{\mathsf T}\), completing null singular spaces orthogonally if necessary.
Apply \(O_Q\) to the syndrome coherently controlled by the negative logical branch, preserving superpositions of the two logical states, then
discard the syndrome. For every compatible \(\nu_\omega\),
\begin{equation}
 \operatorname{Re}\operatorname{Tr}(O_QC_{\nu_\omega})
 =\operatorname{Tr}(O_QC_s)=\|C_s\|_1.
 \label{eq:opt-uniform-payoff}
\end{equation}
This one recovery therefore attains \(G(Q)\) throughout the uncertainty class. For the matching converse, Proposition~\ref{prop:sector-data} uses the compatible symmetrized rate distribution. Its cross block is $C_s$, so Eq.~\eqref{eq:opt-polar-value} bounds every recovery by the same $G(Q)$.
 For symmetric noise the data-conditioned value equals the actual channel's
optimum. For asymmetric noise it is the best guarantee supported by these
observations. For this observation class, the measured real coherence matrix reduces robust channel recovery~\cite{KosutShabaniLidar2008} to the explicit value $G(Q)$ and an operation attaining it uniformly.

This guarantee concerns entanglement fidelity. The recovered channel preserves logical populations and multiplies $|-\rangle\langle+|$ by $ \gamma_{\nu_\omega}=\Tr(O_QC_{\nu_\omega})$, with $| \gamma_{\nu_\omega}|\le1$ and $\operatorname{Re} \gamma_{\nu_\omega}=2G(Q)-1$. The unmeasured sine harmonics can leave $\operatorname{Im} \gamma_{\nu_\omega}\ne0$. For this qubit channel, its half diamond distance from the identity is $D_\diamond=|1- \gamma_{\nu_\omega}|/2$.  The qubit formula follows because the channel difference acts only on the off-diagonal input blocks; each off-diagonal block has trace norm at most $1/2$, attained by an equatorial pure state. Therefore
\begin{equation}
 1-G(Q)\le D_\diamond\le\sqrt{1-G(Q)}.
 \label{eq:opt-fidelity-diamond}
\end{equation}
The upper bound follows from $| \gamma_{\nu_\omega}|\le1$; symmetric noise makes $ \gamma_{\nu_\omega}$ real and attains the lower bound. The population data thus control this operational error too, but do not fix its value for asymmetric noise or determine diamond-distance minimax optimality. In contrast with the known-prior bracket of Proposition~\ref{prop:worst-input}, the unmeasured sine harmonics leave an unknown logical phase in $\gamma_{\nu_\omega}$, which no fixed final rotation can remove.

 The relative syndrome alignment is coherent; it does not measure the logical branch. Together with syndrome disposal it defines a deterministic recovery on the cyclic space. Appendix~\ref{app:optical-implementation} gives its Kraus operators on all five physical dimensions and the analyzer derivation. Realization using passive two-mode optics alone remains unresolved.

\subsection{How two durations recover the quartic guarantee}
\label{sec:optical-calibration-hierarchy}

The three means determine the fixed-decoder score and the optimal guarantee
supported by their information. We now ask how much protection remains when
one or two of these means are available. The code, calibrated generator,
storage time and quasi-static assumption remain those of
Eq.~\eqref{eq:opt-data-minimax}. The recovery may use the measured means,
and must work for every rate distribution compatible with them. Even when the actual noise is narrow, a few small population means can remain compatible with rare finite rotations. The guarantee must therefore retain these alternatives.

Keep \(c_k=\mathbb E[\cos(k\tau)]=1-9Q_k/4\), with \(\tau=\omega T\).
These are the real parts of the Fourier moments
\(\widehat\nu_k=\mathbb E[e^{-ik\tau}]\).
The protocol with one duration measures \(Q_2\) at \(2T/3\); the protocol
with two durations also measures \(Q_1\) at \(T/3\).
Write \(G_1(c_2)\) and \(G_2(c_1,c_2)\) for their optimal worst-prior
entanglement fidelities, and
\(G_3(c_1,c_2,c_4)=G(Q)\) for the three-duration value already derived.
The subscript counts calibration durations.

The entrywise real coherence matrix in Eq.~\eqref{eq:opt-real-matrix}
becomes
\begin{equation}
\begin{gathered}
C_s(c_1,c_2,c_4)=\begin{pmatrix}A_Q&B_Q\\-B_Q&D_Q\end{pmatrix},\\
\begin{aligned}
A_Q&=(4c_1+4c_2+c_4)/9,\\
D_Q&=(4c_1-2c_2-2c_4)/9,\\
B_Q&=\sqrt2(c_1-2c_2+c_4)/9.
\end{aligned}
\end{gathered}
\label{eq:opt-harmonic-cross-matrix}
\end{equation}
Its trace norm has a particularly simple geometry,
\begin{equation}
3\|C_s\|_1=\max\left\{
 |c_4+2c_2|,\sqrt{(c_4-2c_2)^2+8c_1^2}\right\}.
\label{eq:opt-two-polar-branches}
\end{equation}
The two branches correspond to a real reflection or rotation of the
syndrome, coherently conditioned on the logical branch. They meet when
\(c_2c_4=c_1^2\), where \(C_s\) is singular.
To check Eq.~\eqref{eq:opt-two-polar-branches}, use the entries of Eq.~\eqref{eq:opt-harmonic-cross-matrix}. Then \(\|C_s\|_1^2=A_Q^2+D_Q^2+2B_Q^2+2|A_QD_Q+B_Q^2|\), \(A_Q-D_Q=(2c_2+c_4)/3\), and \((A_Q+D_Q)^2+4B_Q^2=[(c_4-2c_2)^2+8c_1^2]/9\).

The missing harmonic is constrained by positivity of the same underlying
probability measure. Put \(S_c=(1+c_2)/2\), the cosine second moment of Eq.~\eqref{eq:calibration-cosine-moments}, in which \(c_1=a_c\). The observed pair is feasible precisely when
\(c_1^2\le  S_c\le1\). Its possible fourth harmonics fill
\([\underline c_4,\overline c_4]\), where
\begin{equation}
\begin{aligned}
\underline c_4&=2c_2^2-1,\\
\overline c_4&=
\begin{cases}
1,&|c_1|\le S_c,\\
1-\dfrac{8(|c_1|-S_c)^2(1-S_c)}{(1-|c_1|)^2},
 &S_c<|c_1|<1.
\end{cases}
\end{aligned}
\label{eq:opt-fourth-harmonic-interval}
\end{equation}
At \( |c_1|=1\), feasibility forces \( S_c=1\) and \(c_4=1\).
This is the real harmonic slice of the bandwidth-four moment set in Eq.~\eqref{eq:calibration-moment-set}. Explicit finite measures attain both bounds. Their construction
in Appendix~\ref{app:optical-bootstrap} also identifies adverse priors for
the recovery problem.

For each feasible triple $(c_1,c_2,c_4)$, Proposition~\ref{prop:sector-data} determines the common optimum while sine harmonics remain unknown. When only $(c_1,c_2)$, or only $c_2$, is measured, Proposition~\ref{prop:finite-calibration} additionally minimizes over the missing cosine means and constructs one decoder for the entire enlarged class.
Retain the calibrated code and quasi-static shared-rate assumption of Eq.~\eqref{eq:opt-data-minimax}. For two means the feasible data are $c_1^2\le(1+c_2)/2\le1$; for one mean they are $c_2\in[-1,1]$. Define
\begin{equation}
\begin{aligned}
c_{4,0}&=
\begin{cases}
2c_2,&c_1^2\ge2c_2^2,\\
c_1^2/c_2,&c_1^2<2c_2^2,
\end{cases}\\
c_{4,*}&=\min\{\overline c_4,
                  \max\{\underline c_4,c_{4,0}\}\}.
\end{aligned}
\label{eq:opt-clipped-harmonic}
\end{equation}
The second branch has \(c_2\ne0\). The exact common guarantees are
\begin{equation}
\begin{aligned}
G_2(c_1,c_2)
&=\frac{1+\|C_s(c_1,c_2,c_{4,*})\|_1}{2},\\
G_1(c_2)
&=\frac12+\frac{2|c_2|+\max\{0,2c_2^2-1\}}6.
\end{aligned}
\label{eq:opt-reduced-data-values}
\end{equation}
Each value is attained by one physical recovery throughout its data class
and has a matching symmetric adverse prior. Consequently,
\begin{equation}
G_1(c_2)\le G_2(c_1,c_2)\le G_3(c_1,c_2,c_4)
\label{eq:opt-calibration-hierarchy}
\end{equation}
for every feasible triple $(c_1,c_2,c_4)$; $G_2$ and $G_1$ use its two-mean and one-mean projections, respectively.

The unconstrained minimization of Eq.~\eqref{eq:opt-two-polar-branches} over the unmeasured
\(c_4\) gives the unconstrained minimum \(c_{4,0}\). The convex
objective is minimized on the feasible interval by clipping this point.
For one duration, \(c_1=0\) is feasible for every \(c_2\in[-1,1]\),
allows the full interval \([2c_2^2-1,1]\), and minimizes the norm at
each fixed \(c_4\). Its optimal fourth harmonic is
\(\max\{0,2c_2^2-1\}\).

These channel-wise minima become common-recovery guarantees through the
explicit polar construction. A useful case is an interior minimum at the
singular crossing \(c_1^2<2c_2^2\),
\(c_{4,0}=c_1^2/c_2\). Let \(O_-\) and \(O_+\) be the attaining
reflection and rotation at this crossing. Use them with probabilities
\begin{equation}
w_-=\frac{2c_2^2-c_1^2}{4c_2^2},
\qquad
w_+=\frac{2c_2^2+c_1^2}{4c_2^2}.
\label{eq:opt-polar-mixture-weights}
\end{equation}
Their mixed coherence payoff is
\((2|c_2|+c_1^2/|c_2|)/3\), independently of the unknown \(c_4\).
Thus randomized coherent alignments remove the sensitivity to the
unmeasured harmonic exactly where the two polar branches meet.
For a clipped endpoint, an active polar alignment gives an affine
supporting payoff whose slope has the required sign over the feasible
interval. Every alignment used here is real, so unobserved sine harmonics
make no contribution to entanglement fidelity.
Appendix~\ref{app:optical-bootstrap} gives all decoder matrices, degenerate
cases and endpoint measures. A symmetric prior attaining the minimizing
harmonics has a real cross matrix; Eq.~\eqref{eq:opt-polar-value} then
bounds every CPTP recovery by the stated value.

Both shorter durations have a consequential weak-noise interpretation.
Let \(\tau=\sigma x\) with a fixed rescaled prior, and define
\(y_j=\mathbb E[x^j]\). If \(y_2>0\) and \(y_4<\infty\), then
\begin{align}
1-G_1&=2y_2\sigma^2
 -\left(\frac43y_2^2+\frac23y_4\right)\sigma^4
 +o(\sigma^4),\notag \\
1-G_2&=\frac{y_4-y_2^2}{2}\sigma^4+o(\sigma^4), \label{eq:opt-duration-orders}\\
G_3-G_2&=o(\sigma^4). \notag
\end{align}
The proof in Appendix~\ref{app:optical-bootstrap} controls the apparently
singular endpoint in Eq.~\eqref{eq:opt-fourth-harmonic-interval}.
For a symmetric prior, Eq.~\eqref{eq:optimal-coefficient} gives the same quartic coefficient for the prior-known optimum, since its statistical residual is $x^2-y_2$.
For asymmetric priors, odd-moment information can improve a prior-known
decoder. Equation~\eqref{eq:opt-duration-orders} remains the guarantee
from population means in that case.

Adding \(Q_1\) therefore changes the leading guaranteed loss from
quadratic to quartic. The later \(Q_4\) measurement can improve the exact
finite-angle value and is still required to identify the fixed-decoder
score over arbitrary distributions. A vanishing quartic coefficient
requires a higher-order analysis; the fourth-moment assumptions here
justify \(o(\sigma^4)\), without an asserted sixth-order remainder.
For fixed-magnitude binary angles the two means determine \(c_4\)
exactly, so \(G_2=G_3\) at every width.

The largest possible increase \(G_2-G_1\) is \(\sqrt5/6\), attained at
\(c_2=0\), \(|c_1|=1/\sqrt2\); a symmetric pair
\(\tau=\pm\pi/4\) realizes it. Figure~\ref{fig:optical-calibration-frontier}
shows the complete data dependence and the narrow-Gaussian comparison. Its sharp corners have two sources. The one-duration optimum changes branch at $c_2=0$ and $c_2=\pm1/\sqrt2$; the two-duration optimum develops additional corners when its crossing minimum reaches a feasible endpoint. By contrast, clipping a smooth minimum preserves the first derivative. The distinction explains both the zero-gain boundary and the bent contours, as derived in Appendix~\ref{app:optical-corner-loci}.
The variance $\operatorname{Var}(\cos\tau)=(1+c_2)/2-c_1^2$ is the cosine variance of Eq.~\eqref{eq:calibration-cosine-moments}. Its fourth-order scaling in Eq.~\eqref{eq:calibration-fringe-variance} bounds finite-angle probability through Eq.~\eqref{eq:calibration-two-fringe-tail}, while the optical saddle fixes the attainable loss coefficient,
\(1-G_2=2\operatorname{Var}(\cos\tau)+o(\sigma^4)\)
along the specified narrow family. The optimality statement additionally
uses the common recovery and adverse-prior construction above.

\begin{figure*}[t]
\centering
\includegraphics[width=\textwidth]{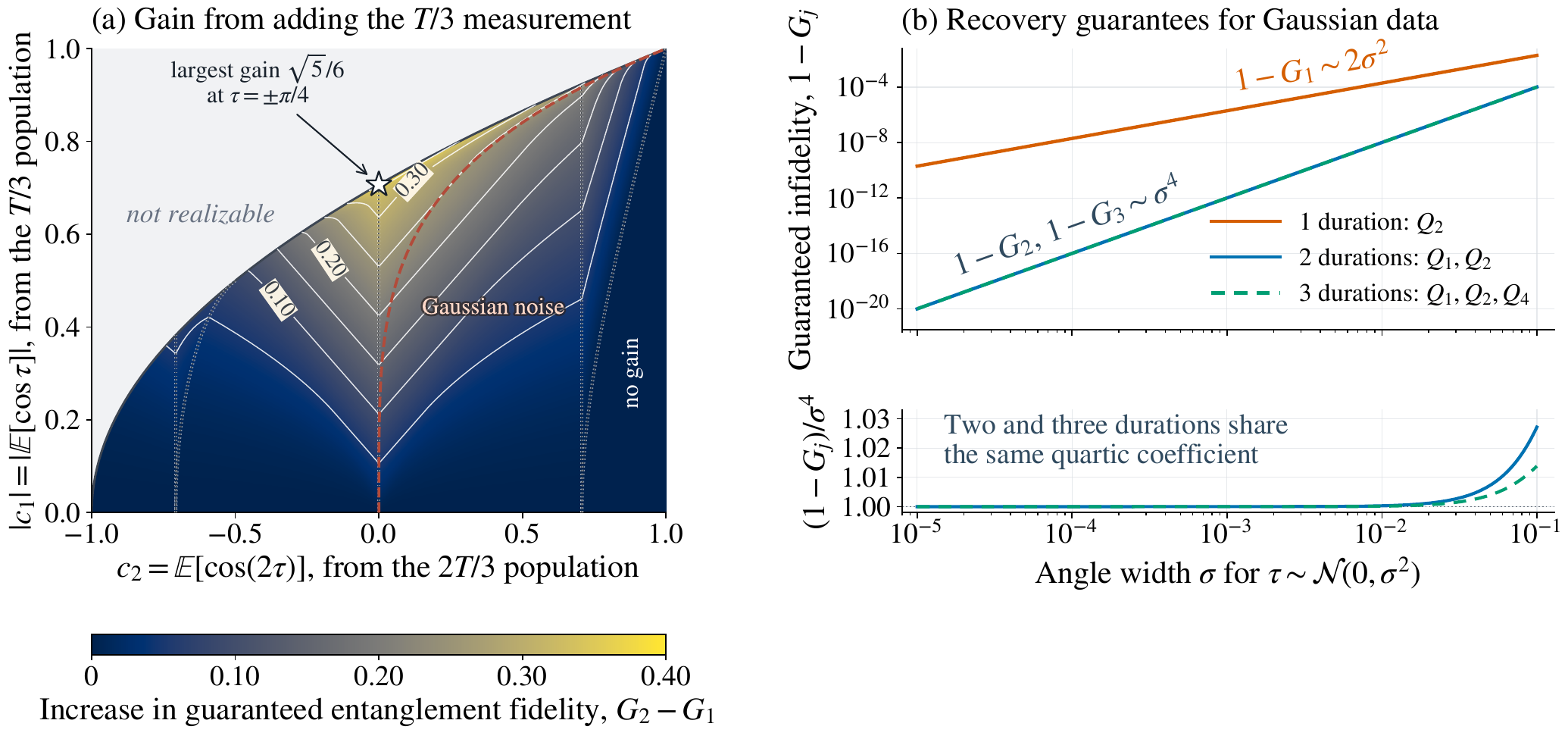}
\caption{Exact calibration-dependent recovery of the four-photon memory.
The angle is \(\tau=\omega T\), and \(Q_k\) is the mean first-shell
population at duration \(kT/3\).
(a) Gain \(G_2(c_1,c_2)-G_1(c_2)\) from adding the \(T/3\) mean to
the \(2T/3\) mean, with \(c_k=1-9Q_k/4\), the measured value of $\mathbb E[\cos k\tau]$.
The guarantee is invariant under $c_1\mapsto-c_1$, so the plot uses $c_1\ge0$. Each point is a possible pair of measured means. The shaded pairs are not realizable, since every angle distribution has $c_1^2\le\mathbb E[\cos^2\tau]=(1+c_2)/2$. The gain is largest on the upper edge, where every run has the same cosine of the storage angle and $c_1,c_2$ therefore also fix $c_4$; the star marks the maximum $\sqrt5/6$, at $\tau=\pm\pi/4$. The gain vanishes at $c_1=0$ and in the wedge $c_2\ge1/\sqrt2$ and $c_1^2\le c_2(2c_2^2-1)$. Dotted curves mark the principal visible corner loci where the optimal recovery changes form; Appendix~\ref{app:optical-corner-loci} derives all corner families. The dashed curve is the data of centered Gaussian angles of width $\sigma$, $c_1=e^{-\sigma^2/2}$ and $c_2=e^{-2\sigma^2}$, so
\(c_2=c_1^4\). It runs from $(1,1)$ toward the origin as $\sigma$ grows.
(b) Guaranteed infidelities $1-G_j$ from $j=1,2,3$ duration means for data generated
by a centered Gaussian angle of width \(\sigma\), that is, along the dashed curve of (a) near $(1,1)$. Each optimization permits all
priors compatible with those data. One mean gives \(1-G_1\sim2\sigma^2\);
two and three give \(1-G_j\sim\sigma^4\).
The lower strip separates the nearly coincident quartic curves. 
Each duration uses five analyzer settings. This exact-mean comparison
does not determine an optimal total shot allocation.}
\label{fig:optical-calibration-frontier}
\end{figure*}

\paragraph{The periodic calibration coefficients in the same memory.}
The general criterion of Proposition~\ref{prop:calibration-prefix} can
also be evaluated exactly here. Because $X_{\rm opt}$ has integer spectrum $\{-2,-1,0,1,2\}$, the class of $\tau=\omega T$ modulo $2\pi$ is the wrapped phase used in Section~\ref{sec:finite-calibration}. For $\beta_r=\mathbb E[\sin^{2r}(\tau/2)]$, the sharp coefficients of Proposition~\ref{prop:calibration-prefix} 
 in
Eq.~\eqref{eq:calibration-sharp-coefficient} are
\begin{equation}
\Upsilon_1=\frac{272+20\sqrt{10}}{243},
\qquad \Upsilon_2=8.
\label{eq:opt-calibration-coefficients}
\end{equation}
Their proof is in Appendix~\ref{app:optical-calibration-coefficients}. For this code, $\chi=1/2$ and Eq.~\eqref{eq:calibration-quartic-factor} give $\mathcal E_3(\epsilon)=3\times2^{1/3}\epsilon^{2/3}+o(\epsilon^{2/3})$ when the sole statistic is $\beta_3=\epsilon\downarrow0$. This value belongs to the single-statistic class of Proposition~\ref{prop:calibration-mismatch}. 

The first coefficient controls the optimal loss given the fundamental
statistic \(\beta_1\), which differs from the single-\(Q_2\)
experiment defining \(G_1\). The second coefficient quantifies the
protection justified by \(\beta_2\) when all compatible prior shapes
are allowed. It can exceed the coefficient inferred along a fixed-shape
narrow family because the adverse priors themselves can change shape as
\(\beta_r\) shrinks. These coefficients connect the general
calibration-order criterion to the same physical memory.

\subsection{What the protocol requires and what it cannot infer}
\label{app:optical-details}

  We now determine how to obtain the three means used by the recovery, retaining the calibrated generator $X_{\rm opt}=(J_x+J_y+J_z)/\sqrt3$ and its two-shell cyclic space. A Fock count alone does not measure the code projector, because $ |0\rangle_L$ contains a coherence between $|40\rangle$ and $|04\rangle$. The calibration therefore combines direct counts with mode mixing that converts this coherence into a count probability. Five optical settings suffice for each calibration ensemble. One directly resolves the Fock occupations. The other four use the balanced analyzer $U_\varphi=e^{-i(\pi/2)J_y}e^{-i\varphi J_z}$, at $\varphi=j\pi/4$, $j=0,1,2,3$, and count the event $(n_a,n_b)=(2,2)$, meaning two photons in each output mode. Write $P_{n_an_b}$ for the direct-count probability of a Fock outcome. If $p_j^{\rm an}$ is the analyzer probability of $(2,2)$ and $d_{\mathrm{cnt}}=P_{22}+(P_{40}+P_{04})/2$, then
\begin{equation}
\begin{gathered}
 P_{\mathrm{code}}=d_{\mathrm{cnt}}+\frac{p_0^{\rm an}-p_1^{\rm an}+p_2^{\rm an}-p_3^{\rm an}}{3},\\
 P_1=1-P_{\mathrm{code}}
\end{gathered}
 \label{eq:opt-count-readout}
\end{equation}
within the cyclic model. The projector identity itself holds for any density
matrix in the fixed four-photon sector. The analyzer's four-phase alternating sum isolates the coherence between $|40\rangle$ and $|04\rangle$. Direct Fock counts add the diagonal part of the code projector. Appendix~\ref{app:optical-implementation} evaluates the detection row and verifies Eq.~\eqref{eq:opt-count-readout}.

Calibration uses independently prepared copies with known logical input; for example, every copy can start in $ |1\rangle_L=|22\rangle$. The population mean is independent of that input under the specified generator. The stored unknown state is therefore not measured, and logical-state process tomography is unnecessary for these functionals. The rate distribution must remain the same across these ensembles and the memory experiment. The last calibration duration exceeds $T$.
A calibrated common pulse envelope can replace durations by pulse areas with
the same ratios if the multiplicative-noise distribution is unchanged.

\paragraph{Why these durations are required.}

The general expectation-identifiability criterion in Section~\ref{sec:calibration-identifiability} applies to the measured population harmonics.
A population mean at a positive calibration duration $t_{\rm cal}$ measures the frequency $3t_{\rm cal}$. The fixed-decoder fidelity in Eq.~\eqref{eq:opt-fixed-coherence} contains the three independent nonzero frequencies $T,2T,4T$. The span criterion therefore requires at least three duration means; with exactly three, their positive durations are $T/3,2T/3,4T/3$, up to order. In particular,
no finite set of these population means with all durations bounded by \(T\)
can identify its \(4T\) harmonic over arbitrary frequency distributions. The necessity result concerns exact identification of the fixed-decoder fidelity from this observable. Equation~\eqref{eq:opt-reduced-data-values} determines the common-recovery guarantee from fewer means, with two durations retaining its leading quartic coefficient.

 With finitely many calibration copies, replace each counting probability by its observed frequency and combine them as in Eq.~\eqref{eq:opt-count-readout}. This defines $Q_k^{\rm obs}$. Substituting these estimates in Eq.~\eqref{eq:opt-real-matrix} defines $C_s^{\rm obs}$. With $N$ independent shots per setting per duration, a total of $15N$, Hoeffding's inequality for weighted averages of independent bounded outcomes~\cite{Hoeffding1963} gives
\begin{equation}
\begin{gathered}
 | Q_k^{\rm obs}-Q_k|\le\epsilon_N,\qquad k=1,2,4,\\
 \epsilon_N=\sqrt{\frac{13\log(6/\alpha_{\rm fail})}{18N}},
\end{gathered}
 \label{eq:opt-data-confidence}
\end{equation}
with probability at least \(1-\alpha_{\rm fail}\), where $0<\alpha_{\rm fail}<1$. Put
\(\eta_N=(3/\sqrt2+9/4)\epsilon_N\).
A real polar recovery obtained from \(C_s^{\rm obs}\) then obeys
\begin{equation}
 F_e( \mathcal R^{\rm obs}\circ\mathcal N_{\nu_\omega})
 \ge\frac{1+\| C_s^{\rm obs}\|_1}{2}-\frac{\eta_N}{2}.
 \label{eq:opt-estimated-recovery}
\end{equation}
Each direct-count shot contributes \(1\)
for outcome \(22\), \(1/2\) for \(40\) or \(04\), and zero otherwise.
Each analyzer shot is a Bernoulli event. The squared weights of the five
independent sample means sum to \(1+4/9=13/9\), giving
Eq.~\eqref{eq:opt-data-confidence} by weighted Hoeffding and a union bound.
The coefficient matrices of \(Q_1,Q_2,Q_4\) in
Eq.~\eqref{eq:opt-real-matrix} have trace norms
\(3/\sqrt2,3/2,3/4\). Thus
\(\|C_s^{\rm obs}-C_s\|_1\le\eta_N\), and the real polar objective proves
Eq.~\eqref{eq:opt-estimated-recovery} even for infeasible empirical means.

Under this concentration bound, an additive fidelity allowance $\delta_F$ has sufficient scaling $N=O(\delta_F^{-2}\log(1/\alpha_{\rm fail}))$. Resolving a generic quartic loss at fixed relative precision therefore has sufficient scaling $N=O(\sigma^{-8}\log(1/\alpha_{\rm fail}))$, with $15N$ total shots. Optimal sample complexity remains a separate question. The duration count determines which exact means are measured, while this sampling bound determines the precision of their estimates.

For fewer measured durations, retain a confidence region inside the moment optimization. This treatment also applies to the three-duration experiment and can exploit feasibility, whereas Eq.~\eqref{eq:opt-estimated-recovery} controls its estimated polar decoder directly. If \(j\in\{1,2,3\}\)
durations each use \(N\) independent shots in each of the five settings,
weighted Hoeffding and a union bound give
\begin{equation}
| Q_k^{\rm obs}-Q_k|\le
\epsilon_{N,j}
=\sqrt{\frac{13\log(2j/\alpha_{\rm fail})}{18N}}
\label{eq:opt-fewer-duration-confidence}
\end{equation}
simultaneously for the measured \(k\), with probability at least
\(1-\alpha_{\rm fail}\). The total is \(5jN\) shots.
The corresponding cosine intervals have half-width
\(9\epsilon_{N,j}/4\).

Intersect this confidence region with the feasible moment set of Eq.~\eqref{eq:calibration-moment-set} and optimize one common recovery over every retained prior, as in Section~\ref{sec:calibration-confidence}. The witness bound in Eq.~\eqref{eq:calibration-confidence-witness} then carries the measured uncertainty into the fidelity guarantee. On the simultaneous confidence event,
the true prior belongs to this set, so the resulting lower bound
applies to the selected physical decoder. Substituting empirical means
into Eq.~\eqref{eq:opt-reduced-data-values}, or projecting them to a
single feasible point, does not by itself give that guarantee.
An empty intersection flags incompatibility between the statistical
region and the assumed model. Drift, detector bias, timing, uncertain
couplings and recovery-implementation errors require their own calibrated allowances beyond independent-trial fluctuations. The shared-phase coincidence probe
for \(\beta_r\), when used, adds a different measurement capability
and must retain its matched-coupling and common-phase assumptions.

Even complete population histories cannot identify the actual optimum without
additional noise information. Let $\omega_0>0$ and take \(\omega=\pm\omega_0\) with probabilities
\(p,1-p\). Every population history is independent of \(p\). At
\(T_{\rm ret}=2\pi/(3\omega_0)\),
the physical state returns to the code with an uncertain logical phase, while the syndrome returns to its initial state, giving
\begin{equation}
\begin{gathered}
 F_{e,\mathrm{fixed}}=\frac14,\\
 F_{e,\mathrm{opt}}(p)=\frac12+\frac14\sqrt{1+3(2p-1)^2}.
\end{gathered}
 \label{eq:opt-sign-ambiguity}
\end{equation}
The observable therefore supports the uniform guarantee \(3/4\), while the
actual optimum can reach one. The quasi-static restriction cannot be replaced by arbitrary sign fluctuations.
Consider two cycles of duration \(T_{\rm ret}\), with Hamiltonian
\(s_j\omega_0X_{\rm opt}\) in cycle \(j\), \(s_j=\pm1\). Every sign sequence has
the same population at every time, because complete cycles shift the sine
argument in Eq.~\eqref{eq:opt-population} by an integer multiple of \(\pi\).
At the end, the three sequences \((++),(+-),(--)\) have coherence phases
\(e^{-2\pi i/3},1,e^{2\pi i/3}\). An equal mixture completely dephases the
logical qubit and has optimum \(1/2\); a known deterministic sequence is
invertible. This model lies outside the quasi-static uncertainty class and
pinpoints the information about intra-run correlations omitted by populations.

\paragraph{Statistical comparisons in the same memory.}
The measured population means select a compatible class of priors. We now keep the selected axis fixed and change the statistical information to a declared noise family, testing the ordinary-moment and finite-alphabet mechanisms at finite width.
 The ordinary-moment comparison of Section~\ref{sec:uncertainty} also has an exact finite-width realization in this memory. With the same ternary family and recovery resources, the nominal and robust worst-case quartic coefficients are $81/200$ and $1/8$ when the nominal model uses $p_0=1/10$. Their ratio is $81/25$, while the nominal decoder has coefficient $9/200$ on its assumed channel. At the finite width $\sigma=1/5$, the robust decoder lowers the worst-case infidelity by a factor $2.89$, and by $2.78$ when Gaussian jitter of standard deviation $\sigma/10$ broadens each angle peak (Fig.~\ref{fig:benchmark}). The least-favorable prior changes with the width and first reaches the binary endpoint at $\sigma_c\approx0.4271$, where $\cos\sigma_c$ is a root of an irreducible sextic polynomial. Appendix~\ref{app:benchmark-details} derives the finite-width saddle, its changing adverse prior and its persistence under Gaussian broadening.

The binary endpoint tests what a vanishing statistical residual means beyond leading order. For this code, the optimal binary loss starts at $16\sigma^6/9$. Symmetric contamination by zero-angle runs restores a quartic term; along $\epsilon=c\sigma^2$ the optimal infidelity divided by $\sigma^6$ tends to $16/9+c/2$ (Fig.~\ref{fig:contamination-crossover}). Appendix~\ref{app:joint-details} proves this simultaneous limit with an explicit uniform remainder and a common-decoder consequence for bounded contamination. The same uniform bound covers smooth broadening. If Gaussian jitter of standard deviation $\delta\sigma$ broadens each binary angle, the optimal infidelity divided by $\sigma^6$ tends to $16/9+2c$ along $\delta^2/\sigma^2\to c$. Broadening of order $\sigma^2$ therefore keeps the sixth-order loss with a changed coefficient, while broadening $o(\sigma^2)$ retains $16/9$. This realizes the conditional crossover mechanism of Section~\ref{sec:finite-support} with a remainder uniform in the contamination probability.

\subsection{The calibration required for the optical axis}
\label{sec:optical-axis-calibration}
The population-data analysis has kept the optical generator fixed. We now test that physical assumption by varying its direction. Axis calibration is also essential to the weak-noise protection. Rotating about a generic fixed direction can distinguish logical states at an earlier order. We  determine which directions retain the correctable first shell and the optimal penalty when they do not. For a real vector
$\mathbf n=(n_x,n_y,n_z)$, write
$H_{\mathbf n}=\mathbf n\cdot\mathbf J$.
The Pauli matrices below act on the ordered logical basis in
Eq.~\eqref{eq:opt-code}. For another real vector $\mathbf m$, direct spin-two matrix elements give
\begin{align}
V^\dagger H_{\mathbf n}V&=0,\label{eq:opt-axis-cross-compression}\\
V^\dagger H_{\mathbf n}H_{\mathbf m}V
&=2(\mathbf n\cdot\mathbf m)I_2 \notag\\
&\quad+\sqrt3(n_xm_x-n_ym_y)\sigma_x \notag\\
&\quad+(2n_zm_z-n_xm_x-n_ym_y)\sigma_z. \notag
\end{align}
Thus a nonzero direction corrects $\{I,H_{\mathbf n}\}$ precisely when
$n_x^2=n_y^2=n_z^2$. These are four unoriented cube-diagonal axes.
For such a direction, joint correction of
$\{I,H_{\mathbf n},H_{\mathbf m}\}$ requires
$n_xm_x=n_ym_y=n_zm_z$, which forces $\mathbf m$ to be parallel to
$\mathbf n$. Separate correctability along two allowed axes therefore does not imply a common recovery for their combined error family.
This statement concerns the first-order family, not arbitrary ordered
products from a driven evolution.

The same compression identities quantify the cost of a fixed axis error. Define $\Delta_{\mathbf n}\ge0$ by
\begin{equation}
\begin{gathered}
\Delta_{\mathbf n}^2
=3(n_x^2-n_y^2)^2+(2n_z^2-n_x^2-n_y^2)^2,\\
\chi_{\rm axis}(\mathbf n)
=|\mathbf n|^2-
\sqrt{|\mathbf n|^4-\Delta_{\mathbf n}^2/4}.
\end{gathered}
\label{eq:opt-axis-calibration-coefficient}
\end{equation}
For a fixed rescaled prior $\nu$ with finite second moment, define
$\bar x=\int x\,d\nu(x)$ and
$v_\nu=\int(x-\bar x)^2d\nu(x)$.
Proposition~\ref{thm:optimal} at $L=0$ gives
\begin{equation}
1-F_{e,\mathrm{opt}}(H_{\mathbf n},\sigma,\nu)
=v_\nu\chi_{\rm axis}(\mathbf n)\sigma^2+o(\sigma^2).
\label{eq:opt-axis-calibration-error}
\end{equation}
 Indeed $A_\star=0$. The compression $V^\dagger H_{\mathbf n}^2V$
has eigenvalues $2|\mathbf n|^2\pm\Delta_{\mathbf n}$, and
$B_\star$ is its positive square root. The two square-root eigenvalues
have variance $\chi_{\rm axis}(\mathbf n)$.
For $v_\nu>0$, every fixed nonzero direction outside the four axes
therefore restores a positive quadratic error, even if the recovery knows
that direction. This formula includes rank-one leakage; when
$\mathbf n=0$ the channel is exactly reversible. It makes no uniform claim when the axis and noise width approach zero error together.

\section{Discussion}
\label{sec:discussion}
The recovery limit separates two resources that a quantum memory uses together. The encoding and generator determine the first error amplitude that distorts logical states. Knowledge of the angle statistics determines how much of that distortion one coherent recovery can cancel. The first contribution is a squared distance from a correctable Krylov boundary after polar alignment; the second is a squared polynomial residual norm. Together with the squared propagation amplitude to the boundary, they fix the leading optimal loss. An encoding can postpone the harmful power, while calibration can reduce the residual at that power. This separation reduces the local coefficient to scalar polynomial orthogonalization and one $d\times d$ matrix square root. Physical reflection can remove the occupied-shell Petz gap, while statistical reflection independently removes unnecessary decoder couplings. The optimal coefficient extends perturbative correctability and constructive-recovery analyses~\cite{Beny2011,BarnumKnill2002,ZhengEtAl2024} by determining the leading unrestricted optimum under the stated finite-moment hypotheses. 

In the language of spread complexity~\cite{BalasubramanianEtAl2022Spread}, a correctable prefix is the initial segment of each encoded state's Lanczos chain on which the recurrence coefficients are common to all logical inputs. Their variation at the first boundary detects the physical distortion $\mathcal V$, information that no single state's spread complexity contains; the noise residual $h_m$ determines how much of it survives recovery. The distinction between coherent rotation and irreversible loss also matters for the error criterion. When $C_{\mathrm{opt}}>0$, removing a known logical rotation makes the optimal worst-input error and entanglement infidelity have the same leading order. For $m\ge2$, an asymmetric prior can otherwise leave a coherent logical error of order $\sigma^{2m-1}$. The bounds in Proposition~\ref{prop:worst-input} control the ratio without assuming that it converges.

The calibration results identify the information required to exploit that physical protection. A specified narrow prior and a class of priors compatible with small measured means are different physical inputs, since the latter can include rare finite rotations. When $\beta_r$ is the sole measured moment, Proposition~\ref{prop:calibration-prefix} gives an optimal worst-prior loss $O(\beta_r)$ exactly when the span $I,H,\ldots,H^{r-1}$ is correctable. It also determines the sharp coefficient. When the first defective power is $H^m$ and $r>m$, the adverse compatible prior changes scale as the measured mean shrinks. Under the hypotheses of Proposition~\ref{prop:calibration-mismatch}, the ordinary-moment projection then fixes both the fractional exponent and its coefficient. The dual polynomial bound identifies the angles that an adverse prior can occupy. The three-point prior for the first protected boundary $m=2$ in Eq.~\eqref{eq:calibration-quartic-factor} turns those contact points into an explicit distribution attaining the limit. So does the four-angle construction for $m=3$, $r=4$ in Eqs.~\eqref{eq:calibration-cubic-majorant}--\eqref{eq:calibration-cubic-exact-value}; it gives $\mathfrak h_{3,4}=(2-\sqrt2)/3^{3/4}$ exactly. These results develop robust-recovery principles~\cite{KosutShabaniLidar2008,TaghaviKosutLidar2010} into quantitative bounds and common decoders for the declared information class.

The nested unit-moment classes make the constants $\mathfrak h_{m,r}$ nonincreasing in $r$; their limit is the Chebyshev value for a known support interval. The order of limits matters. Weak noise is taken at each fixed $r$ before $r$ increases, and a finite periodic-moment constraint continues to permit rare finite angles. Moment control approaches the support optimum through these local constants while retaining a different set of compatible priors at finite noise. With normalized chord distance, $\beta_r=\mathrm W_{2r}(\nu,\delta_0)^{2r}$ also expresses calibration as control of transport distance from zero phase. The correctable powers determine which recovery-error scalings follow from that control, complementing the transport and sensitivity analyses of Refs.~\cite{GaoKleywegt2023,BartlEtAl2021}.

In the optical memory of Section~\ref{sec:optical-application}, the answer is explicit for a hierarchy of population means. For data from a fixed symmetric narrow prior with positive variance and finite fourth moment, the population mean $Q_2$ at calibration duration $2T/3$ 
 leaves a quadratic worst-case loss. Adding the mean $Q_1$ at duration $T/3$ 
 bounds the probability of finite rotations at fourth order through Eq.~\eqref{eq:calibration-two-fringe-tail}; the optical common decoder attains the fully informed quartic coefficient in Eq.~\eqref{eq:opt-duration-orders}. A third population mean, at duration $4T/3$, determines the exact finite-width value of the three-mean data class and the fixed-decoder score. The decoder at a singular polar crossing is a mixture chosen to protect every remaining fourth harmonic. This is an operational reason to distinguish a common recovery from an arbitrary optimizer for one adverse channel. The finite-shot analysis keeps the corresponding statistical uncertainty in the same information set. As in calibrated and adaptive recovery more generally~\cite{Danageozian2023IncompleteRecovery,IyerEtAl2025}, the available observations and their precision are part of the recovery resource.

Moment positivity organizes the compatible channels in this comparison. The shared positivity principle links the construction to moment optimization and quantum-mechanical bootstrap methods~\cite{Lasserre2001,HanHartnollKruthoff2020}. Here it constrains an angle distribution; in the quantum-mechanical bootstrap it constrains state expectations. For one commensurate phase, finite Toeplitz realization closes compatibility at the bandwidth entering the channel~\cite{YangXieStoica2016}. Increasing that bandwidth without new measurements therefore leaves the recovery information unchanged. The useful experimental question is which observation excludes the adverse channels that currently fix the guarantee.

These conclusions concern one scalar integrated angle multiplying a known finite generator and one final deterministic quantum recovery. The local coefficient fixes the rescaled prior and its moment and rank hypotheses; varying-prior limits use the uniform bounds stated separately. The optical application additionally assumes a loss-free memory and the same quasi-static rate distribution during calibration and storage. Time-dependent scalar noise still fits the general channel through its integrated angle, but relating different calibration durations then requires information about its temporal correlations.

Measurement design can use the extra-observable criterion in Eq.~\eqref{eq:calibration-measurement-criterion} to identify data that eliminate all current adverse priors. The confidence multipliers quantify the value of improving a selected observable's precision. Optimizing the observable, duration and shot allocation together would require the attainable measurement set and its noise model. The shared-phase coincidence protocol demonstrates a positive estimator of the required periodic moment, with an explicit cost in probes and independent noise draws. It leaves open which implementation is most efficient for a given memory. Drift or generator uncertainty can be incorporated by enlarging the compatible channel class, if their calibration constraints remain justified.

Implementing the selected recovery is a further physical task. The unrestricted optical optimum gives a target against which a specified apparatus can be assessed. The population analyzers use two-mode operations and photon counts, whereas passive-optical realization of every attaining recovery remains unresolved. Adding photon loss also changes the fixed-number cyclic space and the error family. The original code was designed for amplitude damping~\cite{ChuangLeungYamamoto1997}; joint control of coherent rotations and loss should retain the corresponding cross-error overlaps. The axis criterion already shows why individually correctable perturbations need not admit one common decoder. These are physical implementation and model questions beyond the present loss-free rotation channel.

Extending the noise model to several fluctuating coefficients requires a matrix organization of statistical polynomial degrees. Positive multivariate moments and quotient constructions can encode noise support, but the triangle example proves that their dimension counts alone cannot determine recovery. Even commuting generators can act differently within a multidimensional new degree. A useful extension of the interaction-algebra viewpoint~\cite{KnillLaflammeViola2000} should retain that logical Gram information while exploiting polynomial positivity to constrain the prior. Noncommuting time-dependent noise additionally carries ordered products of generators. The one-angle theory identifies both the projection mechanism to preserve and the point at which a scalar residual ceases to suffice.

For a growing memory, the remaining question is the resource cost of the exact finite-dimensional formulation. Conserved logical sectors reduce the optimization, but the sector number, reachable dimension and harmonic bandwidth can still grow. Acyclic coherence graphs attain the pairwise polar bound, and dividing out the correctable periodic zero lowers the polynomial degree in the positivity constraint by the corrected order. Positive-semidefinite completion on chordal graphs has a characterized scalar-matrix structure~\cite{GroneEtAl1984}. For real scalar matrices, positive-definite completion has also been analyzed on cycle-completable graphs~\cite{BarrettJohnsonLoewy1996}. The recovery problem additionally optimizes positive-semidefinite block overlaps against fidelity and any imposed locality constraints. 
 These earlier completion results suggest reductions whose compatibility with that objective must be checked. Any such reduction must retain both the optimality bound and one realizable common decoder. For the present one-angle problem, the Krylov boundary identifies what the encoding can protect, and the calibrated moment set determines how much of that protection one recovery can guarantee.

\section*{Acknowledgments}
A.C. acknowledges the support and hospitality of the International Centre for Theoretical Sciences (ICTS), Bengaluru, during participation in the program Monopole Moduli 2026 (ICTS/MM2026/09).

\section*{Author contribution and use of AI tools}
The author acknowledges the use of the free versions of ChatGPT (GPT-5 series) and Claude (Claude 5 series) for assistance with organizing the literature, secondary numerical cross-checking, and LaTeX formatting. All primary research, derivations, and conclusions remain the sole work of the author, who retains full responsibility for the manuscript.

\section*{Data availability}
Data will be shared on request.

\appendix

\section*{Appendices}

\section{Explicit attaining recovery and remainder}
\label{app:attainment-details}
The main-text argument identified the harmful boundary amplitude and the noise polynomial that should replace it. This appendix constructs the actual small unitary that performs that coherent subtraction. It also controls the remainder in the averaged fidelity and treats singular leakage, so the attainment statement refers to a complete physical recovery.

 \label{app:attaining-recovery}

\begin{proof}[Proof of Proposition~\ref{thm:optimal}, completion]
Choose an isometry $U_0$ that decodes the old shells as in
Eq.~\eqref{eq:app-forced-prefix-decoding} and aligns the leakage,
\begin{equation}
 U_0E_\star=|e_m\rangle\otimes B_\star,
 \qquad e_m\perp e_0,\ldots,e_L.
 \label{eq:app-polar-leakage-decoding}
\end{equation}
To construct it when $B_\star$ is singular, use the polar decomposition
$E_\star=J_EB_\star$ on $\operatorname{Ran}B_\star$ and map
$J_Ev$ to $e_m\otimes v$ there. This is compatible with the old-shell
action because the physical subspaces are orthogonal.
An enlarged auxiliary space completes the map to an isometry on
$\mathcal H$; no full-rank leakage shell is required.

Define the traceless boundary column and Taylor coefficients
\begin{gather}
\mathcal D=|e_L\rangle\otimes(A_\star)_c+|e_m\rangle\otimes(B_\star)_c, \label{eq:app-boundary-amplitude}
\\
c=\frac{(-i)^mg_L}{m!},\qquad
c_n=\frac{(-i)^ng_n}{n!}\quad(0\le n\le L). \notag
\end{gather}
Equations~\eqref{eq:app-l2-taylor} and
\eqref{eq:app-first-uncontrolled-power} give
\begin{equation}
 \mathcal Q(U_0e^{-i\sigma xH}V)
 =c\sigma^mx^m\mathcal D+o_{L^2}(\sigma^m).
 \label{eq:app-uncalibrated-amplitude}
\end{equation}
We replace $x^m$ by its orthogonal residual $q_m$ using a small unitary
on the recovery output.
The couplings below route each degree-$n$ syndrome amplitude into the occupied-shell and leakage defects. The ratio of Taylor coefficients fixes their phase and normalization; pairing each coupling with its negative adjoint keeps the adjustment unitary.

For $0\le n\le m-2$, set
\begin{align}
R_n={}&|e_L\rangle\langle e_n|\otimes(A_\star)_c
 +|e_m\rangle\langle e_n|\otimes(B_\star)_c, \notag\\
\alpha_n={}&(c/c_n)d_n\sigma^{m-n}, \label{eq:app-lower-shell-generators}\\
G_n={}&-\alpha_nR_n+\overline{\alpha_n}R_n^\dagger. \notag
\end{align}
The sum over these indices is empty for $m=1$.
For the last occupied shell, $L=m-1$,
$\alpha_L=-id_L\sigma/m$ is purely imaginary.
Its diagonal and off-diagonal parts must be treated separately,
\begin{equation}
\begin{split}
G_L={}&-\alpha_L|e_L\rangle\langle e_L|\otimes(A_\star)_c\\
&-\alpha_L|e_m\rangle\langle e_L|\otimes(B_\star)_c\\
&+\overline{\alpha_L}|e_L\rangle\langle e_m|\otimes(B_\star)_c.
\end{split}
\label{eq:app-last-shell-generator}
\end{equation}
The first term is already anti-Hermitian and is included once.
All $G_n$ are anti-Hermitian. Define
\begin{equation}
\begin{gathered}
G_\sigma=\sum_{n=0}^LG_n,\qquad U_\sigma=e^{G_\sigma}U_0,\\
\mathcal R_\sigma(\rho)=\Tr_{\mathcal J}(U_\sigma\rho U_\sigma^\dagger).
\end{gathered}
\label{eq:app-attaining-recovery}
\end{equation}
Thus $\mathcal R_\sigma$ is a physical CPTP recovery for every real
$\sigma$.

The leading amplitude in shell $e_n$ before this rotation is
$c_n\sigma^nx^nI_d$.
The forward action of $G_n$ on it is
$-c d_n\sigma^mx^n\mathcal D$.
For $n\le m-2$, the reverse action first reaches an occupied shell at
order $\sigma^{m-n+L}$, which is at least $\sigma^{m+1}$.
The reverse leakage term acts on an amplitude of order $\sigma^m$ and
is smaller still. Higher terms in each old-shell amplitude add at least
one power of $\sigma$.
Finally $\|G_\sigma\|=O(\sigma)$ and
$G_\sigma U_0e^{-i\sigma xH}V=O_{L^2}(\sigma^m)$, so two or more
generator actions contribute $O_{L^2}(\sigma^{m+1})$.
It follows that
\begin{equation}
 \mathcal Q(U_\sigma e^{-i\sigma xH}V)
 =c\sigma^m q_m(x)\mathcal D+o_{L^2}(\sigma^m).
 \label{eq:app-attaining-amplitude}
\end{equation}
The finite moment of order $2m$ controls every polynomial retained in
this expansion and the Taylor remainder in
Eq.~\eqref{eq:app-l2-taylor}.
Since
$\|\mathcal D\|_{2,d}^2=
\operatorname{Var}_d(A_\star)+\operatorname{Var}_d(B_\star)$,
squaring Eq.~\eqref{eq:app-attaining-amplitude} proves the upper bound
matching Eq.~\eqref{eq:app-optimal-converse}.
If $h_m=0$, then $q_m=0$ almost surely and the same construction has
error $o(\sigma^{2m})$.
\end{proof}

\section{Control of the Petz singular hierarchy}
\label{app:petz-details}
The Petz coefficient depends on singular values that approach zero at different powers of the angle width. A square-root expansion that discarded the smaller scales would miss part of the leading error. We bound the absolute-value map and separate those scales before evaluating the nonscalar Gram blocks used in the main text. The final subsection compares the resulting coefficient with finite-model evaluations, including noncommuting defects and singular leakage.

\subsection{Singular reduction of the Gram square root}

\begin{proof}[Proof of Proposition~\ref{thm:petz}, reduction]
We use the singularity-safe estimate
\begin{equation}
 \big\||X|-|Y|\big\|_{\mathrm{HS}}
 \le\sqrt2\,\|X-Y\|_{\mathrm{HS}}.
 \label{eq:app-absolute-value-lipschitz}
\end{equation}
For Hermitian $A,B$ with eigenvalues $a_i,b_j$ and eigenvectors
$u_i,v_j$, the spectral-basis identity
\begin{equation}
 \|f(A)-f(B)\|_{\mathrm{HS}}^2
 =\sum_{i,j}|f(a_i)-f(b_j)|^2|\langle u_i,v_j\rangle|^2
 \label{eq:app-hs-functional-calculus}
\end{equation}
shows that a scalar $1$-Lipschitz function is $1$-Lipschitz in this norm.
Apply this to $f(a)=|a|$ and the Hermitian dilations
$\left(\begin{smallmatrix}0&X\\X^\dagger&0\end{smallmatrix}\right)$,
$\left(\begin{smallmatrix}0&Y\\Y^\dagger&0\end{smallmatrix}\right)$.
Their absolute values are block diagonal, with blocks
$|X^\dagger|,|X|$ and $|Y^\dagger|,|Y|$ respectively, proving
Eq.~\eqref{eq:app-absolute-value-lipschitz}.
Finite-rank approximation extends the same estimate to the
Hilbert--Schmidt column operators used here.

\paragraph{Separating the scalar singular scales.}
The singular-scale reduction assumes $h_m>0$; the main proof treats $h_m=0$ separately.

Apply $U_\sigma$ from Eq.~\eqref{eq:app-attaining-recovery} to the
physical output, without taking its partial trace.
It is an isometry, so $\mathsf N^\dagger \mathsf N$ and the Petz score remain unchanged.
Taking the componentwise scalar logical trace and using
Eq.~\eqref{eq:app-attaining-amplitude} gives
\begin{equation}
\begin{gathered}
\begin{aligned}
\mathsf N={}&f_\sigma\otimes I_d+o_{\rm HS}(\sigma^m)\\
 &+\widetilde a_\sigma\bigl(|e_L\rangle(A_\star)_c
 +|e_m\rangle(B_\star)_c\bigr)\langle\phi_m|,
\end{aligned}\\
\widetilde a_\sigma=c\sqrt{h_m}\sigma^m.
\end{gathered}
\label{eq:app-petz-column-expansion}
\end{equation}
Here $f_\sigma$ is a scalar operator from the noise-label space to
$\mathcal J$. The remainder includes all other Kraus columns
collectively, by Parseval's identity.

Its old-shell row functions satisfy
\begin{equation}
 (f_\sigma)_n(x)=c_n\sigma^nx^n+O_{L^2}(\sigma^{n+1})
 \quad(0\le n\le L),
 \label{eq:app-scalar-old-rows}
\end{equation}
whereas
\begin{equation}
(f_\sigma)_m(x)=c\beta\sigma^mx^m+o_{L^2}(\sigma^m),
 \quad \beta=\operatorname{tr}_d(B_\star).
 \label{eq:app-scalar-new-row}
\end{equation}
All remaining scalar rows are $o_{L^2}(\sigma^m)$.
The centered first-order generators have zero logical trace; their
products enter beyond the leading order of the corresponding row.
They therefore preserve the terms in
Eqs.~\eqref{eq:app-scalar-old-rows}--\eqref{eq:app-scalar-new-row}.

Gram--Schmidt on these actual row functions gives a factorization
$f_\sigma=L_\sigma D_\sigma\Psi_\sigma$ after discarding
$o_{\mathrm{HS}}(\sigma^m)$ rows.
The rows of $\Psi_\sigma$ are orthonormal, and the lower triangular
factor satisfies $L_\sigma=I+O(\sigma)$.
Indeed an overlap of a row of order $\sigma^n$ with the normalized
residual of row $k<n$ has size $O(\sigma^n)$, while that earlier
residual has norm of order $\sigma^k$.
The diagonal entries satisfy
\begin{equation}
\begin{aligned}
(D_\sigma)_{nn}&\sim|c_n|\sqrt{h_n}\sigma^n\quad(n\le L),\\
(D_\sigma)_{mm}&\sim|c|\beta\sqrt{h_m}\sigma^m\quad(\beta>0).
\end{aligned}
\label{eq:app-scalar-singular-scales}
\end{equation}
Multiplicative singular-value bounds for $L_\sigma D_\sigma$ transfer
these asymptotics to the singular values of $f_\sigma$.
To identify the singular vectors, let $u_{n,\sigma}^{\rm sv}$ and
$v_{n,\sigma}^{\rm sv}$ be normalized left and right singular vectors
of the retained scalar operator, with singular value $s_n\asymp\sigma^n$,
for $0\le n\le L$ and also $n=m$ when $\beta>0$.
Let $\Pi_n^{\rm aux}$ project onto
$\operatorname{span}\{e_0,\ldots,e_n\}$, and let
$\Pi_{n,\sigma}^{\rm row}$ project in the noise-label space onto
$\operatorname{span}\{f_\sigma^\dagger e_0,\ldots,f_\sigma^\dagger e_n\}$.
The singular-vector equations
$f_\sigma v_{n,\sigma}^{\rm sv}=s_nu_{n,\sigma}^{\rm sv}$ and
$f_\sigma^\dagger u_{n,\sigma}^{\rm sv}=s_nv_{n,\sigma}^{\rm sv}$ give
\begin{equation}
\begin{gathered}
\|(I-\Pi_{n,\sigma}^{\rm row})v_{n,\sigma}^{\rm sv}\|,\quad
\|(I-\Pi_n^{\rm aux})u_{n,\sigma}^{\rm sv}\|\\
\le \frac{\|(I-\Pi_n^{\rm aux})f_\sigma\|_\infty}{s_n}
=O(\sigma).
\end{gathered}
\end{equation}
For the first bound, the row projection removes
$f_\sigma^\dagger\Pi_n^{\rm aux}$; for the second, project the first
singular-vector equation onto the later auxiliary rows.
Those rows have norm $O(\sigma^{n+1})$. At the last retained row the
numerator is zero. Rescaling the first $n+1$ rows by their leading powers
of $\sigma$ gives independent limiting monomials, since
$h_0,\ldots,h_n>0$. Thus $\Pi_{n,\sigma}^{\rm row}$ converges to the
projection onto $\operatorname{span}\{1,x,\ldots,x^n\}$.
Starting at $n=0$, orthogonality to the earlier singular vectors now
forces $v_{n,\sigma}^{\rm sv}$ to converge to $\phi_n$ and
$u_{n,\sigma}^{\rm sv}$ to $e_n$, each up to phase. The estimate controls
the discarded later-row component; convergence of the row spaces
completes the argument without an additional noise-moment assumption.

When $\beta=0$, positivity implies $B_\star=0$.
After discarding the $o_{\mathrm{HS}}(\sigma^m)$ scalar rows, take the
normalized projection of $\phi_m$ onto the complement of the old
right singular space. It tends to $\phi_m$ and is exactly unused in
the truncated scalar operator. This handles the case in which no
new scalar singular value exists.

 Scalar rotations of the output auxiliary basis and of the noise labels
preserve Eq.~\eqref{eq:petz-gram}.
They reduce Eq.~\eqref{eq:app-petz-column-expansion}, up to
$o_{\mathrm{HS}}(a_\sigma)$, to scalar blocks and
\begin{gather}
T_\sigma=\begin{pmatrix}s_LI_d&a_\sigma A_{\rm sv}\\0&a_\sigma B_\star\end{pmatrix}, \label{eq:app-triangular-petz-block}
\\
a_\sigma=|\widetilde a_\sigma|>0,\quad a_\sigma/s_L=O(\sigma), \notag\\
A_{\rm sv}=e^{i\varphi_{\rm sv}}(A_\star)_c. \notag
\end{gather}
 The phase $\varphi_{\rm sv}$ depends on the scalar singular-vector conventions
and does not affect the norm of $A_{\rm sv}$.
\end{proof}

\subsection{Finite-model checks of the Petz coefficient}
\label{app:petz-numerics}
 For these checks the Hamiltonian has $L+2$ syndrome shells, encoding at shell zero. Earlier occupied-shell diagonals are $\tfrac3{10}I_d$; the final unused-shell diagonal is $\tfrac{17}{100}I_d$. The scalar hoppings are $b_1=4/5$ and, when needed, $b_2=11/10$. The boundary diagonal is $A_\star$ and its positive hopping is $B_\star$. For $d=2$, set
\begin{equation}
 A=\tfrac1{10}I+\tfrac15 Z+\tfrac3{10}X,\quad
 B=\begin{pmatrix}11/10&1/5\\1/5&3/5\end{pmatrix}.
\end{equation}
 The four cases use $(L,A_\star,B_\star)=(0,A,0)$, $(0,\tfrac1{10}I,B)$, $(1,A,B)$, and $(2,A,\operatorname{diag}(1,0))$. The two-logical-qubit case has $d=4$, $L=1$, and
\begin{align}
 A_\star&=\tfrac15 Z\otimes I+\tfrac3{10}X\otimes I+\tfrac3{20}Z\otimes X,\\
 B_\star&=I_4+\tfrac1{10}X\otimes I+\tfrac15 I\otimes Z+\tfrac3{25}Z\otimes Z.\notag
\end{align}
 The prior has nodes $(-2,-1/2,3/4,2)$ and weights $(1/5,2/5,1/10,3/10)$. We evaluated Eq.~\eqref{eq:petz-gram} with propagators $e^{-is_{\rm test}xH}V$ at $s_{\rm test}=\pm1/50,\pm1/100,\pm1/200,\pm1/400$, using 80 decimal digits. At positive width $\sigma=|s_{\rm test}|$, the negative coordinate evaluates the reflected prior; the leading coefficient is invariant under this reflection. The Gram square-root residuals were below $10^{-60}$. We averaged the two signed infidelities after division by $\sigma^{2m}$ and extrapolated in positive $\sigma$, cancelling powers one through three. The predicted coefficient is recovered with relative error below $2\times10^{-6}$ in all five cases. This procedure permits the $|\sigma|$ correction noted after Proposition~\ref{thm:petz} without assuming an even analytic remainder.

\section{Finite-alphabet Gram estimate}
\label{app:alphabet-details}
A vanishing noise-polynomial norm removes the leading term but does not by itself determine the next order. The following Gram estimate uses the fixed finite list of angles to control that remainder. The additional cancellation comes from the antisymmetry of the first nonscalar overlap term on the final scalar singular direction.
 \label{app:alphabet}
\begin{proof}[Proof of Proposition~\ref{prop:finite-alphabet}]

Assume a nondegenerate physical prefix and fixed distinct noise nodes $x_1,\ldots,x_r$ with positive weights $p_j$, where $2\le r\le L+1$. Define $E_j=\sqrt{p_j}e^{-i\sigma x_jH}V$, $\mathsf N_\sigma=[E_1\ \cdots\ E_r]$ and
\begin{equation}
\begin{gathered}
\mathsf G_{\rm at}=\mathsf N_\sigma^\dagger \mathsf N_\sigma,\quad
 \bar{\mathsf G}_{\rm at}=d^{-1}\Tr_{\rm logical}\mathsf G_{\rm at},\\
\mathsf D_{\rm at}= \bar{\mathsf G}_{\rm at}\otimes I_d,\quad \Delta_\sigma=\mathsf G_{\rm at}-\mathsf D_{\rm at}.
\end{gathered}
\end{equation}
The bar denotes the normalized trace over the logical register, leaving an $r\times r$ Gram matrix on the noise labels. With $s=2L+1$, the first possible nonscalar term is
\begin{equation}
\begin{gathered}
\Delta_\sigma=\frac{(-i)^s\sigma^s}{s!}\mathsf A_s\otimes(M_s)_c+O(\sigma^{s+1}),\\
(\mathsf A_s)_{jk}=\sqrt{p_jp_k}(x_k-x_j)^s.
\end{gathered}
\label{eq:alphabet-antisym}
\end{equation}
 $\mathsf A_s$ is real antisymmetric. The scalar Gram eigenvalues, in decreasing order, obey $\ell_j(\sigma)\asymp\sigma^{2j}$ for $0\le j<r$. Projection onto the first $r$ physical shells has the matrix $\operatorname{diag}(1,\sigma,\ldots,\sigma^{r-1})(\mathsf V_{\rm at}+O(\sigma))$, where $(\mathsf V_{\rm at})_{nj}=(-i)^ng_n\sqrt{p_j}x_j^n/n!$ is an invertible weighted Vandermonde matrix. This bounds singular values from below. A rank-$j$ Taylor approximation through degree $j-1$ has error $O(\sigma^j)$ and gives the upper bounds.

The smallest-eigenvalue unit vector approaches, up to phase, the real null vector $u_0$ of the first $r-1$ weighted Vandermonde rows, with error $O(\sigma)$. To see this, project its column-map image onto each shell $n<r-1$; its norm $O(\sigma^{r-1})$ forces the corresponding row products to vanish at least as $O(\sigma)$. The common null space is one dimensional. Thus the final diagonal block of $\Delta_\sigma$ in the scalar eigenbasis is $O(\sigma^{2L+2})$, since $u_0^{\mathsf T}\mathsf A_su_0=0$.

 Set $\mathsf R_{\rm at}=\sqrt{\mathsf G_{\rm at}}-\sqrt{\mathsf D_{\rm at}}$. The exact mixed Sylvester identity $\sqrt{\mathsf G_{\rm at}} \mathsf R_{\rm at}+\mathsf R_{\rm at}\sqrt{\mathsf D_{\rm at}}=\Delta_\sigma$ first gives $\mathsf R_{\rm at}=O(\sigma^{2L-r+2})$. The map $\mathcal M_{\mathsf D_{\rm at}}(X)=\sqrt{\mathsf D_{\rm at}} X+X\sqrt{\mathsf D_{\rm at}}$ divides block $jk$ by $\sqrt{\ell_j(\sigma)}+\sqrt{\ell_k(\sigma)}$ upon inversion. Every denominator except the final diagonal is bounded below at order $\sigma^{r-2}$; the final diagonal has the additional numerator power derived above. Therefore
 \begin{equation}
\begin{aligned}
  \mathcal M_{\mathsf D_{\rm at}}^{-1}(\Delta_\sigma)&=O(\sigma^{2L-r+3}),\\
  \mathcal M_{\mathsf D_{\rm at}}^{-1}(\mathsf R_{\rm at}^2)&=O(\sigma^{4L-3r+5}).
\end{aligned}
\end{equation}
 Their difference is $2(L-r+1)\ge0$ for $r\le L+1$, including equality at the endpoint $r=L+1$. The quadratic remainder therefore contributes at order $\sigma^{2L-r+3}$ or smaller. From $\mathsf R_{\rm at}=\mathcal M_{\mathsf D_{\rm at}}^{-1}(\Delta_\sigma-\mathsf R_{\rm at}^2)$ we obtain $\mathsf R_{\rm at}=O(\sigma^{2L-r+3})$. Also $\|\Delta_\sigma\|=o(\lambda_{\min}\mathsf D_{\rm at})$, so $\mathsf G_{\rm at}$ is positive for sufficiently small nonzero $\sigma$. Equation~\eqref{eq:petz-gram}, together with $\mathcal Q_c\sqrt{\mathsf D_{\rm at}}=0$, now proves Eq.~\eqref{eq:alphabet-order}. The matrix dimensions remain fixed throughout; there is no uniformity under colliding nodes.
\end{proof}

\section{Contamination and constrained approximation}
\label{app:contamination-details}
Rare angles outside a finite alphabet break the polynomial identity that improved recovery. We evaluate the resulting least-squares problem without treating its singular atomic moment matrix as invertible. Separating its range and kernel identifies the linear contamination cost and the correction to it.
 \label{app:contamination}
\begin{proof}[Proof of Proposition~\ref{prop:contamination}]
Let $s_*$ minimize the integral in Eq.~\eqref{eq:contamination}, set $q_*=\pi_{\rm at}s_*$ and let $z=(1,x,\ldots,x^{m-1})^{\mathsf T}$. Denote the degree-$(m-1)$ moment matrices of the atomic and contaminating distributions by $\mathsf H_{\rm at}$ and $\mathsf H_c$, and put $b=\int zq_*\,d\nu_c$. Every monic competitor is $q_*+z^{\mathsf T}v$. Its squared norm under $\nu_\epsilon$ is
\begin{equation}
 \epsilon C_*+2\epsilon b^{\mathsf T}v
 +v^{\mathsf T}\mathsf H_\epsilon v,\quad
 \mathsf H_\epsilon=(1-\epsilon)\mathsf H_{\rm at}+\epsilon\mathsf H_c.
\end{equation}
Minimization gives the exact identity
\begin{equation}
 h_m(\nu_\epsilon)=\epsilon C_*
 -\epsilon^2b^{\mathsf T}\mathsf H_\epsilon^{-1}b.
 \label{eq:contamination-identity}
\end{equation}
Vectors in $\ker\mathsf H_{\rm at}$ represent lower-degree polynomials vanishing at every atom. Stationarity of the constrained minimizer makes $b$ orthogonal to this kernel. In the range/kernel decomposition, the range block of $\mathsf H_\epsilon$ is order one, its off-diagonal blocks are order $\epsilon$ and its positive kernel block is order $\epsilon$. Its range Schur complement is therefore $\mathsf H_{\rm at}|_{\rm ran}+O(\epsilon)$, so
\begin{equation}
 h_m(\nu_\epsilon)=\epsilon C_*
 -\epsilon^2b^{\mathsf T}\mathsf H_{\rm at}^+b+O(\epsilon^3).
\end{equation}
This proves Eq.~\eqref{eq:contamination} without assuming an invertible atomic moment matrix.

 For unknown contamination on $\mathcal K$, minimize instead
\begin{equation}
 (1-\epsilon)\sum_jp_jq(x_j)^2+\epsilon\|q\|_{\infty,\mathcal K}^2
\end{equation}
 over monic degree-$m$ polynomials. Let $\mathcal W(v)$ be the smallest squared uniform norm under the constraints $q(x_j)=v_j$. Interpolation makes every $v$ feasible. The compact-set norm controls the finite polynomial coefficient vector, so the minimum exists and $\mathcal W$ is a finite convex function. At zero its value is $C_{\mathcal K,\rm at}=\min_s\|\pi_{\rm at}s\|_{\infty,\mathcal K}^2$. Choose a subgradient $g$ at zero and put $P_{\rm at}=\operatorname{diag}(p_j)$. The affine lower bound $\mathcal W(v)\ge C_{\mathcal K,\rm at}+g^{\mathsf T}v$ and the admissible choice $v=0$ give
\begin{equation}
 \epsilon C_{\mathcal K,\rm at}
 -\frac{\epsilon^2g^{\mathsf T}P_{\rm at}^{-1}g}{4(1-\epsilon)}
 \le\mathfrak h_\epsilon\le\epsilon C_{\mathcal K,\rm at}.
\end{equation}
This proves Eq.~\eqref{eq:robust-contamination} for fixed nodes, weights and compact support.
\end{proof}

\section{A two-generator obstruction to degree-only recovery}
\label{app:triangle-obstruction}

This example changes the noise model to two commuting generators and a vector-valued angle. It tests whether the univariate finite-alphabet argument can be extended using the number of independent polynomial noise functions alone. Let
\begin{equation}
\begin{gathered}
\mathcal H=\mathbb C^3\otimes\mathbb C^2,\quad V=|s\rangle\otimes I_2,\\
|s\rangle=(1,1,1)^{\mathsf T}/\sqrt3,\\
\mathbf e_j=(\cos(2\pi j/3),\sin(2\pi j/3)),\\
H_\ell=\operatorname{diag}(e_{0\ell},e_{1\ell},e_{2\ell})\otimes Z,\\
j=0,1,2,\qquad \ell=1,2.
\end{gathered}
\label{eq:triangle-generators}
\end{equation}
The generators commute, and direct averaging over the three directions gives
\begin{equation}
V^\dagger H_\ell V=0,\qquad
V^\dagger H_\ell H_kV=\tfrac12\delta_{\ell k}I_2.
\end{equation}
Hence one recovery corrects $\{I,H_1,H_2\}$ exactly. Let the vector angle take the three values
\begin{equation}
\mathbf x_k=(\cos(\vartheta+2\pi k/3),\sin(\vartheta+2\pi k/3))
\end{equation}
with equal probabilities, and evolve by $\exp[-i\sigma(x_1H_1+x_2H_2)]$. The points are noncollinear, so $1,x_1,x_2$ already span every function on this alphabet. In particular, every degree-two statistical residual against degree at most one vanishes.

The logical $Z$ labels remain in orthogonal sectors. Their syndrome vectors are complex conjugates, and their three noise branches are cyclic shifts of $3^{-1/2}(e^{-i\sigma\cos(\vartheta+2\pi j/3)})_j$, up to a harmless relabeling. For $r=0,1,2$, define their discrete Fourier amplitudes by
\begin{equation}
v_r=\frac13\sum_{j=0}^2 e^{-i\sigma\cos(\vartheta+2\pi j/3)}
 e^{-2\pi i rj/3}.
\end{equation}
Fourier transformation of the cross-sector matrix gives singular values $|v_0|^2,|v_1v_2|,|v_1v_2|$. Parseval gives $\sum_r|v_r|^2=1$. The two-sector case of Proposition~\ref{prop:sector-recovery} therefore gives the exact unrestricted optimum
\begin{equation}
1-F_{e,\mathrm{opt}}=\tfrac12(|v_1|-|v_2|)^2.
\label{eq:triangle-exact}
\end{equation}
The small-width amplitudes are
\begin{equation}
\begin{aligned}
v_1&=-\frac{i\sigma e^{i\vartheta}}2-\frac{\sigma^2e^{-2i\vartheta}}8+O(\sigma^3),\\
v_2&=-\frac{i\sigma e^{-i\vartheta}}2-\frac{\sigma^2e^{2i\vartheta}}8+O(\sigma^3).
\end{aligned}
\end{equation}
Their modulus difference gives Eq.~\eqref{eq:triangle-obstruction}; exchanging $\sigma$ and $-\sigma$ exchanges the two moduli, so the loss has an even expansion. At $\vartheta=\pi/6$, the three generator values are $\sqrt3/2,0,-\sqrt3/2$, and, for $|\sqrt3\sigma/2|<2\pi/3$, direct substitution gives
\begin{equation}
1-F_{e,\mathrm{opt}}=\frac29[1-\cos(\sqrt3\sigma/2)]^2.
\end{equation}
At $\vartheta=0$, two generator values coincide and $|v_1|=|v_2|$, giving exact recovery for every width.

The univariate proof used a one-dimensional last polynomial degree, on which the leading antisymmetric Gram contribution vanished. The triangle has two new degree-one directions. The analogous matrix contribution can act within their span. Counting polynomial functions captures the noise alphabet but leaves precisely this recovery obstruction undetermined.

\section{Finite moments and recovery conventions}
\label{app:calibration-moments}

The single-circle moment problem closes at finite order. To see sufficiency of Eq.~\eqref{eq:calibration-moment-set}, realize its positive Toeplitz matrix as a Gram matrix of vectors \(x_0,\ldots,x_K\) in a space of dimension \(q=\operatorname{rank}\mathsf T_K\). The shift \(x_j\mapsto x_{j+1}\), for \(0\le j<K\), preserves their pairwise inner products because the matrix is Toeplitz. It is well defined on their span because a zero-norm linear combination remains zero-norm after the shift. Its domain and image have equal dimension, so it extends to a unitary \(U\) on the \(q\)-dimensional Gram space. In particular, \(x_j=U^jx_0\) for every \(j\le K\).

Write the spectral decomposition as \(U=\sum_\ell e^{i\theta_\ell}\Pi_\ell\). With the inner product conjugate-linear in its first argument,
\begin{equation}
 \langle x_i,x_j\rangle
 =\langle x_0,U^{j-i}x_0\rangle
 =\sum_\ell
   e^{-i(i-j)\theta_\ell}\|\Pi_\ell x_0\|^2.
 \label{eq:calibration-atomic-construction}
\end{equation}
The weights are nonnegative and sum to \(\|x_0\|^2=1\). They therefore define the desired circle probability measure, supported on at most \(q\le K+1\) phases. A zero-mass residual matrix has the zero measure as its representation. The construction applies also to unnormalized positive Toeplitz residuals, a fact used in the atom bound.

We next check the Choi pairing in Eq.~\eqref{eq:calibration-payoff}. For a deterministic phase let \(B_\theta=U_\theta V\). Entanglement fidelity is
\begin{equation}
 f_{\mathcal R}(\theta)
 =\frac1{d^2}\sum_\alpha
       |\operatorname{Tr}(R_\alpha B_\theta)|^2.
 \label{eq:calibration-kraus-trace}
\end{equation}
With \(\mathsf b_\theta=\operatorname{vec}(B_\theta^\dagger)\), the specified tensor order gives
\(\langle \mathsf b_\theta|R_\alpha\rangle\!\rangle
=\sum_{l,i}(B_\theta)_{il}(R_\alpha)_{li}
=\operatorname{Tr}(R_\alpha B_\theta)\).
Since \(\mathsf b_\theta=\sum_a e^{in_a\theta}\mathsf b_a\), averaging
\(|\mathsf b_\theta\rangle\langle \mathsf b_\theta|/d^2\)
gives exactly Eq.~\eqref{eq:calibration-payoff}. The partial trace in Eq.~\eqref{eq:calibration-recovery-choi} equals
\((\sum_\alpha R_\alpha^\dagger R_\alpha)^{\mathsf T}\),
which is \(I_P\) for a trace-preserving recovery. Useful normalization checks are
\begin{equation}
 \operatorname{Tr}\mathsf W=\frac1d,\qquad
 \operatorname{Tr}\!\left[
       \frac{I_d\otimes I_P}{d}\,\mathsf W\right]=\frac1{d^2}.
 \label{eq:calibration-choi-checks}
\end{equation}
The second expression is the fidelity of a recovery replacing every input by the maximally mixed logical state.

\section{Polynomial witnesses and singular calibration}
\label{app:calibration-witnesses}

For a fixed recovery, minimizing its fidelity over Eq.~\eqref{eq:calibration-moment-set} is a finite semidefinite moment problem. Its dual consists of the pointwise lower bounds in Eq.~\eqref{eq:calibration-polynomial-witness}. The correspondence follows by averaging a Gram polynomial and using
\(\mathbb E[v_K^\dagger Zv_K]=\operatorname{Tr}(Z\mathsf T_K)\).
Every real trigonometric polynomial of degree at most \(K\) has a Hermitian Gram representation; its positivity on the circle permits a positive one by Fej\'er--Riesz factorization.

There is no duality gap, including at singular moment data. To verify a sufficient strict-feasibility condition on the witness side, take any Hermitian coefficient-matching Gram matrix for \(f_{\mathcal R}\). Adding a sufficiently large multiple of the identity makes it positive definite and changes the polynomial by a constant, since \(v_K^\dagger I v_K=K+1\). This is absorbed by a sufficiently negative \(\lambda_0\), with the other multipliers set to zero. Strict dual feasibility and compactness of the nonempty primal moment set give equality of the primal minimum and dual supremum. Taking the best recovery gives Eq.~\eqref{eq:calibration-minimax-sdp}. The supremum over finite multipliers may nevertheless be unattained.

For an explicit boundary example, take an unencoded qubit and
\begin{equation}
 U_\theta=\operatorname{diag}(1,e^{-i\theta}),\qquad
 \mathbb E_\nu [s(\theta)^2]=0.
 \label{eq:calibration-singular-example}
\end{equation}
Positivity of \(s^2\) restricts the prior to the zero phase, so identity recovery has fidelity one. Any recovery perfect on this full unencoded input is the identity channel, whose loss at a general phase is \(s\). An attained target-one polynomial witness would require
\begin{equation}
 -s-\lambda_1s^2\ge0\qquad(0\le s\le1),
 \label{eq:calibration-unattained-multiplier}
\end{equation}
which fails for all finite \(\lambda_1\) near \(s=0\). Targets \(1-\eta\), for every \(\eta>0\), are feasible because
\begin{equation}
 \eta-s+\frac{s^2}{4\eta}
 =\frac{(s-2\eta)^2}{4\eta}\ge0.
 \label{eq:calibration-limiting-witness}
\end{equation}
Thus the physical optimum is attained and the witness supremum equals it, while the unreduced multipliers diverge. Using the support implied by the zero moment removes the redundant boundary description. A strictly positive calibrated Toeplitz matrix is a sufficient condition for witness attainment in the unreduced moment problem.

When an optimal witness does attain its value, complementary slackness has a physical interpretation. If \(\widehat\nu_*,Y_*,J_*,Z_*\) attain the common value, then
\begin{align}
 (I_d\otimes Y_*-\mathsf W(\widehat\nu_*))J_*&=0,\nonumber\\
 \mathsf T_K(\widehat\nu_*)Z_*&=0.
 \label{eq:calibration-complementarity}
\end{align}
These identities follow from zero trace pairing of positive matrices. Every representing adverse prior is supported on zeros of
\(f_{\mathcal R_*}-\sum_j\lambda_jg_j^{\rm cal}\),
because its nonnegative integral is zero. A nonzero such residual of degree \(K\) has at most \(K\) distinct circle zeros because its Fej\'er--Riesz factor has degree at most \(K\). If \(\mathsf T_K(\widehat\nu_*)\) is positive definite, the second identity forces \(Z_*=0\), and the common decoder's entire pointwise fidelity belongs to the measured-observable span.

\section{The adverse two-angle coefficient}
\label{app:calibration-rare-pair}

The coefficient controlling a rare finite rotation can be evaluated without an expansion of the full uncertain channel. Fix a phase \(\theta_*\) and set
\begin{equation}
 \mathsf C_{\rm pair}=V^\dagger U_{\theta_*}V,\quad
 \mathsf B_{\rm pair}=(I_P-VV^\dagger)U_{\theta_*}V.
 \label{eq:calibration-pair-blocks}
\end{equation}
For the prior \((1-p)\delta_0+p\delta_{\theta_*}\),
\begin{equation}
\begin{gathered}
 F_{e,\rm opt}=1-p\,\kappa_{\rm pair}(\theta_*)+o(p),\\
 \begin{aligned}
 \kappa_{\rm pair}(\theta_*)
 &=1-\frac{|\operatorname{Tr}\mathsf C_{\rm pair}|^2}{d^2}\\
 &\quad-\frac{\left[\operatorname{Tr}
       \sqrt{I_d-(\mathsf C_{\rm pair})^\dagger\mathsf C_{\rm pair}}\right]^2}{d^2}.
 \end{aligned}
\end{gathered}
 \label{eq:calibration-pair-coefficient}
\end{equation}
The subscript identifies these angle-dependent maps and distinguishes them from the local Krylov boundary.

We derive Eq.~\eqref{eq:calibration-pair-coefficient} from the right derivative of optimized fidelity at zero contamination and the polar alignment of leakage.
For a recovery \(\mathcal R\), write \(a_{\mathcal R}=f_{\mathcal R}(0)\) and \(b_{\mathcal R}=f_{\mathcal R}(\theta_*)-f_{\mathcal R}(0)\). The optimized fidelity is the maximum of \(a_{\mathcal R}+p b_{\mathcal R}\) over a compact set. Its right derivative at zero is the maximum slope among recoveries with \(a_{\mathcal R}=1\). To justify the upper bound, choose maximizing recoveries as \(p\downarrow0\) and a convergent subsequence. Comparison with a perfect zero-angle decoder forces their \(a_{\mathcal R}\to1\). The limit is zero-angle-perfect, and
\((a_{\mathcal R}-1)/p+b_{\mathcal R}\le b_{\mathcal R}\)
bounds the limiting slope by the largest slope on that face. A maximizing member of the face gives the reverse inequality.

A zero-angle-perfect recovery has \(R_\alpha V=c_\alpha I_d\) and \(\sum_\alpha|c_\alpha|^2=1\). A unitary change of Kraus representation makes \(R_0V=I_d\) and \(R_{\alpha>0}V=0\). Put \(Q=I_P-VV^\dagger\). The \(V\)-to-\(Q\) block of trace preservation gives
\begin{equation}
 0=V^\dagger\!\left(\sum_\alpha R_\alpha^\dagger R_\alpha\right)Q
   =R_0Q,
\end{equation}
so \(R_0=V^\dagger\). The remaining Kraus maps obey
\(\sum_{\alpha>0}R_\alpha^\dagger R_\alpha=Q\)
and can be chosen arbitrarily subject to this condition. Their optimal contribution to the rare-phase numerator is
\begin{equation}
\begin{aligned}
 &\max_{\{R_{\alpha>0}\}}
   \sum_{\alpha>0}|\operatorname{Tr}(R_\alpha\mathsf B_{\rm pair})|^2\\
 &\qquad=\left[\operatorname{Tr}
     \sqrt{(\mathsf B_{\rm pair})^\dagger\mathsf B_{\rm pair}}\right]^2.
\end{aligned}
\label{eq:calibration-pair-polar}
\end{equation}
For the upper bound, use a singular-value decomposition
\(\mathsf B_{\rm pair}=\sum_j b_j|u_j\rangle\langle v_j|\).
The vectors with components
\(\langle v_j|R_\alpha|u_j\rangle\), indexed by \(\alpha>0\),
have Euclidean norm at most one by trace preservation. The triangle inequality bounds the square root of the left side by \(\sum_jb_j\). The partial polar inverse attains the bound; Kraus maps on the unused orthogonal complement complete trace preservation and make no contribution to this score. Since
\((\mathsf B_{\rm pair})^\dagger \mathsf B_{\rm pair}
=I_d-(\mathsf C_{\rm pair})^\dagger \mathsf C_{\rm pair}\),
the derivative above proves Eq.~\eqref{eq:calibration-pair-coefficient}.

For completeness, \(\kappa_{\rm pair}(\theta_*)\ge0\), with equality precisely when \(\mathsf C_{\rm pair}\) is scalar. If its singular values are \(a_j\in[0,1]\), then
\begin{align}
 &|\operatorname{Tr}\mathsf C_{\rm pair}|^2
  +\left(\sum_j\sqrt{1-a_j^2}\right)^2\\
 &\qquad\le\left(\sum_ja_j\right)^2
       +\left(\sum_j\sqrt{1-a_j^2}\right)^2\le d^2. \notag
\end{align}
The second equality requires coincident unit vectors
\((a_j,\sqrt{1-a_j^2})\); the first requires a common polar phase. Together these give a scalar compression, including the zero case. Conversely a scalar compression saturates both bounds. This is the Knill--Laflamme condition for the pair \(\{V,U_{\theta_*}V\}\).

The rare prior with \(p=\epsilon/s(\theta_*)\) consequently gives
\begin{equation}
 \Upsilon_1\ge
 \sup_{\theta:\,s(\theta)>0}\frac{\kappa_{\rm pair}(\theta)}{s(\theta)}.
 \label{eq:calibration-first-fringe-lower}
\end{equation}
The general result uses this lower bound without assuming equality, since different phases can prefer incompatible zero-angle-perfect recoveries. If every compression \(V^\dagger U_\theta V\) is scalar, finite Fourier inversion makes every \(V^\dagger P_aV\) scalar. All differences \(U_{\theta'}^\dagger U_\theta\) then obey the exact correction condition, so the whole orbit has a common perfect recovery.

\section{Sharp adverse-phase bounds from two fringes}
\label{app:calibration-interval}

For a probability measure on \(u\in[-1,1]\), the first two moments \(a_c,S_c\) are feasible precisely when
\begin{equation}
 \begin{pmatrix}1&a_c\\a_c&S_c\end{pmatrix}\succeq0,
 \qquad 1-S_c\ge0.
 \label{eq:calibration-interval-cone}
\end{equation}
Necessity follows from the square of every affine polynomial and from \(1-u^2\ge0\). Sufficiency can be seen directly. For \(-1<a_c<1\), the point mass at \(a_c\) has second moment \(a_c^2\), while the endpoint distribution with mean \(a_c\) has second moment one. A convex mixture realizes every \(S_c\in[a_c^2,1]\). At \(a_c=\pm1\), feasibility forces the corresponding endpoint point mass. Scaling proves the same characterization for unnormalized finite measures.

An atom of weight \(w\) at \(u_0\) leaves the residual moment matrix
\begin{equation}
 \begin{pmatrix}1&a_c\\a_c&S_c\end{pmatrix}
 -w\begin{pmatrix}1\\u_0\end{pmatrix}
      \begin{pmatrix}1&u_0\end{pmatrix}
 \succeq0
\end{equation}
and the residual support inequality
\(1-S_c-w(1-u_0^2)\ge0\).
For \(v_c=S_c-a_c^2>0\), the rank-one matrix bound is
\begin{equation}
 w_{\max}(u_0)=\min\left\{
 \frac{v_c}{v_c+(a_c-u_0)^2},\
 \frac{1-S_c}{1-u_0^2}\right\}.
 \label{eq:calibration-interval-atom}
\end{equation}
The second restriction is vacuous at \(u_0=\pm1\). Both residual conditions are sufficient by Eq.~\eqref{eq:calibration-interval-cone}, so the result is sharp. For \(v_c=0\) the distribution is \(\delta_{a_c}\). The atom in \(u\) represents the total mass at the phase pair \(\pm\arccos u_0\), with the coincident endpoints counted once.

We next prove the finite-angle event bound. Assume \(-1\le t<a_c\), and first take \(v_c>0\). Put
\(u_1=a_c+v_c/(a_c-t)\).
If \(u_1\le1\), the quadratic inequality
\begin{equation}
 \mathbf1_{\{u\le t\}}
 \le\frac{(u-u_1)^2}{(t-u_1)^2}
\end{equation}
gives Eq.~\eqref{eq:calibration-two-fringe-tail} after averaging. The two-point distribution on \(\{t,u_1\}\), with weight
\(v_c/[v_c+(a_c-t)^2]\) at \(t\), has the prescribed mean and variance and saturates the bound. When \(v_c=0\), the event has probability zero and the same formula applies.

If \(-1<t<a_c\) and \(u_1>1\), the active endpoint changes the extremizer. On \([-1,1]\),
\begin{equation}
 \mathbf1_{\{u\le t\}}
 \le 1-\frac{(u+1)(u-t)}{2(1-t)}.
 \label{eq:calibration-arc-majorant}
\end{equation}
For \(u\le t\), the subtracted product is nonpositive. For \(t\le u\le1\), it increases from zero to \(2(1-t)\), so the right side stays nonnegative. Averaging proves
\begin{equation}
 \sup\Pr(u\le t)
 =\frac{2-t-S_c-(1-t)a_c}{2(1-t)}.
 \label{eq:calibration-arc-endpoint-branch}
\end{equation}
Equality is attained on \(\{-1,t,1\}\), with weights
\begin{align}
 w_t&=\frac{1-S_c}{1-t^2},\nonumber\\
 w_{-1}&=\frac{S_c-a_c+(1-a_c)t}{2(1+t)},  \label{eq:calibration-arc-extremizer}\\
 w_1&=\frac{S_c+a_c-(1+a_c)t}{2(1-t)}. \nonumber
\end{align}
Their sum and first two moments are \(1,a_c,S_c\). The condition \(u_1>1\) is equivalent to
\(S_c>a_c-(1-a_c)t\), making \(w_{-1}>0\). Also
\(S_c+a_c-(1+a_c)t\ge(a_c+1)(a_c-t)>0\),
and \(w_t\ge0\). Hence all weights are feasible. The two branches join continuously when \(u_1=1\). At \(t=-1<a_c\), feasibility already gives \(u_1\le1\), so the first branch covers that endpoint.
With just the mean \(a_c\) known,
\(\mathbf1_{\{u\le t\}}\le(1-u)/(1-t)\)
gives \((1-a_c)/(1-t)\), attained on \(\{t,1\}\). This proves the sharp first-fringe comparison in the main text.

Finally, the narrow-family statement needs only a finite fourth moment. Let \(Y_\sigma=(1-\cos(\sigma x))/\sigma^2\). Pointwise \(Y_\sigma\to x^2/2\), while \(|Y_\sigma|\le x^2/2\); dominated convergence in \(L^2\) therefore gives
\begin{equation}
 \operatorname{Var}(\cos(\sigma x))
 =\sigma^4\operatorname{Var}(Y_\sigma)
 =\frac{\operatorname{Var}(x^2)}4\sigma^4+o(\sigma^4).
\end{equation}
Also \(1-a_c=\mathbb E [x^2]\,\sigma^2/2+o(\sigma^2)\). For fixed \(t<1\) and \(\mathbb E [x^2]>0\), this second-order distance to the endpoint exceeds
\(v_c/(a_c-t)=O(\sigma^4)\)
for sufficiently small \(\sigma\), justifying the use of the first branch.

\section{Optical recovery completion and analyzer identities}
\label{app:optical-implementation}
Two separate operations entered the optical application. The memory recovery coherently aligns syndrome amplitudes, whereas calibration estimates a code population on independent copies. We first complete the recovery from the four-dimensional cyclic space to all five physical dimensions, then derive the analyzer identity used to measure that population.

The vectors $|w_{n,\pm}\rangle=W_n|\pm\rangle$ form a physical basis of the cyclic space. After coherent alignment, tracing over the two syndrome states gives two physical-to-logical Kraus maps $R_n^{\rm opt}$, where the superscript identifies the optical recovery. Their sum $\sum_{n=0}^1R_n^{\rm opt}\rho(R_n^{\rm opt})^\dagger$ is the recovered logical state for any input $\rho$ supported there. In this representation the Kraus label describes the discarded syndrome; it need not be measured. The maps are
\begin{equation}
R_n^{\rm opt}=|+\rangle\langle w_{n,+}|
+|-\rangle\sum_{j=0}^1(O_Q)_{nj}\langle w_{j,-}|,
\end{equation}
for $n=0,1$. 
The unused normalized physical state is
\begin{equation}
\begin{aligned}
\sqrt6|{\rm dark}\rangle={}&-|40\rangle+(1+i)|31\rangle\\
 &+(1-i)|13\rangle+|04\rangle,\\
X_{\rm opt}|{\rm dark}\rangle={}&0.
\end{aligned}
\end{equation}
Adding $R_{\rm dark}^{\rm opt}=|-\rangle\langle{\rm dark}|$ completes the recovery, since $\sum_{n=0}^1(R_n^{\rm opt})^\dagger R_n^{\rm opt}+(R_{\rm dark}^{\rm opt})^\dagger R_{\rm dark}^{\rm opt}=I_5$.
Re-encoding can be appended if the recovered output must remain optical.
This recovery needs a coherent change to the logical/syndrome factorization,
a conditional syndrome alignment and syndrome disposal. A unitary dilation
exists, but passive two-mode optics alone has not been shown to realize it.
The analyzer settings measure independently prepared copies and do not enact
this coherent recovery on the stored state.

The detection row of the balanced analyzer in the ordered input basis
\((40,31,22,13,04)\) is
\begin{equation}
 r_\varphi=\langle22|U_\varphi
 =\left(\frac{\sqrt6}{4}e^{-2i\varphi},0,-\frac12,0,
 \frac{\sqrt6}{4}e^{2i\varphi}\right).
\end{equation}
Expanding \(r_\varphi\rho r_\varphi^\dagger\) at the four phases cancels
the diagonal and second-harmonic terms in the alternating sum. The remaining
term is \(3(\rho_{40,04}+\rho_{04,40})/2\). Adding direct counts proves
Eq.~\eqref{eq:opt-count-readout}.

\section{Optical harmonic completion and common recoveries}
\label{app:optical-bootstrap}

The reduced-data problem has two parts. A moment optimization determines
which unobserved fourth harmonics remain possible; polar alignment
constructs one recovery that is optimal against this entire interval. We give
the endpoint measures as well as the decoders, making the minimax
statement reproducible without solving a numerical optimization.

\subsection{The exact fourth-harmonic interval}

We first derive the feasible interval in Eq.~\eqref{eq:opt-fourth-harmonic-interval}, which is the statistical input to Eq.~\eqref{eq:opt-reduced-data-values}. Put \(u=\cos\tau\), so
\(\mathbb E[u]=c_1\), \(\mathbb E[u^2]= S_c=(1+c_2)/2\), and
\(c_4=8\mathbb E[u^4]-8 S_c+1\).
The conditions \(c_1^2\le  S_c\le1\) are necessary by variance
positivity and \(|u|\le1\). For \( S_c>0\), the measure on
\(\{\sqrt{ S_c},-\sqrt{ S_c}\}\) with masses
\begin{equation}
p_\pm=\frac{1\pm c_1/\sqrt{ S_c}}2
\label{eq:opt-lower-endpoint-measure}
\end{equation}
has these moments. For \( S_c=0\), use \(\delta_0\).
This proves sufficiency. It also attains the Jensen bound
\(\mathbb E[u^4]\ge  S_c^2\), giving
\(\underline c_4=2c_2^2-1\).

If \( |c_1|\le S_c\), the pointwise inequality
\(u^4\le u^2\) is sharp. The measure on \(\{1,-1,0\}\) with masses
\begin{equation}
\frac{ S_c+c_1}{2},\quad
\frac{ S_c-c_1}{2},\quad
1- S_c
\label{eq:opt-upper-three-point-measure}
\end{equation}
attains \(\mathbb E[u^4]= S_c\), hence \(\overline c_4=1\).
For \( S_c< |c_1|<1\), first reflect \(u\) if necessary to give it
mean \( |c_1|>0\), and set
\begin{equation}
t_c=\frac{ |c_1|- S_c}{1- |c_1|},
\quad 0<t_c\le |c_1|<1.
\label{eq:opt-upper-interior-node}
\end{equation}
The quadratic polynomial
\begin{equation}
q_{t_c}(u)=(1+2t_c+3t_c^2)u^2
-2t_c(1+t_c)^2u+t_c^2(1+2t_c)
\end{equation}
majorizes \(u^4\) on \([-1,1]\), as follows from the factorization
\begin{equation}
q_{t_c}(u)-u^4
=(1-u)(u-t_c)^2(u+1+2t_c)\ge0.
\label{eq:opt-fourth-majorant}
\end{equation}
A measure on \(\{1,t_c\}\) with mass
\begin{equation}
p_1=\frac{ |c_1|-t_c}{1-t_c}
=\frac{ S_c- |c_1|^2}{1-2 |c_1|+ S_c}
\label{eq:opt-upper-two-point-measure}
\end{equation}
at \(1\) attains equality. Its first two moments are \( |c_1|, S_c\).
Reflect its nodes when \(c_1<0\). Thus
\begin{equation}
\max\mathbb E[u^4]
= S_c-(1- |c_1|)t_c^2(1+t_c),
\end{equation}
which gives Eq.~\eqref{eq:opt-fourth-harmonic-interval}.
At \( S_c= |c_1|^2\), the measure is deterministic at
\(u=c_1\). At \( |c_1|= S_c\), the two upper-bound formulas agree.
The endpoint \( |c_1|= S_c=1\) is deterministic and is treated directly.

Both endpoint measures have the same first two moments. When
\(\underline c_4<\overline c_4\), their mixture with upper-endpoint
weight
\begin{equation}
\lambda_4=
\frac{c_4-\underline c_4}{\overline c_4-\underline c_4}
\label{eq:opt-endpoint-mixture}
\end{equation}
attains any intermediate fourth harmonic. Replace each atom \(u\) by
equal masses at \(\tau=\pm\arccos u\). This produces a symmetric
angle prior with the required \(c_1,c_2,c_4\), and hence a compatible
rate prior at \(\omega=\tau/T\). No minimal atom count is needed.

\subsection{Polar alignments and the saddle construction}

We now complete the common-decoder construction attaining Eq.~\eqref{eq:opt-reduced-data-values}. Introduce the fixed real orthogonal matrices
\begin{gather}
O_z=\begin{pmatrix}1&0\\0&-1\end{pmatrix},\quad
O_c=\frac13\begin{pmatrix}1&2\sqrt2\\-2\sqrt2&1\end{pmatrix},\notag \\
O_s=\frac13\begin{pmatrix}2\sqrt2&-1\\1&2\sqrt2\end{pmatrix}. \label{eq:opt-basic-polar-alignments}
\end{gather}
Their payoffs against Eq.~\eqref{eq:opt-harmonic-cross-matrix} are
\begin{equation}
\begin{aligned}
\operatorname{Tr}(O_zC_s)&=(2c_2+c_4)/3,\\
\operatorname{Tr}(O_cC_s)&=(2c_2-c_4)/3,\\
\operatorname{Tr}(O_sC_s)&=2\sqrt2c_1/3.
\end{aligned}
\label{eq:opt-basic-polar-payoffs}
\end{equation}
The combination \(\eta O_c+\xi O_s\) is a rotation when
\(\eta^2+\xi^2=1\). If
\(\rho_4=[(2c_2-c_4)^2+8c_1^2]^{1/2}>0\), the rotation
\begin{equation}
O_+(c_4)=
\frac{(2c_2-c_4)O_c+2\sqrt2c_1O_s}{\rho_4}
\label{eq:opt-attaining-rotation}
\end{equation}
attains the square-root branch. The reflection
\(O_-(c_4)=\operatorname{sgn}(2c_2+c_4)O_z\)
attains the other branch when its argument is nonzero.
A convex mixture of these physical decoders is CPTP. The resulting
matrix in the linear payoff is a real contraction; it need not itself
be the unitary of one deterministic alignment.
Each individual alignment has the full five-dimensional completion
described in Appendix~\ref{app:optical-implementation}.

The unconstrained minimizer of the convex norm in
Eq.~\eqref{eq:opt-two-polar-branches} is unique.
When \(c_1^2\ge2c_2^2\), the square-root branch dominates at its own
minimum \(c_4=2c_2\). Otherwise the decreasing and increasing active
branches meet at \(c_4=c_1^2/c_2\). These give
Eq.~\eqref{eq:opt-clipped-harmonic}.

Suppose first that \(c_{4,0}\) lies in the feasible interval.
If \(c_1^2\ge2c_2^2\) and \((c_1,c_2)\ne(0,0)\), take
\(O=\operatorname{sgn}(c_1)O_s\). Its payoff
\(2\sqrt2|c_1|/3\) is independent of \(c_4\).
For \(c_1^2<2c_2^2\), set \(\varepsilon_2=\operatorname{sgn}(c_2)\).
At \(c_{4,0}=c_1^2/c_2\), the mixture
in Eq.~\eqref{eq:opt-polar-mixture-weights} has the explicit contraction
\begin{equation}
O_*=
\varepsilon_2\frac{2c_2^2-c_1^2}{4c_2^2}(O_z+O_c)
+\frac{c_1}{\sqrt2|c_2|}O_s.
\label{eq:opt-common-kink-decoder}
\end{equation}
Substitution in Eq.~\eqref{eq:opt-basic-polar-payoffs} gives
\begin{equation}
\operatorname{Tr}[O_*C_s(c_1,c_2,c_4)]
=\frac{2|c_2|+c_1^2/|c_2|}{3}
\label{eq:opt-common-kink-payoff}
\end{equation}
for every \(c_4\). At \(c_1=c_2=0\), logical dephasing attains
fidelity \(1/2\), and the compatible choice \(c_4=0\) proves the converse.

If clipping selects \(c_{4,*}\ne c_{4,0}\), choose any active
real polar alignment at that endpoint. Its payoff is a supporting
affine function of the norm. Since the endpoint is strictly to one side
of the unique unconstrained minimizer, all supporting slopes there
have the sign making this affine payoff at least its endpoint value
throughout the feasible interval.
If \(c_{4,*}=c_{4,0}\) lies at an endpoint, use the constant-payoff
construction above. If the feasible interval is a single point, any real
polar completion suffices. Thus singular endpoints cause no gap in the
common-decoder construction.

For one-duration data, the matrices in
Eq.~\eqref{eq:opt-basic-polar-alignments} give a similarly explicit
recovery. For \(0<|c_2|\le1/\sqrt2\), use
\(\varepsilon_2O_z\) and \(\varepsilon_2O_c\) with equal
probabilities. Their average payoff is \(2|c_2|/3\), independently
of \(c_1,c_4\). For \(c_2\ge1/\sqrt2\), use \(O_z\); for
\(c_2\le-1/\sqrt2\), use \(-O_c\). The bound
\(c_4\ge2c_2^2-1\) proves the claimed guarantee in both cases.
At \(c_2=0\), use logical dephasing.

All of these payoffs are real functions of the measured cosine data.
They remain valid for asymmetric priors with arbitrary compatible sine
moments. Conversely, use the endpoint mixture
Eq.~\eqref{eq:opt-endpoint-mixture} with the minimizing \(c_{4,*}\),
and \(c_1=0\) in the one-duration problem. The resulting symmetric
prior has real coherence matrix \(C_s\).
Equation~\eqref{eq:opt-polar-value} bounds every recovery by precisely
the value attained above. This finishes the all-CPTP saddle proof.

\subsection{The narrow-family limit and largest information gain}

This subsection contains two arguments. The first, a narrow-family expansion, proves Eq.~\eqref{eq:opt-duration-orders}; the second determines the largest gain from the additional duration. For \(\tau=\sigma x\) with a fixed rescaled distribution and
\(y_2>0,y_4<\infty\), dominated Taylor remainders give
\begin{equation}
\begin{aligned}
c_1&=1-\frac{y_2}{2}\sigma^2
       +\frac{y_4}{24}\sigma^4+o(\sigma^4),\\
c_2&=1-2y_2\sigma^2
       +\frac{2y_4}{3}\sigma^4+o(\sigma^4).
\end{aligned}
\label{eq:opt-cosine-narrow-expansion}
\end{equation}
For sufficiently small nonzero \(\sigma\), these imply
\(c_1> S_c\), \(c_1^2<2c_2^2\), and
\(c_{4,0}=c_1^2/c_2>1\ge\overline c_4\).
The minimizer is therefore \(c_{4,*}=\overline c_4\).
Set
\begin{equation}
\delta_c=1-c_1,\quad
v_c= S_c-c_1^2=\operatorname{Var}(\cos\tau).
\end{equation}
Then \(\delta_c\sim y_2\sigma^2/2\),
\(v_c=(y_4-y_2^2)\sigma^4/4+o(\sigma^4)\), and
\((c_1- S_c)/(1-c_1)=c_1-v_c/\delta_c\).
Using \(v_c/\delta_c=O(\sigma^2)\), the exact upper endpoint expands as
\begin{equation}
\begin{aligned}
\overline c_4
&=1-8\left(c_1-\frac{v_c}{\delta_c}\right)^2
             (2\delta_c-\delta_c^2-v_c)\\
&=1-16\delta_c+40\delta_c^2+40v_c+o(\sigma^4)\\
&=1-8y_2\sigma^2+\frac{32y_4}{3}\sigma^4+o(\sigma^4).
\end{aligned}
\label{eq:opt-upper-endpoint-narrow}
\end{equation}
The square-root polar branch is active here, with
\begin{equation}
(\overline c_4-2c_2)^2+8c_1^2
=9-18(y_4-y_2^2)\sigma^4+o(\sigma^4).
\end{equation}
This proves the \(G_2\) expansion in Eq.~\eqref{eq:opt-duration-orders}.
The actual \(c_4\) has the same expansion as
Eq.~\eqref{eq:opt-upper-endpoint-narrow}, proving the \(G_3\) statement.
Expanding the \(c_2>1/\sqrt2\) branch of \(G_1\) proves its formula.
These arguments include \(y_4=y_2^2\); they then assert a vanishing
quartic coefficient, leaving the next order to the stronger analysis.

For completeness, the largest gain in the feasible moment region also
has an exact value. At fixed \(c_2\), increasing \(|c_1|\) raises
the norm at every \(c_4\) and shrinks its feasible interval.
Thus \(G_2-G_1\) is maximized on
\(c_1^2=(1+c_2)/2\), where \(c_4=2c_2^2-1\).
The square-root radicand is
\(4c_2^4-8c_2^3+8c_2+5\). Its difference from the competing squared branch is
\(4(1-c_2)(2c_2+1)^2\ge0\), with equality at $c_2=-1/2$ and $c_2=1$.
Thus the square-root branch gives the norm throughout this boundary, including the two ties. For $-1\le c_2\le1$, it obeys
\begin{equation}
\sqrt{4c_2^4-8c_2^3+8c_2+5}
\le\sqrt5+2|c_2|.
\end{equation}
For \(c_2>0\), subtracting the squared right-hand side and dividing
by \(4c_2\) leaves
\((2-c_2)(1-c_2^2)-\sqrt5\le0\).
For \(c_2=-t<0\), the corresponding expression is
\((t-1)(t+1)(t+2)-\sqrt5\le0\).
Since \(6(G_1-1/2)\ge2|c_2|\), the gain is at most \(\sqrt5/6\).
Equality occurs at \(c_2=0\), \(|c_1|=1/\sqrt2\), including the
symmetric phase pair \(\pm\pi/4\).

\paragraph{Where the calibration gain has corners.}
\label{app:optical-corner-loci}
The exact frontier also distinguishes physical changes in the adverse channel from numerical contour artifacts. In this paragraph set $x=|c_1|$, $y=c_2$, $a=x^2$, and $S_c=(1+y)/2$. The feasible half-domain is $0\le x\le\sqrt{S_c}$, $-1\le y\le1$. Write $\ell=2y^2-1$ and $u=\overline c_4$ from Eq.~\eqref{eq:opt-fourth-harmonic-interval}. Then
\begin{gather}
6(G_2-G_1)=\min_{\ell\le t\le u}f(t)-2|y|-\max(0,\ell), \notag \\
f(t)=\max\{|t+2y|,\sqrt{(t-2y)^2+8a}\}. \label{eq:opt-gain-variational}
\end{gather}
The squared branches differ by $8(ty-a)$. The unconstrained minimum lies at $t=2y$ for $a\ge2y^2$, with value $\sqrt{8a}$, and at their crossing $t=a/y$ otherwise, with value $2|y|+a/|y|$. Convexity gives the clipped minimizer in Eq.~\eqref{eq:opt-clipped-harmonic}. The values and gradients agree at $a=2y^2$. Clipping the smooth minimum $t=2y$ costs only second order in its displacement. Also $u$ is continuously differentiable where its formula changes at $x=S_c$, because its correction to one is quadratic in $x-S_c$. Thus additional corners can arise only when the crossing minimum is clipped, or from the explicit absolute-value and maximum terms in Eq.~\eqref{eq:opt-gain-variational}.

The same branches determine exactly where the additional calibration has zero gain. At $a=0$, the one-mean minimizer remains feasible and $G_2=G_1$. For $a>0$ and $|y|<1/\sqrt2$, either unconstrained minimum above exceeds $2|y|$, so the gain is positive. For $y\le-1/\sqrt2$, every feasible $t\ge\ell\ge0$ obeys $\sqrt{(t-2y)^2+8a}>\ell+2|y|$, again giving positive gain. For $y\ge1/\sqrt2$, the linear branch is at least $\ell+2y$. Equality requires $t=\ell$, where the square-root branch lies below it exactly when $a\le y\ell$. Thus the zero set consists of $c_1=0$ and the wedge $c_2\ge1/\sqrt2$, $c_1^2\le c_2(2c_2^2-1)$. The sign restriction on $c_2$ is essential.

Away from intersections, the complete interior corner families are $y=0$, $y=\pm1/\sqrt2$, and
\begin{equation}
\begin{gathered}
x^2=y(2y^2-1),\quad 1/\sqrt2<y<1,\\
x^2=b(1-2b^2),\quad y=-b,\\
\quad (\sqrt3-1)/2<b<1/\sqrt2,\\
x^2=y,\quad 1/2<y<1.
\end{gathered}
\label{eq:opt-gain-corner-loci}
\end{equation}
Each horizontal family is restricted to the feasible interior with $x>0$. The negative lower-clipping curve touches the outer boundary at $(x,y)=(1/2,-1/2)$, which is excluded from its interior locus. The reflected curves describe negative $c_1$. To verify the nonzero jumps, on the positive lower-clipping curve the local gain is $\max\{0,a/y-\ell\}/6$, so its transverse $x$-slope jumps by $x/(3y)$. This also identifies the zero-gain region. On the negative curve, put $P_b=2b^2+2b-1$ and $D_b=1+2b-2b^2$. The crossing value and clipped square root agree at $D_b$, and their $x$-slope difference divided by six is $xP_b/(3bD_b)>0$. Feasibility follows from $S_c-b(1-2b^2)=(b+1)(2b-1)^2/2\ge0$. For upper clipping, $a=yu$ reduces on the flat branch to $a=y$, with jump magnitude $x(2y-1)/[3y(1+2y)]$. On the curved branch the factorization
\begin{equation}
\begin{split}
(yu-x^2)(1-x)^2={}\\
-(x-y)^2[x^2+2xy-2x-y^2-y+1]
\end{split}
\end{equation}
has no root in the strict feasible interior $x>S_c$ that satisfies $a<2y^2$. For $y>0$, real quadratic roots require $y\ge1/2$, and the larger root $1-y+\sqrt{y(2y-1)}$ is at most $S_c$. For $y<0$, even the smaller root $1-y-\sqrt{y(2y-1)}$ lies at or beyond $\sqrt{S_c}$, with equality only at the tangency already identified. The double root $x=y$ lies outside $x>S_c$. This exhausts the upper-clipping possibilities. The explicit terms from $G_1$ contribute jumps $-2/3$ at $y=0$ and $-\sqrt2/3$ at $y=\pm1/\sqrt2$ wherever $G_2$ is smooth. At intersections the active branches determine the combined directional slopes.

Figure~\ref{fig:optical-calibration-frontier}(a) draws the closures of these curves as dotted lines, except the upper-clipping family $x^2=y$. For $1/2<y<1/\sqrt2$ that family is the level curve $G_2-G_1=1/6$: at the clipped endpoint $t=u=1$ the two polar branches coincide, the one-sided gradients are parallel, and the corner changes only the contour spacing. For $y>1/\sqrt2$ the gain along it is $(1-y^2)/3$, and contours cross it with bends below $5^\circ$. The figure also omits the continuation of the negative lower-clipping curve beyond its tangency with the outer boundary at $(1/2,-1/2)$, for $(\sqrt3-1)/2<b<1/2$. That piece lies within $0.05$ of the boundary, and its bend decreases from about $12^\circ$ at the tangency to zero at $b=(\sqrt3-1)/2$, where the crossing and smooth minima coincide. Some endpoints are differentiable, so a drawn endpoint need not have a nonzero jump. The initial upper-clipping point and the initial negative lower-clipping point are continuously differentiable; at $(0,\pm1/\sqrt2)$ the relative derivative on the feasible domain is zero. At the tangency $(1/2,-1/2)$, the clipped strip has width of order $(y+1/2)^2$, leaving relative gradient $(1/3,1/6)$ despite the neighboring corner arcs. The origin is a corner since $G_2-G_1=\sqrt2|c_1|/3$ along $c_2=0$, while it vanishes along $c_1=0$. On the half-domain the diagonal $x=y>0$ instead gives gain $x/6$, incompatible with the linear derivative inferred from the axes. The final junction $(1,1)$ is also a corner. Along $y=1-h$, $x=1-kh$, the gain divided by $h$ tends to $(5-2k)/6$ for $1/2\le k\le5/2$ and to zero for $k>5/2$, which precludes a single relative derivative. Other boundary points introduce no extra corner family. These branch arguments, together with Eq.~\eqref{eq:opt-gain-variational}, separate nonsmooth data dependence from a smooth change of optimal representation.

\subsection{Exact coefficients for periodic-moment calibration}
\label{app:optical-calibration-coefficients}

To derive Eq.~\eqref{eq:opt-calibration-coefficients}, write \(s=\sin^2(\tau/2)\). The fixed prefix decoder has loss
\begin{equation}
\ell_{\rm pre}(\tau)=\frac89s^2(9-16s+8s^2).
\label{eq:opt-prefix-periodic-loss}
\end{equation}
To optimize its comparison with all zero-angle-perfect recoveries,
use the two orthogonal isometries \(W_0=V\) and
\(W_1=X_{\rm opt}V/\sqrt2\).
For any recovery perfect on \(W_0\), write its Kraus restrictions
as \(R_\alpha W_0=c_\alpha I_2\).
A unitary Kraus rotation sets \(R_0W_0=I_2\) and
\(R_{\alpha>0}W_0=0\). Trace preservation on the off-diagonal shell
block then forces \(R_0W_1=0\).
The remaining \(R_{\alpha>0}W_1\) are Kraus maps of an arbitrary
qubit channel \(\mathcal S\). Consequently the recovery acts on the
cyclic space as
\begin{equation}
\mathcal R
\begin{pmatrix}\rho_{00}&\rho_{01}\\\rho_{10}&\rho_{11}\end{pmatrix}
=\rho_{00}+\mathcal S(\rho_{11}).
\label{eq:opt-perfect-zero-recovery}
\end{equation}
In particular, coherences between the two shells do not affect its
logical output.

For the equally weighted pair \(\pm\tau\), shell one has probability
\(t_\tau=8\sin^2(3\tau/2)/9\). Its normalized logical channel is
\begin{equation}
\begin{gathered}
\mathcal D_x=\frac{1+x}{2}\operatorname{id}
+\frac{1-x}{2}\operatorname{Ad}_{Z_b},\\
x=\cos\tau,\qquad
Z_b=|+\rangle\langle+|-|-\rangle\langle-|.
\end{gathered}
\label{eq:opt-pair-terminal-channel}
\end{equation}
Here \(\operatorname{Ad}_{Z_b}(\rho)=Z_b\rho Z_b^\dagger\), and \(Z_b\) is the original \(\sigma_y\) expressed in its conserved logical-branch basis. If \(t_\tau=0\), the shell contribution vanishes.
For any qubit CPTP map with Kraus matrices \(B_\alpha\), Pauli
orthogonality and trace preservation give
\(\frac14\sum_{\alpha,\mu=0}^3
|\operatorname{Tr}(B_\alpha\sigma_\mu)|^2=1\),
where \(\sigma_0=I_2\). The two components along \(I_2,Z_b\)
therefore imply
\(F_e(\mathcal S)+F_e(\mathcal S\operatorname{Ad}_{Z_b})\le1\).
Hence
\begin{equation}
\max_{\mathcal S}F_e(\mathcal S\mathcal D_x)=\frac{1+|x|}{2}.
\end{equation}
For \(x\ge0\), the identity channel attains the optimum.
The prefix decoder therefore optimizes this pair's fidelity over
every recovery perfect at zero.

Its quotient by \(s\) has maximum
\begin{equation}
\max_{0\le s\le1}\frac89s(9-16s+8s^2)
=\frac{272+20\sqrt{10}}{243}
\end{equation}
at \(s_*=(8-\sqrt{10})/12\).
The derivative is proportional to \(9-32s+24s^2\); evaluating both
stationary points and the endpoints proves the maximum.
At the maximizing angle,
\(\cos\tau_*=(\sqrt{10}-2)/6>0\).
Every zero-angle-perfect recovery has a maximum quotient at least
its mean quotient on \(\pm\tau_*\); the pairwise optimization above
bounds that mean below by the prefix value.
The prefix recovery also gives the matching upper bound.
This proves the first value in Eq.~\eqref{eq:opt-calibration-coefficients}
using Eq.~\eqref{eq:calibration-sharp-coefficient}.

Exact correction of both \(I\) and \(X_{\rm opt}\) also requires each
\(R_\alpha W_1\) to be scalar. Thus
\(\mathcal S=\operatorname{id}\) in
Eq.~\eqref{eq:opt-perfect-zero-recovery}; the prefix-correcting channel
on the cyclic space is unique. Its quotient by \(s^2\) has supremum
eight as \(s\to0\), proving the second value.

The adverse priors can be seen explicitly. For the first coefficient
use \((1-p)\delta_0+p(\delta_{\tau_*}+\delta_{-\tau_*})/2\),
with \(p=\epsilon/s_*\).
For the second coefficient, let both \(p\) and \(s\) vanish.
The exact symmetric-prior polar formula gives
\begin{equation}
\begin{aligned}
&1-F_{e,\mathrm{opt}}\!\left((1-p)\delta_0+
               \frac p2(\delta_\tau+\delta_{-\tau})\right)\\
&\qquad=8p(1-p)s^2+O(ps^3+p^3s^4).
\end{aligned}
\label{eq:opt-calibration-adverse-pair}
\end{equation}
For example, \(p=\sqrt\epsilon\) and \(s=\epsilon^{1/4}\) give
\(\beta_2=ps^2=\epsilon\) and loss \(8\epsilon[1+o(1)]\).
This changing-shape family explains the difference from a fixed-shape
narrow prior whose decoder can exploit more distributional information.

\section{Finite-width benchmark checks}
\label{app:benchmark-details}\label{app:smooth-benchmark}
This appendix evaluates the generic declared-family comparison for the four-photon memory. It retains the exact saddle, adverse-channel switch and broadened-noise case underlying the main-text comparison.
 \subsection{A matched comparison of nominal and robust recovery}
\label{sec:benchmark}

We evaluate the declared family of Section~\ref{sec:uncertainty}, Eq.~\eqref{eq:three-angle-family}, for the selected optical axis. Proposition~\ref{thm:minimax} fixes its weak-noise guarantee; the sector trace norm determines the exact finite-width comparison. This comparison changes the information available to recovery from population means to a family constraint. The code, Hamiltonian and actual channel are identical for all decoders; the information used in their design changes. A run in the family below either has no rotation or has one of two opposite rotations of magnitude $\sigma$. The parameter $p$ is the probability of a nonzero rotation. Its rescaled-angle distribution is $\nu_p$ of Eq.~\eqref{eq:three-angle-family}.

 The benchmark decoders receive $\sigma$ and the interval $p\in[0,1]$. The exact population means of Section~\ref{sec:optical-application} would add constraints on $p$ and are excluded from this comparison. The actual angle is $\tau=\sigma x$, with $x\sim\nu_p$, so its nonzero magnitude is $\sigma$ and its root-mean-square value is $\sigma\sqrt p$.

 The nominal recovery optimizes for the assumed value $p_0=1/10$. The robust recovery uses the entire allowed interval. An oracle recovery knows the actual $p$, while still not knowing the realized angle. All three have the same unrestricted CPTP recovery resources. The nominal decoder uses an assumed value of $p$ that these calibration data do not verify.

  \begin{figure*}[tp]
 \centering
 \includegraphics[width=\textwidth]{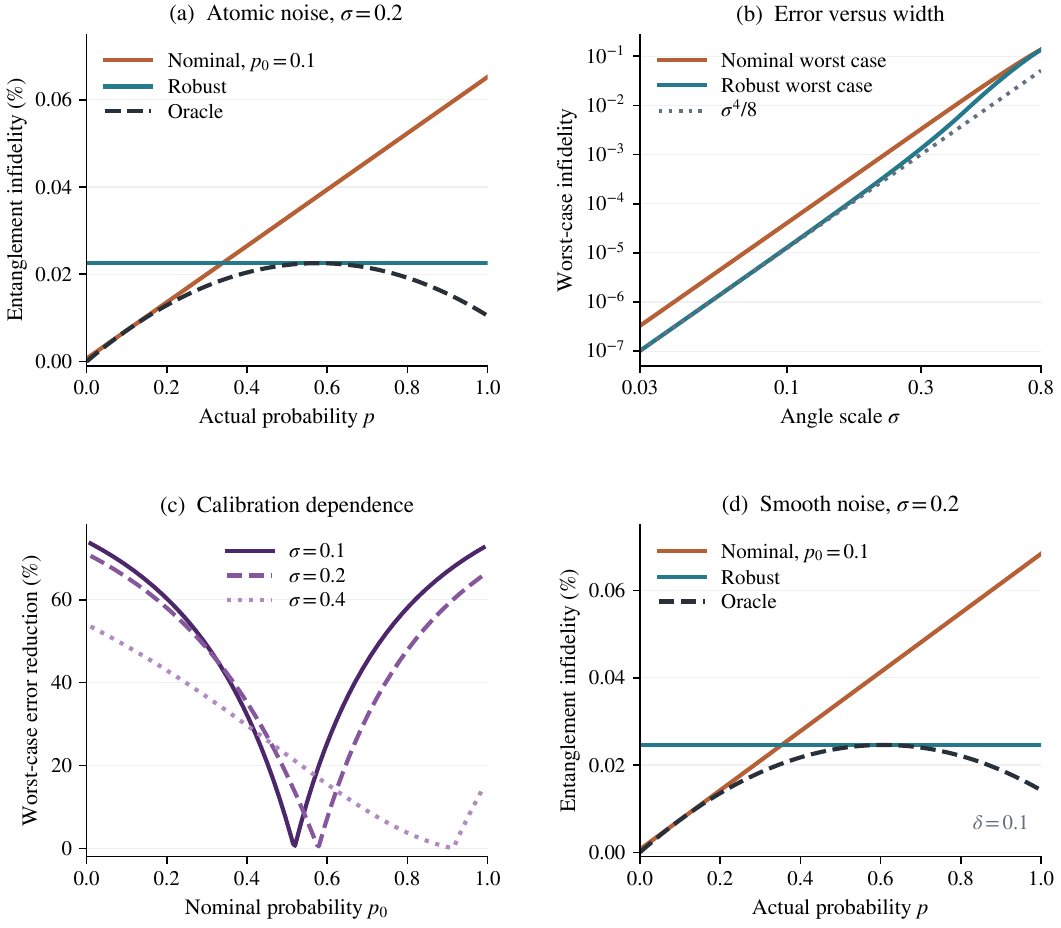}
 \caption{Recovery on the same optical code with different noise information in decoder design. The angle is zero with probability $1-p$ and equals each of $\pm\sigma$ with probability $p/2$. Panels (a), (b) and (d) use a nominal decoder designed for $p_0=1/10$; panel (c) varies $p_0$. The robust decoder knows $p\in[0,1]$, and the oracle knows the actual $p$. None knows the realized angle. (a) Actual-channel infidelity at $\sigma=1/5$. Robustness improves the worst case at a cost near the nominal channel. (b)   Worst-case errors over $p$, compared with the robust quartic coefficient. (c) Percentage reduction relative to the nominal worst-case infidelity, $100(1-E_R/E_{\rm nom})$. Here $E_R$ and $E_{\rm nom}$ are the robust and nominal worst-case entanglement infidelities over $p\in[0,1]$ at the plotted width, with $E_{\rm nom}$ evaluated for the plotted $p_0$. The reduction vanishes at the least-favorable nominal choice. (d) The comparison persists when independent Gaussian jitter of standard deviation $\delta\sigma$, with $\delta=1/10$, broadens each angle peak. Curves use exact finite matrices and characteristic functions.}
 \label{fig:benchmark}
\end{figure*}

For the optical blocks in Eq.~\eqref{eq:opt-conditional}, the real coherence matrix is
\begin{equation}
 C_p=\begin{pmatrix}1-p+pa_{\rm bin}&pb_{\rm bin}\\-pb_{\rm bin}&pd_{\rm bin}\end{pmatrix},
 \label{eq:benchmark-C}
\end{equation}
where
\begin{align}
 a_{\rm bin}&=(4\cos\sigma+4\cos2\sigma+\cos4\sigma)/9,\nonumber\\
 d_{\rm bin}&=(4\cos\sigma-2\cos2\sigma-2\cos4\sigma)/9,\label{eq:benchmark-coeffs}\\
 b_{\rm bin}&=\sqrt2(\cos\sigma-2\cos2\sigma+\cos4\sigma)/9.\nonumber
\end{align}
For $0\le\sigma\le\pi/2$, both quantities
\begin{align}
 d_{\rm bin}&=\tfrac89\sin^2(3\sigma/2)\cos\sigma,\\
 \det C_1&=\tfrac89\sin^2(3\sigma/2)\cos^2(3\sigma/2)
\end{align}
are nonnegative.

 Hence $\det C_p=p[(1-p)d_{\rm bin}+p\det C_1]\ge0$ throughout the family. Since $C_p$ has opposite off-diagonal entries and a nonnegative determinant, the trace-norm formula in Eq.~\eqref{eq:opt-data-value} reduces to $\|C_p\|_1^2=(\Tr C_p)^2+4p^2b_{\rm bin}^2$. This is why two real coordinates suffice here. Put $r_{\rm bin}=1-a_{\rm bin}-d_{\rm bin}$ and $\mathbf v_p=(1-pr_{\rm bin},2pb_{\rm bin})$. The oracle fidelity becomes
\begin{equation}
 F_{e,\mathrm{oracle}}(p)=\frac{1+\|\mathbf v_p\|_2}{2}.
 \label{eq:benchmark-oracle}
\end{equation}
 Here the trace norm becomes Euclidean and the allowed matrices form the segment $\{\mathbf v_p:0\le p\le1\}$. The dual recovery objective is its scalar product with a unit vector. The nearest point to the origin determines both the least-favorable channel and a supporting direction that protects the whole segment. It is
\begin{equation}
\begin{gathered}
 p_R=\operatorname{clip}_{[0,1]}\frac{r_{\rm bin}}{r_{\rm bin}^2+4b_{\rm bin}^2},\\
 F_R=\frac{1+\|\mathbf v_{p_R}\|_2}{2}.
\end{gathered}
 \label{eq:benchmark-saddle}
\end{equation}
 Here $\operatorname{clip}_{[0,1]}$ restricts a real number to that interval, and the zero-width channel is handled separately by exact inversion. The relative syndrome rotation with first column $\mathbf v_{p_R}/\|\mathbf v_{p_R}\|_2$ attains $F_R$ for the least-favorable channel. Projection onto a convex segment gives $\mathbf v_{p_R}\cdot(\mathbf v_p-\mathbf v_{p_R})\ge0$, so the same recovery has fidelity at least $F_R$ for every $p$. Conversely, every recovery is bounded by the oracle optimum at the admissible $p_R$. Equation~\eqref{eq:benchmark-saddle} is therefore the exact all-CPTP minimax value. At an interior minimizer the payoff is constant across the family.

 The same coordinates specify the nominal decoder without a separate optimization. Choose its relative syndrome rotation to have first column $\mathbf v_{p_0}/\|\mathbf v_{p_0}\|_2$. Its fidelity on the actual channel is
\begin{equation}
 F_{e,\mathrm{nom}}(p)
 =\frac12\left(1+\frac{\mathbf v_{p_0}\cdot \mathbf v_p}{\|\mathbf v_{p_0}\|_2}\right).
 \label{eq:benchmark-nominal-finite}
\end{equation}
 Indeed its recovered coherence is the dot product of this unit vector with $\mathbf v_p$. This expression evaluates the same physical recovery throughout the uncertain family and is the finite-width nominal score used below.

 The least-favorable prior varies with the noise width and first reaches the endpoint $p_R=1$ at $\sigma_c=\arccos z_c=0.4270649725077595\ldots$, where $z_c$ is the unique root in $(9/10,23/25)$ of the polynomial, irreducible over $\mathbb Q$, $16z^6+32z^5-32z^3-14z^2+4z+3$. The next subsection derives this first switch; the inverse-cosine characterization specifies the threshold exactly. At a singular endpoint, the polar map leaves null directions undetermined. The segment-projection prescription fixes those directions so that the same recovery retains its guarantee for every $p$.

The weak-noise limit makes the tradeoff explicit. Symmetry removes the linear term in the noise polynomial, and $\mathbb E[x^2]=p$, so the oracle cancels $x^2-p$. A nominal recovery instead cancels $x^2-p_0$, leaving both the fluctuation around the true mean and the error in the assumed mean. The robust choice controls that residual over all $p\in[0,1]$. With $m=2$, $\mathcal V=1$, $g_1^2=2$ and $\chi=1/2$, the three infidelities are
 \begin{align}
 1-F_{e,\mathrm{oracle}}(p)
 &=\tfrac12p(1-p)\sigma^4+o(\sigma^4),\notag \\
 1-F_R&=\tfrac18\sigma^4+O(\sigma^6),  \label{eq:benchmark-small}
\\
 1-F_{e,\mathrm{nom}}(p)
 &=\tfrac12\big[(1-p)p_0^2+p(1-p_0)^2\big] \notag \\
 &\qquad\times\sigma^4+O(\sigma^6). \notag
 \end{align}
 At $p_0=1/10$, the nominal and robust worst-case coefficients are $81/200$ and $1/8$, respectively. Their ratio tends to $81/25$. At the nominal channel itself, the nominal coefficient is $9/200$, so robustness has a cost near the assumed model. A nominal design using the saddle-point recovery at the least-favorable prior is already robust. At $p=1$, the vanishing quartic oracle coefficient is followed by $16\sigma^6/9+O(\sigma^8)$; binary noise is not exactly corrected at sufficiently small nonzero width in this optical example.

 Independent Gaussian broadening of the three peaks retains an exactly solvable affine family. Figure~\ref{fig:benchmark} compares its decoders alongside the atomic model; the Gaussian-broadening argument below gives the characteristic functions, determinant signs and numerical checks.

\subsection{The first adverse-prior switch and smooth noise}

As the angle width grows, the adverse prior can move from an interior mixture to the binary endpoint. We first locate that change, then test whether broadening the discrete angles preserves the recovery advantage.

\paragraph{The first change in the least-favorable prior.}

To derive the threshold characterized above, set $z_c=\cos\sigma_c$. The endpoint condition reduces to
\begin{equation}
2z_c^2+4z_c+3-16z_c^2(1-z_c)(1+z_c)^3=0.
\label{eq:benchmark-switch-root}
\end{equation}
 The isolated root in $(9/10,23/25)$ gives the stated exact threshold. The threshold $\sigma_c$ is transcendental and therefore has no rational $p/q$ expression. Indeed $w=e^{i\sigma_c}$ obeys $w^2-2z_cw+1=0$ and is algebraic, whereas Hermite--Lindemann makes $e^{i\sigma_c}$ transcendental if the nonzero $\sigma_c$ were algebraic~\cite{BakerWustholz2008}. The endpoint condition $r_{\rm bin}=r_{\rm bin}^2+4b_{\rm bin}^2$ reduces to Eq.~\eqref{eq:benchmark-switch-root} after excluding the zero-width solution. The polynomial is strictly increasing on $[9/10,1]$ and changes sign across the stated interval, so its isolated root gives the first switch.

\paragraph{Smooth noise and the same finite-width decision.}

 Independent Gaussian jitter preserves the affine family used to compare nominal and robust recovery. The strict improvement reported below concerns the nominal design $p_0=1/10$ at $\delta=1/10$ and $\sigma=1/5$; the finite-width construction and weak-noise coefficients have the parameter domains stated separately. Let $x=x_{\rm at}+\delta z_{\rm G}$, where $x_{\rm at}$ has distribution~\eqref{eq:three-angle-family}, $z_{\rm G}$ is an independent unit Gaussian and $\delta>0$ is fixed. Its characteristic function at the relevant arguments is
\begin{equation}
 \varphi_{p,\delta}(k\sigma)=e^{-\delta^2k^2\sigma^2/2}(1-p+p\cos k\sigma).
 \label{eq:smooth-char}
\end{equation}
 The coherence matrix is still affine in $p$. Put $c_k=\varphi_{p,\delta}(k\sigma)$
for $k=1,2,4$. Its entries are
\begin{equation}
\begin{gathered}
C_{p,\delta}=\begin{pmatrix}a_{p,\delta}&b_{p,\delta}\\-b_{p,\delta}&d_{p,\delta}\end{pmatrix},\\
\begin{aligned}
a_{p,\delta}&=(4c_1+4c_2+c_4)/9,\\
d_{p,\delta}&=(4c_1-2c_2-2c_4)/9,\\
b_{p,\delta}&=\sqrt2(c_1-2c_2+c_4)/9.
\end{aligned}
\end{gathered}
\label{eq:smooth-coherence-matrix}
\end{equation}

The planar segment construction for broadened noise requires the coherence determinant to remain nonnegative. This sign fixes the branch of the trace-norm formula used in the minimax argument. The bounds below prove it over the stated interval for $\delta=1/10$.

 For the broadened optical matrix with $\delta=1/10$ and $0<\sigma\le0.8$, the determinant is nonnegative throughout the family. To check this without a numerical sign test, at the endpoint $p=1$ put $c_k^{\rm sm}=e^{-\delta^2k^2\sigma^2/2}\cos(k\sigma)$. Then $c_2^{\rm sm},c_4^{\rm sm}\le c_1^{\rm sm}$, so the lower diagonal is nonnegative. The bounds $c_1^{\rm sm}\ge0.677824$, $c_2^{\rm sm}\ge-0.28$, $c_4^{\rm sm}\ge-1$ make $4c_1^{\rm sm}+4c_2^{\rm sm}+c_4^{\rm sm}>0$, proving the upper diagonal positive. At $p=0$ both diagonals are nonnegative directly. Affine interpolation preserves their signs and hence the nonnegativity of $\det C_{p,\delta}=a_{p,\delta}d_{p,\delta}+b_{p,\delta}^2$.

The segment-projection argument therefore gives the exact saddle between the broadened endpoint vectors. With $v_{\rm jit}=\delta^2$, the noise norm is $h_2=p(1-p)+4p v_{\rm jit}+2v_{\rm jit}^2$. For $v_{\rm jit}\le1/4$ its maximum occurs at $p=1/2+2v_{\rm jit}$, giving
\begin{equation}
 1-F_R=(1/8+v_{\rm jit}+3v_{\rm jit}^2)\sigma^4+o(\sigma^4).
 \label{eq:smooth-leading}
\end{equation}
 For the nominal decoder designed for the broadened family at $p_0=1/10$, whose residual is $x^2-p_0-v_{\rm jit}$, the worst coefficient is $81/200+2v_{\rm jit}+v_{\rm jit}^2$. The guarantee applies to the declared Gaussian-broadened mixture family, whose characteristic function fixes the required harmonics.

 At $\sigma=1/5$, the atomic family has nominal worst infidelity $6.52382\times10^{-4}$ and robust worst infidelity $2.25547\times10^{-4}$, a factor of $2.8924$. With $\delta=1/10$, the corresponding values are $6.84308\times10^{-4}$ and $2.46141\times10^{-4}$, a factor of $2.7801$. Figure~\ref{fig:benchmark} also shows the dependence on nominal calibration and noise width.

\section{Uniform optical expansion and constants}
\label{app:joint-details}
The finite-alphabet and contamination results of Propositions~\ref{prop:finite-alphabet} and~\ref{prop:contamination} motivate a simultaneous-limit question. For the four-photon memory, the optical trace-norm formula in Eq.~\eqref{eq:opt-polar-value}, with the binary coherence matrix of Eq.~\eqref{eq:benchmark-C} at $p=1$, has a positive sixth-order loss; the following expansion controls how contamination competes with it.
  \subsection{A uniform crossover out of binary noise}

\label{sec:joint-optical}

For an equally weighted binary prior $x=\pm1$, every run has the same squared rotation magnitude. The first harmful amplitude in this memory is proportional to $x^2$, so one calibrated coherent subtraction cancels it across the entire ensemble. The quartic loss vanishes, leaving a sixth-order term. Adding other magnitudes restores a fluctuating $x^2$ and hence quartic loss. We now determine how small this departure from binary support must be to retain the sixth-order behavior. A remainder bound uniform in the changing prior will control this simultaneous limit.

For the same optical code and generator, write $y_j=\mathbb E[x^j]$ and
$\mathcal E(\sigma,\nu)=1-F_{e,\mathrm{opt}}(\sigma,\nu)$.
The optical trace-norm formula gives an explicit bound in terms of two physical controls, an upper bound on the eighth moment and a positive lower bound on the second moment.
 Consider any family of symmetric rescaled priors with fixed constants $K_8\ge1$, $\eta>0$ satisfying
$y_8\le K_8$ and $y_2\ge\eta$. Uniformly over that family,
\begin{equation}
\begin{gathered}
\begin{aligned}
\mathcal E(\sigma,\nu)={}&\frac{y_4-y_2^2}{2}\sigma^4\\
&+\frac{75y_2y_4-11y_6}{36}\sigma^6+R_\nu,
\end{aligned}\\
|R_\nu|\le10K_8^2\sigma^8,\quad
0\le\sigma\le\min\left\{\frac14,\sqrt{\frac\eta{8K_8^2}}\right\}.
\end{gathered}
\label{eq:uniform-optical-sixth}
\end{equation}
 The distribution may vary with $\sigma$. Recovery is optimized over all CPTP maps separately for each prior; the remainder bound is uniform over the stated family. The uniform remainder is derived below, after its consequences for the changing prior. Equation~\eqref{eq:joint-minimax} below gives a common-decoder consequence when the available information is an upper bound on contamination.

\begin{figure*}[t]
\centering
\includegraphics[width=\textwidth]{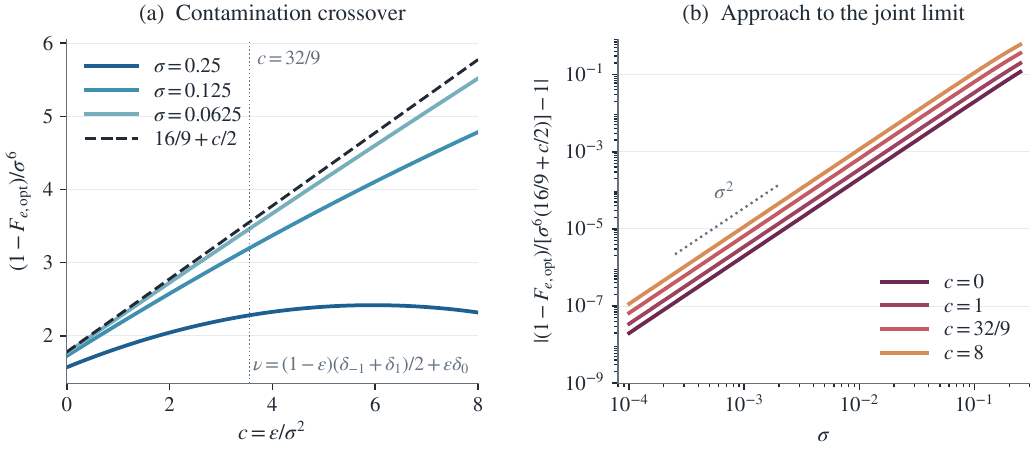}
\caption{Joint weak-noise and contamination limit for the optical code, with $\nu=(1-\epsilon)(\delta_{-1}+\delta_1)/2+\epsilon\delta_0$. (a) Exact prior-informed optimal infidelity divided by $\sigma^6$ along $\epsilon=c\sigma^2$. Recovery is selected separately for each specified prior. The dashed line is the proved limit $16/9+c/2$; its binary and contamination contributions are equal at $c=32/9$. (b) Relative discrepancy from this limit at fixed $c$. The uniform limit is analytical; finite-width curves illustrate its range of accuracy.}
\label{fig:contamination-crossover}
\end{figure*}

Let $C_\nu$ denote Eq.~\eqref{eq:opt-cross-matrix} for the angle $\tau=\sigma x$, with $x\sim\nu$. In this subsection let $b_\nu$ be the upper-right entry of $C_\nu$ and set $r_\nu=1-\Tr C_\nu$. The proof of Eq.~\eqref{eq:uniform-optical-sixth} expands its three cosine harmonics. The stated width range keeps its determinant nonnegative, so its trace norm is the Euclidean norm of $(1-r_\nu,2b_\nu)$. Rationalizing that norm controls the error uniformly as $h_2$ vanishes. Section~\ref{sec:optical-uniform-remainder} derives the explicit bound; no pointwise asymptotic formula is substituted into a varying prior.

For symmetric noise, $h_2=y_4-y_2^2=\operatorname{Var}(x^2)$ measures variation of the squared angle magnitude. Small $h_2$ therefore means that $x^2$ is nearly constant in mean square, even if the sign remains unknown. To compare the sixth-order coefficient with its binary value $16y_2^3/9$, write
\begin{equation}
\begin{split}
75y_2y_4-11y_6={}&64y_2^3+64y_2h_2\\
&-11\operatorname{Cov}(x^2,x^4),
\end{split}
\label{eq:optical-covariance-identity}
\end{equation}
 and apply the Cauchy--Schwarz bound
\begin{equation}
|\operatorname{Cov}(x^2,x^4)|\le\sqrt{K_8h_2}.
\end{equation}
Consequently,
\begin{equation}
\begin{split}
&\left|\mathcal E-\frac{h_2}{2}\sigma^4-\frac{16}{9}y_2^3\sigma^6\right|\\
&\qquad\le\frac{64K_8h_2+11\sqrt{K_8h_2}}{36}\sigma^6+10K_8^2\sigma^8.
\end{split}
\label{eq:joint-general-bound}
\end{equation}
 This error is uniformly small relative to the two retained terms as
$(h_2,\sigma)\to(0,0)$, because $y_2\ge\eta$.
In particular, let $y_{2,0}$ be the positive limiting second moment. Then
\begin{equation}
\begin{gathered}
y_2\to y_{2,0}>0,\qquad h_2/\sigma^2\to c<\infty\\
\Longrightarrow\quad
\mathcal E/\sigma^6\to c/2+16y_{2,0}^3/9.
\end{gathered}
\label{eq:joint-general-limit}
\end{equation}
  If instead $h_2\to0$ and $h_2/\sigma^2\to\infty$, then
$\mathcal E/(h_2\sigma^4)\to1/2$.
The lower second-moment bound excludes a simultaneous collapse of the
entire rescaled angle distribution to zero.

For the three-angle family, put $p=1-\epsilon$ and abbreviate $\mathcal E(\sigma,p)=\mathcal E(\sigma,\nu_p)$.
The constants sharpen to
\begin{equation}
\begin{gathered}
\begin{aligned}
\mathcal E(\sigma,1-\epsilon)={}&\frac{\epsilon(1-\epsilon)}2\sigma^4\\
&+\frac{(1-\epsilon)(64-75\epsilon)}{36}\sigma^6+R,
\end{aligned}\\
|R|\le6\sigma^8,\qquad 0\le\epsilon\le1,\quad0\le\sigma\le1/4.
\end{gathered}
\label{eq:joint-three-angle}
\end{equation}
The bound is uniform even at the deterministic endpoint $p=0$. Section~\ref{sec:optical-uniform-remainder} proves it from the same finite-width expression.

Let $\mathcal E_{\rm lead}=\epsilon\sigma^4/2+16\sigma^6/9$.
Equation~\eqref{eq:joint-three-angle} implies the relative bound
\begin{equation}
 \frac{|\mathcal E-\mathcal E_{\rm lead}|}{\mathcal E_{\rm lead}}
 \le\frac{139}{64}\epsilon+\frac{27}{8}\sigma^2
 \qquad(\sigma>0).
\label{eq:joint-relative-bound}
\end{equation}
Thus every path $\epsilon/\sigma^2\to c$ has
$\mathcal E/\sigma^6\to16/9+c/2$.
The quartic and sextic terms become comparable at
$\epsilon=(32/9)\sigma^2$.
Zero-angle events break the binary identity $x^2=1$ and restore
quartic error, even though they reduce the angle variance to
$1-\epsilon$. Retaining the binary advantage therefore requires control
of the noise alphabet as well as its variance. This crossover has an exact partial-information interpretation.
Suppose only $0\le\epsilon\le\bar\epsilon\le1/4$ is known, and write $\mathcal N_{\sigma,\epsilon}=\mathcal N_{\sigma,\nu_{1-\epsilon}}$ within this family.
For $0<\sigma\le1/4$,
use the binary quantities $r_{\rm bin},b_{\rm bin}$ defined in Section~\ref{sec:benchmark} at $p=1$, and write $D_{\rm bin}=r_{\rm bin}^2+4b_{\rm bin}^2$.
Section~\ref{sec:optical-uniform-remainder} bounds the stationary point by $r_{\rm bin}/D_{\rm bin}\le12/17<3/4$.
The recovery vector is $\mathbf v_p=(1-pr_{\rm bin},2pb_{\rm bin})$.
Its unconstrained norm minimum lies below $p=3/4$, while the admissible
segment is $p\in[1-\bar\epsilon,1]\subset[3/4,1]$.
Its nearest point is therefore the endpoint $p=1-\bar\epsilon$.
The segment-projection argument used for the benchmark proves
\begin{equation}
 \max_{\mathcal R}\inf_{0\le\epsilon\le\bar\epsilon}
 F_e(\mathcal R\mathcal N_{\sigma,\epsilon})
 =F_{e,\mathrm{opt}}(\sigma,\nu_{1-\bar\epsilon}).
\label{eq:joint-minimax}
\end{equation}
The endpoint's polar decoder is common to the whole uncertainty class.
Hence the same uniform crossover holds with $\epsilon$ replaced by the
known upper bound $\bar\epsilon$.

The effect persists for smooth noise.
For a fixed symmetric contaminant $\nu_c$ with finite eighth moment,
let $y_j^{(c)}=\int x^j d\nu_c$ and
$C_*=\int(x^2-1)^2d\nu_c$.
Mixing it into the binary prior with weight $\epsilon$ gives exactly
$h_2=\epsilon C_*-\epsilon^2(y_2^{(c)}-1)^2$.
Equation~\eqref{eq:uniform-optical-sixth} therefore gives
\begin{equation}
\begin{split}
\mathcal E={}&\frac{C_*\epsilon}{2}\sigma^4+\frac{16}{9}\sigma^6\\
&+O(\epsilon^2\sigma^4+\epsilon\sigma^6+\sigma^8).
\end{split}
\end{equation}
 uniformly for $\epsilon$ in a fixed sufficiently small interval.
The coefficient $C_*$ is the Christoffel-weighted residual for the binary
node polynomial $x^2-1$.
Alternatively, broaden both peaks by setting $x=x_{\rm at}+\delta z_{\rm G}$, where
$x_{\rm at}=\pm1$ equally likely and $z_{\rm G}$ is an independent unit Gaussian.
With $v_{\rm jit}=\delta^2$, its moments are
$y_2=1+v_{\rm jit}$, $y_4=1+6v_{\rm jit}+3v_{\rm jit}^2$ and
$y_6=1+15v_{\rm jit}+45v_{\rm jit}^2+15v_{\rm jit}^3$.
Equation~\eqref{eq:uniform-optical-sixth} becomes
\begin{equation}
\begin{split}
\mathcal E={}&(2v_{\rm jit}+v_{\rm jit}^2)\sigma^4\\
&+\left(\frac{16}{9}+10v_{\rm jit}+5v_{\rm jit}^2+\frac53v_{\rm jit}^3\right)\sigma^6\\
&+O(\sigma^8).
\end{split}
\label{eq:joint-smooth}
\end{equation}
 uniformly on any fixed bounded interval of $v_{\rm jit}$.
Thus $\delta^2/\sigma^2\to c$ gives
$\mathcal E/\sigma^6\to16/9+2c$.
The physical angular broadening is $\sigma\delta$. Broadening of order $\sigma^2$ preserves sixth-order suppression with a changed coefficient; broadening $o(\sigma^2)$ retains the binary coefficient $16/9$. Figure~\ref{fig:contamination-crossover} shows how the finite-width curves approach the controlled limit. These conclusions close the joint limit for this optical memory.
An arbitrary-code joint-limit result requires the corresponding
terminal higher-order recovery coefficient.

\subsection{Uniform remainder and endpoint control}
\label{sec:optical-uniform-remainder}

The binary-noise crossover lets the prior vary while the angle width shrinks. We therefore bound the cosine and square-root remainders uniformly over the declared moment class. The same estimates locate the least-favorable endpoint when contamination is known only through an upper bound.

To prove Eq.~\eqref{eq:uniform-optical-sixth}, write the real coherence matrix for the symmetric prior as
\begin{equation}
C_\nu=\begin{pmatrix}a_\nu&b_\nu\\-b_\nu&d_\nu\end{pmatrix},\quad r_\nu=1-a_\nu-d_\nu.
\end{equation}
Cosine remainders give
\begin{align}
|a_\nu-1+2y_2\sigma^2|&\le\tfrac32K_8\sigma^4, \notag \\
|d_\nu-2y_2\sigma^2|&\le3K_8\sigma^4,\\
r_\nu&=y_4\sigma^4-\tfrac{11}{18}y_6\sigma^6+\rho_r,\notag \\
b_\nu&=-\tfrac{y_2}{\sqrt2}\sigma^2+\tfrac{25\sqrt2y_4}{24}\sigma^4+\rho_b,\notag \\
|\rho_r|&\le\tfrac{K_8}{5}\sigma^8,\notag \\
|\rho_b|&\le K_8\sigma^6. \notag
\end{align}
Here the last two remainder constants follow respectively from
$(8+2\cdot2^8+4^8)/(9\cdot8!)<1/5$ and
$\sqrt2(1+2\cdot2^6+4^6)/(9\cdot6!)<1$.
The declared width range makes $a_\nu>1/2$ and
$d_\nu\ge13\eta\sigma^2/8>0$ for $\sigma>0$.
Thus $\det C_\nu=a_\nu d_\nu+b_\nu^2>0$, so the exact trace-norm formula gives
\begin{equation}
\begin{gathered}
\mathcal E=\frac{1-\sqrt{1-w}}2,\quad w=2r_\nu-r_\nu^2-4b_\nu^2,\\
0\le\mathcal E-\frac w4
=\frac{w^2}{4(1+\sqrt{1-w})^2}\le\frac{w^2}4.
\end{gathered}
\label{eq:joint-rationalization}
\end{equation}
 In particular $0\le w\le1$.
Squaring the expansion of $b_\nu$, with $\sigma\le1/4$, bounds its remainder
through sixth order by $16K_8^2\sigma^8$ in $4b_\nu^2$.
Also $|r_\nu|\le2K_8\sigma^4$.
The remainder of $w$ is therefore at most
$(2K_8/5+20K_8^2)\sigma^8$, and
$0\le w\le2r_\nu\le4K_8\sigma^4$ bounds the final square-root correction
by $4K_8^2\sigma^8$. Their sum after division by four is below
$10K_8^2\sigma^8$, proving Eq.~\eqref{eq:uniform-optical-sixth}.
These conservative bounds require no convergent full moment expansion.

For completeness, let $r_{\rm bin},b_{\rm bin}$ refer to $p=1$.
Their preceding cosine bounds give $0\le r_{\rm bin}\le\sigma^4$ and
\begin{equation}
 \left|4b_{\rm bin}^2-2\sigma^4+\frac{25}{3}\sigma^6\right|<16\sigma^8.
\end{equation}
The mixture has $w_p=2pr_{\rm bin}-p^2(r_{\rm bin}^2+4b_{\rm bin}^2)$.
Its remainder through degree six is at most $87\sigma^8/5$, uniformly
in $p$, and $0\le w_p\le2\sigma^4$.
Equation~\eqref{eq:joint-rationalization} bounds the remaining correction
by $\sigma^8$, giving $(87/20+1)\sigma^8<6\sigma^8$.

\paragraph{The least-favorable contamination endpoint.}
For the binary quantities of Section~\ref{sec:joint-optical}, put $D_{\rm bin}=r_{\rm bin}^2+4b_{\rm bin}^2$. The preceding estimates and $0<\sigma\le1/4$ give
\begin{equation}
\begin{gathered}
D_{\rm bin}\ge\sigma^4\left(2-\frac{25}{3}\sigma^2-16\sigma^4\right)
\ge\frac{17}{12}\sigma^4,\\
r_{\rm bin}/D_{\rm bin}\le12/17<3/4.
\end{gathered}
\end{equation}
Here $r_{\rm bin}\le\sigma^4$ was used for the ratio. This places the stationary norm minimum below the allowed segment $p\in[1-\bar\epsilon,1]\subset[3/4,1]$, proving the endpoint assertion used in Eq.~\eqref{eq:joint-minimax}.

\section{Computational methods and checks}
\label{app:computational-checks}
Exact polynomial arithmetic and symbolic identities were checked with SymPy 1.13.1~\cite{MeurerEtAl2017}; high-precision matrix and crossover evaluations used mpmath 1.3.0~\cite{Mpmath2023}, and the figures were drawn with Matplotlib 3.10.0~\cite{Hunter2007}. These computations test the finite matrices, polynomial identities and remainder bounds specified by the proofs; the general statements rest on the proofs.

The Petz coefficient of Proposition~\ref{thm:petz} was compared with exact Gram square roots in the five finite models of Appendix~\ref{app:petz-numerics} at 80-digit precision; extrapolation in the width reproduced the predicted coefficients with relative error below $2\times10^{-6}$. The calibration-frontier curves of Fig.~\ref{fig:optical-calibration-frontier} use 100-digit evaluations of the exact formulas. At seven widths from $10^{-5}$ to $10^{-1}$ they agree with an independently formulated computation of the same guarantees to relative precision below $10^{-79}$. The finite-width curves of Fig.~\ref{fig:contamination-crossover} evaluate the exact two-dimensional coherence formula at 100 decimal digits and agree with retained 160-digit physical-Hamiltonian computations.

The benchmark scores of Appendix~\ref{app:benchmark-details} were checked with explicit Kraus maps in the five-dimensional optical space; across 180 decoder/channel checks the largest score discrepancy was $1.34\times10^{-15}$. Six robust and eighteen oracle recovery problems were also solved as Choi semidefinite programs~\cite{FletcherShorWin2007} with CVXPY 1.9.2 and Clarabel. Their largest objective discrepancy from the analytical answer was $4.57\times10^{-8}$, with every run returning the reduced-accuracy status \texttt{optimal\_inaccurate}. A second comparison tested thirteen unrestricted recovery problems, including singular and noncommuting boundaries, optical cases, a known unitary and complete depolarization. It used CVXPY 1.9.3~\cite{DiamondBoyd2016,AgrawalEtAl2018}, with each problem evaluated by Clarabel 0.11.1~\cite{GoulartChen2024} and SCS 3.3.1~\cite{ODonoghueEtAl2016}. Repaired primal--dual gaps were at most $6.54\times10^{-8}$; SCS returned twelve \texttt{optimal} and one \texttt{optimal\_inaccurate} statuses, and Clarabel one and twelve, respectively.

The worst-input bracket of Proposition~\ref{prop:worst-input} was checked on the three-shell model of Eq.~\eqref{eq:three-shell-teaching-model} with logical dimensions two, three and four, for symmetric and asymmetric priors at widths down to $1/40$. The diamond norm was computed from its semidefinite program with CVXPY 1.9.2 and Clarabel 0.11.1, using NumPy 2.2.2~\cite{HarrisEtAl2020} and SciPy 1.18.1~\cite{VirtanenEtAl2020}. The corrected ratio of worst-input error to entanglement infidelity was $1.000$, $1.41$--$1.44$ and $1.57$ for $d=2,3,4$, while the uncorrected asymmetric error grew as $\sigma^{2m-1}$. The constants $\mathfrak h_{m,r}$ were computed from the dual in Eq.~\eqref{eq:calibration-mismatch-dual} and, independently, by maximizing over symmetric atomic priors; the two agree within $4\times10^{-11}$ for $(m,r)=(3,4),(3,5),(4,5),(4,6)$, and the dual reproduces Eq.~\eqref{eq:calibration-quartic-factor} at $m=2$. The exact value of $\mathfrak h_{3,4}$, its polynomial identity and the orthogonality of the stated prior were verified symbolically. The boundary Lanczos identity stated after Eq.~\eqref{eq:physical-defect} was checked for $d=3$, $L=2$ in a random physical frame; over 200 random logical inputs the coefficients agreed with the formula to $7\times10^{-15}$. These finite checks corroborate the analytical statements at the tested parameters.

\bibliographystyle{quantum}
\bibliography{references}
\onecolumngrid
\end{document}